\def\isReadyToSubmit{0}   
\def\isCameraReady{1}     
\def\isUsenix{0}          
\def\isACM{1}

\ifnum\isUsenix=1
    \documentclass[10pt,twocolumn]{article}
\else
    \ifnum\isACM=1
        \documentclass[sigplan,10pt]{acmart}

        \ifnum\isCameraReady=0
            \renewcommand\footnotetextcopyrightpermission[1]{}
        \else
            
            \copyrightyear{2023}
            \acmYear{2023}
            \setcopyright{rightsretained}

            \acmConference[SOSP '23]{ACM SIGOPS 29th Symposium on Operating Systems Principles}{October 23--26, 2023}{Koblenz, Germany}
            \acmDOI{10.1145/3600006.3613174}
            \acmISBN{979-8-4007-0229-7/23/10}
        \fi

    \else
        \documentclass[conference,compsoc]{IEEEtran}
        \newcommand{\subparagraph}{}
    \fi
\fi

\ifnum\isACM=0
    \usepackage[top=1in, bottom=1in, left=1in, right=1in]{geometry}
    \usepackage[hyphens]{url}
    \usepackage[usenames,dvipsnames,svgnames,table]{xcolor}
    \usepackage{amsfonts}
    \usepackage{amsmath,amssymb,amsthm}
\fi

\usepackage{times}
\usepackage{graphicx}
\usepackage{subfigure}
\usepackage{array} 
\usepackage{color}
\usepackage{listings} 
\usepackage{xspace} 
\usepackage[noend]{algpseudocode} 
\usepackage{algorithm} 
\usepackage{wrapfig} 
\usepackage{flushend} 
\usepackage{multirow} 
\usepackage{booktabs} 
\usepackage{anyfontsize} 
\usepackage{paralist} 
\usepackage[font=footnotesize]{caption}
\usepackage{courier} 
\usepackage{verbatim}
\usepackage{calc}
\usepackage{amsthm}
\usepackage{pbox}
\usepackage{dsfont} 
\usepackage{soul} 

\usepackage[compact,small]{titlesec}

\usepackage{style/macros}

\titlespacing{\section}{0pt}{*.4}{*.1}
\titlespacing{\subsection}{0pt}{*.4}{*.1}
\titlespacing{\subsubsection}{0pt}{*.4}{*.1}

\usepackage{etoolbox}
\makeatletter
\patchcmd{\ttlh@hang}{\parindent\z@}{\parindent\z@\leavevmode}{}{}
\patchcmd{\ttlh@hang}{\noindent}{}{}{}
\makeatother

\let\oldenumerate\enumerate
\renewcommand{\enumerate}{
    \oldenumerate
    \setlength{\itemsep}{.5pt}
    \setlength{\parskip}{0pt}
    \setlength{\parsep}{0pt}
}

\let\olditemize\itemize
\renewcommand{\itemize}{
    \olditemize
    \setlength{\itemsep}{1pt}
    \setlength{\parskip}{0pt}
    \setlength{\parsep}{0pt}
}

\makeatletter
\renewcommand{\ALG@beginalgorithmic}{\footnotesize}
\makeatother

\newcommand{\ie}{{\em i.e.,~}}
\newcommand{\eg}{{\em e.g.,~}}

\def\F{Fig.~}

\def\Thm{Thm.~}
\def\A{Alg.~}

\newcommand{\ar}[3]{} 
\include{latexdefs}

\newcommand{\heading}[1]{{\vspace{0pt}\noindent\bf{#1}}} 

{\makeatletter
\gdef\xxxmark{%
    \expandafter\ifx\csname @mpargs\endcsname\relax 
        \expandafter\ifx\csname @captype\endcsname\relax 
            \marginpar{\textcolor{red}{xxx~}}
        \else
            \textcolor{red}{xxx~}
        \fi
    \else
        \textcolor{red}{xxx~}
    \fi}
\gdef\xxx{\@ifnextchar[\xxx@lab\xxx@nolab}
\long\gdef\xxx@lab[#1]#2{{\bf [\xxxmark \textcolor{red}{#2} ---{\sc #1}]}}
\long\gdef\xxx@nolab#1{{\bf [\xxxmark \textcolor{red}{#1}]}}
\ifnum\isReadyToSubmit=1
    \long\gdef\xxx@lab[#1]#2{}\long\gdef\xxx@nolab#1{}
\fi
}

{\makeatletter
\gdef\edit{\@ifnextchar[\edit@lab\edit@nolab}
\long\gdef\edit@lab[#1]#2{[\textcolor{red}{#2} ---{\sc #1}]}
\long\gdef\edit@nolab#1{[\textcolor{red}{#1}]}
\ifnum\isReadyToSubmit=1
    \long\gdef\edit@lab[#1]#2{[#2]}
\fi
}

\newcommand{\ignore}[1]{}

\definecolor{grey}{rgb}{0.5,0.5,0.5}

\algdef{SE}[VARIABLES]{Variables}{EndVariables}
{\algorithmicvariables}{}
\algnewcommand{\algorithmicvariables}{\textbf{global variables}}

\algdef{SE}{Begin}{End}{\textbf{begin}}{\textbf{end}}

\ifnum\isUsenix=1
    \usepackage{style/usenix-2020-09}
\fi

\usepackage{lipsum} 

\usepackage{cancel}

\usepackage{caption}
\DeclareCaptionFont{mysize}{\footnotesize} 
\def\ie{{i.e.}\xspace}
\def\eg{{e.g.}\xspace}

\def\whp{w.h.p.\xspace}

\def\sysname{Turbo\xspace}

\def\pmw{PMW\xspace}

\def\pmwbypass{PMW-Bypass\xspace}

\def\pmwbypasstree{tree-structured PMW-Bypass\xspace}
\def\Pmwbypasstree{Tree-structured PMW-Bypass\xspace}
\def\exactcache{Exact-Cache\xspace}

\def\exactcachetree{tree-structured \exactcache\xspace}

\def\nonpartitioneddb{non-partitioned database\xspace}
\def\Nonpartitioneddb{Non-partitioned database\xspace}

\def\Partitionedstaticdb{Partitioned static database\xspace}
\def\partitionedstreamingdb{partitioned streaming database\xspace}
\def\Partitionedstreamingdb{Partitioned streaming database\xspace}

\def\nonpartitioneddbimprovement{$1.7-15.9\times$\xspace}  
\def\partitionedstaticdbimprovement{$1.9-4.7\times$\xspace}  
\def\partitionedstreamingdbimprovement{$1.9-5.4\times$\xspace}  

\def\theoreticalconvergence{worst-case convergence\xspace}
\def\Theoreticalconvergence{Worst-case convergence\xspace}

\def\empiricalconvergence{empirical convergence\xspace}
\def\Empiricalconvergence{Empirical convergence\xspace}

\newcommand{\lap}[1]{\operatorname{Lap}(#1)}
\def\kzipf{k_{\text{zipf}}}
\def\betaconv{\rho}
\def\lr{\text{\em lr}}
\def\E{\mathbb{E}}

\def\covid19{Covid\xspace}
\def\citibike{CitiBike\xspace}

\def\abccorp{ABC\xspace}

\begin{document}

\date{\today}

\title{\LARGE Turbo: Effective Caching in Differentially-Private Databases}


\author{Kelly Kostopoulou*}
\affiliation{%
  \institution{Columbia University}
  \city{}
  \country{}}
\email{kelkost@cs.columbia.edu}

\author{Pierre Tholoniat*}
\affiliation{%
  \institution{Columbia University}
  \city{}
  \country{}}
\email{pierre@cs.columbia.edu}

\author{Asaf Cidon}
\affiliation{%
  \institution{Columbia University}
  \city{}
  \country{}}
\email{asaf.cidon@columbia.edu}

\author{Roxana Geambasu}
\affiliation{%
  \institution{Columbia University}
  \city{}
  \country{}}
\email{roxana@cs.columbia.edu}

\author{Mathias L\'ecuyer}
\affiliation{%
  \institution{University of British Columbia}
  \city{}
  \country{}}
\email{mathias.lecuyer@ubc.ca}


\begin{abstract}

Differentially-private (DP) databases allow for privacy-preserving analytics over sensitive datasets or data streams.  In these systems, {\em user privacy} is a limited resource that must be conserved with each query.  We propose {\em \sysname}, a novel, state-of-the-art caching layer for linear query workloads over DP databases.  \sysname builds upon private multiplicative weights (\pmw), a DP mechanism that is powerful in theory but ineffective in practice, and transforms it into a highly-effective caching mechanism, {\em \pmwbypass}, that uses prior query results obtained through an external DP mechanism to train a \pmw to answer arbitrary future linear queries accurately and ``for free'' from a privacy perspective.  Our experiments on public \covid19 and CitiBike datasets show that \sysname with \pmwbypass conserves \nonpartitioneddbimprovement more budget compared to vanilla PMW and simpler cache designs, a significant improvement.  Moreover, \sysname provides support for range query workloads, such as timeseries or streams, where opportunities exist to further conserve privacy budget through DP parallel composition and warm-starting of \pmw state.  Our work provides a theoretical foundation and general system design for effective caching in DP databases.

\end{abstract}

\settopmatter{printfolios=true}
\maketitle

\def\thefootnote{*}\footnotetext{These authors contributed equally to this work.\\
  Extended version of the SOSP' 23 paper.
}\def\thefootnote{\arabic{footnote}}

\ifnum\isCameraReady=1
  \pagestyle{plain} 
\fi



\section{Introduction}
\label{sec:introduction}

\abccorp collects lots of user data from its digital products to analyze trends, improve existing products, and develop new ones. To protect user privacy, the company uses a restricted interface that removes personally identifiable information and only allows queries over aggregated data from multiple users. Internal analysts use interactive tools like Tableau to examine static datasets and run jobs to calculate aggregate metrics over data streams. Some of these metrics are shared with external partners for product integrations. However, due to data reconstruction attacks on similar ``anonymized'' and ``aggregated'' data from other sources, including the US Census Bureau~\cite{garfinkel2019understanding} and Aircloak~\cite{cohen2020linear}, the CEO has decided to pause external aggregate releases and severely limit the number of analysts with access to user data statistics until the company can find a more rigorous privacy solution.

The preceding scenario, while fictitious, is representative of what often occurs in industry and government, leading to obstacles to data analysis or incomplete privacy solutions. In 2007, Netflix withdrew ``anonymized'' movie rating data and canceled a competition due to de-anonymization attacks~\cite{narayanan2008robust}. In 2008, genotyping aggregate information from a clinical study led to the revelation of participants' membership in the diagnosed group, prompting the National Institutes of Health to advise against the public release of statistics from clinical studies~\cite{nih_genomic_data_sharing_policy}. In 2021, New York City excluded demographic information from datasets released from their CitiBike bike rental service, which could reveal sensitive user data~\cite{citibikeData}. The city's new, more restrained data release not only remains susceptible to privacy attacks but also prevents analyses of how demographic groups use the service.

Differential privacy (DP) provides a rigorous solution to the problem of protecting user privacy while analyzing and sharing statistical aggregates over a database. DP guarantees that analysts cannot confidently learn anything about any individual in the database that they could not learn if the individual were not in the database. Industry and government have started to deploy DP for various use cases~\cite{dp_use_cases}, including publishing trends in Google searches related to \covid19~\cite{Bavadekar2020Sep}, sharing LinkedIn user engagement statistics with outside marketers~\cite{rogers2020linkedin}, enabling analyst access to Uber mobility data while protecting against insider attacks~\cite{chorus}, and releasing the US Census' 2020 redistricting data~\cite{censusTopDown}. To facilitate the application of DP, industry has developed a suite of systems, ranging from specialized designs like the US Census TopDown~\cite{censusTopDown} and LinkedIn Audience Engagements~\cite{rogers2020linkedin} to more general DP SQL systems, like GoogleDP~\cite{zetasql}, Uber Chorus~\cite{chorus}, and Tumult Analytics~\cite{tumult}.  

DP systems face a significant challenge that hinders their wider adoption: they struggle to handle large workloads of queries while maintaining a reasonable privacy guarantee. This is known as the ``running out of privacy budget'' problem and affects any system, whether DP or not, that aims to release multiple statistics from a sensitive dataset. A seminal paper by Dinur and Nissim~\cite{dinur2003revealing} proved that releasing too many accurate linear statistics from a dataset fundamentally enables its reconstruction, setting a lower bound on the necessary error in queries to prevent such reconstruction. Successful reconstructions of the US Census 2010 data~\cite{garfinkel2019understanding} and Aircloak's data~\cite{cohen2020linear} from the aggregate statistics released by these entities exemplify this fundamental limitation. DP, while not immune to this limitation, provides a means of bounding the reconstruction risk. DP randomizes the output of a query to limit the influence of individual entries in the dataset on the result. Each new DP query increases this limit, consuming part of a {\em global privacy budget} that must not be exceeded, lest individual entries become vulnerable to reconstruction.

Recent work proposed treating the global privacy budget as a {\em system resource} that must be managed and conserved, similar to traditional resources like CPU~\cite{privatekube}.
When computation is expensive, {\em caching} is a go-to solution: it uses past results to save CPU on future computations.
Caches are ubiquitous in all computing systems -- from the processor to operating systems and databases -- enabling scaling to much larger workloads than would otherwise be afforded with fixed resources.
In this paper, we thus ask: {\em How should caching work in DP systems to significantly increase the number of queries they can support under a privacy guarantee?}  While DP theory has explored algorithms to reuse past query results to save privacy budget in future queries, there is no general DP caching system that is effective in common practical settings.

We propose {\em \sysname}, the first general and effective caching layer for DP SQL databases that boosts the number of linear queries (such as sums, averages, counts) that can be answered accurately under a fixed, global DP guarantee.  In addition to incorporating a traditional {\em exact-match cache} that saves past DP query results and reuses them if the same query reappears, \sysname builds upon a powerful theoretical construct, known as {\em private multiplicative weights (PMW)}~\cite{pmw}, that leverages past DP query results to learn a histogram representation of the dataset that can go on to answer {\em arbitrary} future linear queries for free once it has converged.
While PMW has compelling convergence guarantees in theory, we find it ineffective in practice, being overrun even by an exact-match cache.

We make three main contributions to PMW design to boost its effectiveness and applicability.
First, we develop {\em \pmwbypass}, a variant of PMW that bypasses it during the privacy-expensive learning phase of its histogram, and switches to it once it has converged to reap its free-query benefits.
This change requires a new mechanism for updating the histogram despite bypassing the PMW, plus new theory to justify its convergence.  The \pmwbypass technique is highly effective, significantly outperforming both the exact-match cache and vanilla PMW in the number of queries it can support.  Second, we optimize our mechanisms for workloads of range queries that do not access the entire database. These types of queries are typical in timeseries databases and data streams.
For such workloads, we organize the cache as a tree of multiple \pmwbypass objects and demonstrate that this approach outperforms alternative designs.  Third, for streaming workloads, we develop warm-starting procedures for tree-structured \pmwbypass histograms, resulting in faster convergence. 

We formally analyze each of our techniques, focusing on privacy, per-query accuracy, and convergence speed.
Each technique represents a contribution on its own and can be used separately, or, as we do in \sysname, as part of the first {\em general, effective, and accurate DP-SQL caching design}.  We prototype \sysname on TimescaleDB, a timeseries database, and use Redis to store caching state. We evaluate \sysname on workloads based on \covid19 and CitiBike datasets. We show that \sysname significantly improves the number of linear queries that can be answered with less than $5\%$ error (w.h.p.) under a global $(10, 0)$-DP guarantee, compared to not having a cache and alternative cache designs.  Our approach outperforms the best-performing baseline in each workload by 1.7 to 15.9 times, and even more significantly compared to vanilla PMW and systems with no cache at all (such as most existing DP systems).  These results demonstrate that our \sysname cache design is both general and effective in boosting workloads in DP SQL databases and streams, making it a promising solution for companies like \abccorp that seek an effective DP SQL system to address their user data analysis and sharing concerns. We make \sysname available open-source at \url{https://github.com/columbia/turbo}, part of a broader set of infrastructure systems we are developing for DP, all described here: \url{https://systems.cs.columbia.edu/dp-infrastructure/}.

\section{Background}
\label{sec:background}

\heading{Threat model.}
We consider a threat model known as {\em centralized differential privacy}: one or more untrusted analysts query a dataset or stream through a restricted, aggregate-only interface implemented by a trusted database engine of which \sysname is a trusted component.
The goal of the database and \sysname is to provide accurate answers to the analysts' queries without compromising the privacy of individual users in the database.
The two main adversarial goals that an analyst may have are membership inference and data reconstruction. Membership inference is when the adversary wants to determine whether a known data point is present in the dataset.  Data reconstruction involves reconstructing unknown data points from a known subset of the dataset. To achieve their goals, the adversary can use composition attacks to single out contributions from individuals, collude with other analysts to coordinate their queries, link anonymized records to public datasets, and access arbitrary auxiliary information {\em except for} timing side-channel information.
Previous research demonstrated attacks under this threat model~\cite{narayanan2008robust, de2013unique, ganta2008composition, cohen2020linear, garfinkel2019understanding, homer2008resolving}.

\heading{Differential privacy (DP).}
DP~\cite{Dwork:2006:CNS:2180286.2180305} randomizes aggregate queries over a dataset to prevent membership inference and data reconstruction~\cite{wasserman2010statistical,dong2022gaussian}. DP randomization (a.k.a. noise) ensures that the probability of observing a specific result is stable to a change in one datapoint (e.g., if user $x$ is removed or replaced in the dataset, the distribution over results remains similar).
More formally, a query $Q$ is $(\epsilon, \delta)$-DP if, for any two datasets $D$ and $D'$ that differ by one datapoint, and for any result subset $S$ we have: $\mathbb{P}(Q(D) \in S) \leq e^\epsilon \mathbb{P}(Q(D') \in S) + \delta$.
$\epsilon$ quantifies the privacy loss due to releasing the DP query's result (higher means less privacy), while $\delta$ can be interpreted as a failure probability and is set to a small value.

Two common mechanisms to enforce DP are the $\textrm{Laplace}$ and $\textrm{Gaussian}$ mechanisms.
They add noise from an appropriately scaled Laplace/Gaussian distribution to the true query result, and return the noisy result.  As an example, for counting queries and a database of size $n$, adding noise from $\textrm{Laplace}(0, 1/n\epsilon)$, ensures $(\epsilon, 0)$-DP (a.k.a. pure DP); adding noise from $\textrm{Gaussian}(0, \sqrt{2\ln(1.25/\delta)} / n\epsilon)$ ensures $(\epsilon, \delta)$-DP.
The accuracy for such queries can be controlled probabilistically by converting it into the $(\epsilon, \delta)$ parameters.

Answering multiple queries on the same data fundamentally degrades privacy~\cite{dinur2003revealing}.
DP quantifies this over a sequence of DP queries using the {\em composition property}, which in its basic form states that releasing two $(\epsilon_1, \delta_1)$-DP and $(\epsilon_2, \delta_2)$-DP queries is $(\epsilon_1+\epsilon_2, \delta_1+\delta_2)$-DP. When queries access disjoint data subsets, their composition is $(\max(\epsilon_1,\epsilon_2), \max(\delta_1,\delta_2))$-DP and is called {\em parallel composition}. Using composition, one can enforce a global $(\epsilon_G, \delta_G)$-DP guarantee over a workload, with each DP query ``consuming'' part of a {\em global privacy budget} that is defined upfront as a system parameter~\cite{rogers2016privacy}.

Good values of the global privacy budget in interactive DP SQL systems remain subject for debate~\cite{choosing_epsilon}, but generally, an ideal value for strong theoretical guarantees is $\epsilon_G=0.1$, while $\epsilon_G=1$ are considered acceptable. Larger values are often considered vacuous semantically, since individuals' privacy risk grows with $e^{\epsilon_G}$.  
In this paper, we aim to achieve values of $\epsilon_G=1$ or smaller over a query workload.

\heading{Private multiplicative weights (\pmw).}
\pmw is a DP mechanism to answer online linear queries with bounded error~\cite{pmw}.
We defer detailed description of PMW, plus an example illustrating its functioning, to \S\ref{sec:detailed-design} and only give here an overview.
PMW maintains an approximation of the dataset in the form of a {\em histogram}: estimated counts of how many times any possible data point appears in the dataset.
When a query arrives, PMW estimates an answer using the histogram and computes the {\em error of this estimate} against the real data in a DP way, using a DP mechanism called {\em sparse vector (SV)}~\cite{privacybook} (described shortly).
If the estimate's error is low, it is returned to the analyst, consuming no privacy budget (i.e., the query is answered ``for free'').
If the estimate's error is large, then \pmw executes the DP query on the data with the Laplace/Gaussian mechanism, consuming privacy budget as needed. It returns the DP result and also uses it to update the histogram for more accurate estimates to future queries.

An additional cost in using \pmw comes from the SV, a well-known DP mechanism that can be used to test the error of a sequence of query estimates against the ground truth with DP guarantees and limited privacy budget consumption~\cite{privacybook}.  We refer the reader to textbook descriptions of SV for detailed functioning~\cite{privacybook} and provide here only an overview of its semantics. SV is a stateful mechanism that receives queries and estimates for their results one by one, and assesses the error between these estimates and the ground-truth query results. While the estimates have error below a preset threshold with high probability, SV returns success and consumes {\em zero privacy}. However, as soon as SV detects a large-error estimate, it requires a {\em reset}, which is a privacy-expensive operation that re-initializes state within the SV to continue the assessments.
In common SV implementations, a reset costs as much as $3\times$ the privacy budget of executing one DP query on the data.

The theoretical vision of \pmw is as follows.
Under a stream of queries, \pmw first goes through a ``training'' phase, where its histogram is inaccurate, requiring frequent SV resets and consuming budget.
Failed estimation attempts update the histogram with low-error results obtained by running the DP query.
Once the histogram becomes sufficiently accurate, the SV tests consistently pass, thereby ameliorating the initial training cost.
Theoretical analyses provide a compelling {\em worst-case convergence} guarantee for the histogram, determining a worst-case number of updates required to train a histogram that can answer \emph{any future linear query} with low error~\cite{hardt2010multiplicative}.
However, no one has examined whether this worst-case bound is practical and if \pmw outperforms natural baselines, such as an exact-match cache.
\section{\sysname Overview}
\label{sec:overview}

\sysname is a caching layer that can be integrated into a DP SQL engine, significantly increasing the number of linear queries that can be executed under a fixed, global $(\epsilon_G,\delta_G)$-DP guarantee.
We focus on {\em linear queries} like sums, averages, and counts (defined in \S\ref{sec:detailed-design}), which are widely used in interactive analytics and constitute the class of queries supported by approximate databases such as BlinkDB~\cite{blinkdb}.
These queries enable powerful forms of caching like \pmw, and also allow for accuracy guarantees, which are important when doing approximate analytics, as one does on a DP database.

\subsection{Design Goals}
\label{sec:goals}

In designing \sysname, we were guided by several goals:
\begin{itemize}
    \item[{\bf (G1)}] {\em Guarantee privacy:} \sysname must satisfy ($\epsilon_G, \delta_G$)-DP.

    \item[{\bf (G2)}] {\em Guarantee accuracy:} \sysname must ensure {\em $\mathit{(\alpha, \beta)}$-accuracy} for each query, defined for $\alpha > 0$, $\beta \in (0,1)$ as follows: if $R'$ and $R$ are the returned and true results, then $|R' - R| \le \alpha$ with $(1-\beta)$ probability.  If $\beta$ is small, a result is {\em $\mathit{\alpha}$-accurate w.h.p.} (with high probability).

    \item[{\bf (G3)} and {\bf (G4)}] {\em Provide \theoreticalconvergence guarantees but optimize for \empiricalconvergence:} We aim to maintain \pmw's theoretical convergence (G3), but we prioritize for {\em \empiricalconvergence} speed, a new metric that measures, on a workload, the number of updates needed to answer most queries for free (G4).

    \item[{\bf (G5)}] {\em Improve privacy budget consumption:} We aim for {\em significant improvements} in privacy budget consumption compared to both not having a cache and having an exact-match cache or a vanilla \pmw.

    \item[{\bf (G6)}] {\em Support multiple use cases:}  \sysname should benefit multiple important workload types, including static and streaming databases, and queries that arrive over time.

    \item[{\bf (G7)}] {\em Easy to configure:} \sysname should include few knobs with fairly stable performance. 
\end{itemize}

(G1) and (G2) are strict requirements.
(G3) and (G4) are driven by our belief that DP systems should not only possess meaningful theoretical properties but also be optimized for practice.
(G5) is our main objective.
(G6) requires further attention, given shortly.
(G7) is driven by the limited guidance from \pmw literature on parameter tuning.
\pmw meets goals (G1-G3) but falls significantly short for (G4-G7).
\sysname achieves all goals; we provide theoretical analyses for (G1-G3) in \S\ref{sec:detailed-design} and empirical evaluations for (G4-G7) in \S\ref{sec:evaluation}.

\subsection{Use Cases}
\label{sec:use-cases}

The DP literature is fragmented, with different algorithms developed for different use cases.
We seek to create a {\em general system} that supports multiple settings, highlighting three here:

\heading{(1) Non-partitioned databases} are the most common use case in DP. A group of untrusted analysts issue queries over time against a static database, and the database owner wishes to enforce a global DP guarantee.
\sysname should allow a larger workload of queries compared to existing approaches.

\heading{(2) and (3) Partitioned databases} are less frequently investigated in DP theory literature, but important to distinguish in practice~\cite{pinq,parallel_composition_2022}.
When queries tend to access different data ranges, it is worth partitioning the data and accounting for consumed privacy budget in each partition separately through DP's parallel composition.
This lowers privacy budget consumption in each partition and permits more non- or partially-overlapping queries against the database.
This kind of workload is inherent in {\em timeseries} and {\em streaming databases}, where analysts typically query the data by {\em windows of time}, such as how many new Covid cases occurred in the week after a certain event, or what is the average age of positive people over the past week.
We distinguish two cases:

{\bf (2) \Partitionedstaticdb}, where the database is static and partitioned by an attribute that tends to be accessed in ranges, such as time, age, or geo-location. All partitions are available at the beginning. Queries arrive over time and most are assumed to run on some range of interest, which can involve one or more partitions. \sysname should provide significant benefit not only compared to the baseline caching techniques, but also compared to not having partitioning.

    {\bf (3) \Partitionedstreamingdb,} where the database is partitioned by time and partitions arrive over time. In such workloads, queries tend to run {\em continuously} as new data becomes available. Hence, new partitions see a similar query workload as preceding partitions.  \sysname should take advantage of this similarity to further conserve privacy.

For all three use cases, we aim to support {\em online workloads} of queries that are not all known upfront.
As \S\ref{sec:related-work} reviews, most works on optimizing global privacy budget consumption operate in the {\em offline setting}, where all queries are known upfront.  For that setting, algorithms are known to answer all queries simultaneously with optimal use of privacy budget. However, this setting is unrealistic for real use cases, where analysts adapt their queries based on previous results, or issue new queries for different analyses. In such cases, which correspond to the {\em online setting}, we require adaptive algorithms that accurately answer queries on-the-fly. \sysname does this by making effective use of PMW, as we next describe.

\subsection{\sysname Architecture}
\label{sec:architecture}

\begin{figure}[t]
    \includegraphics[width=\linewidth]{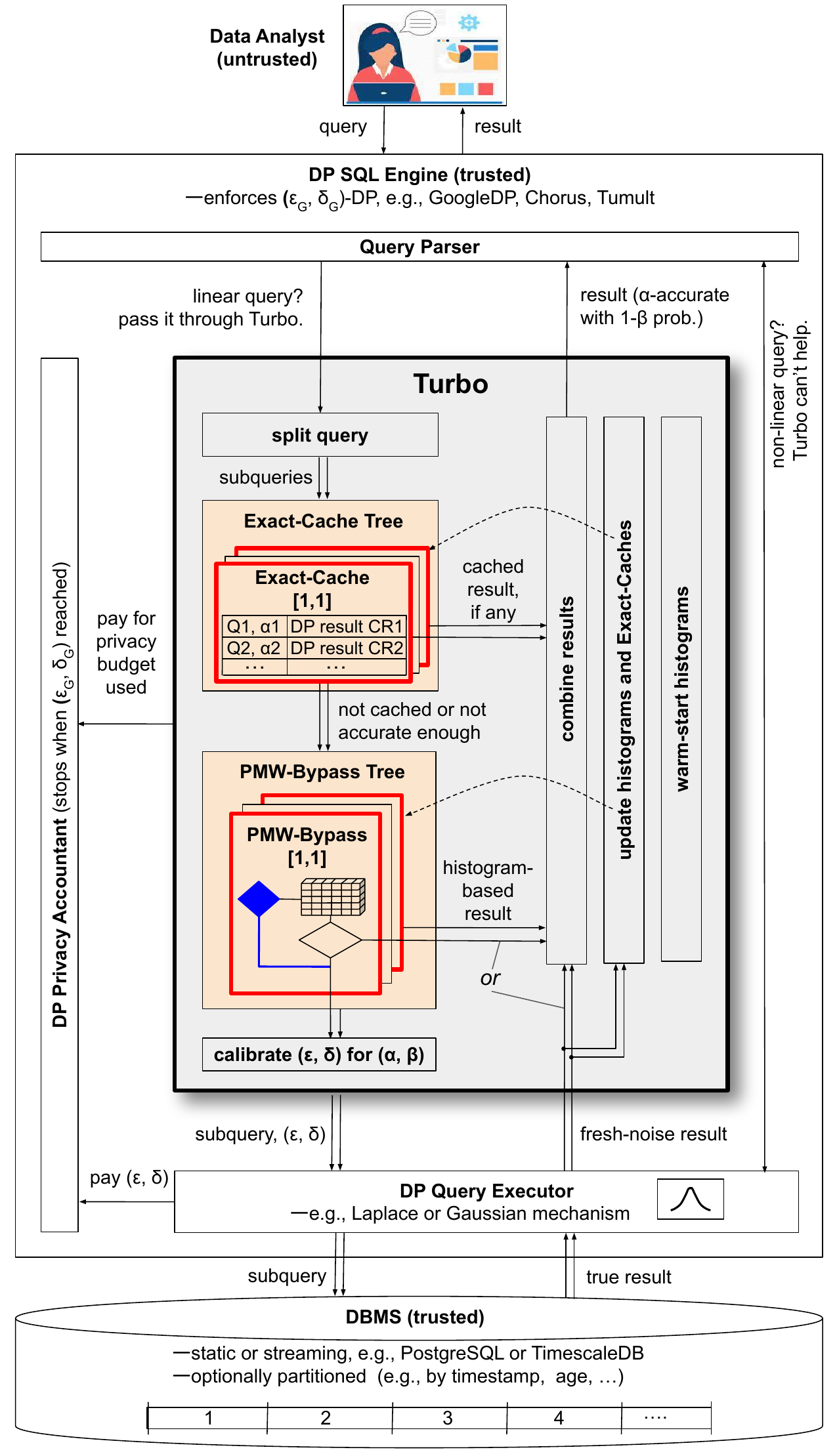}
    \caption{{\bf \sysname architecture.}
    }
    \label{fig:architecture}
\end{figure}

\F\ref{fig:architecture} shows the \sysname architecture.
It is a caching layer that can be added to a DP SQL engine, like GoogleDP~\cite{zetasql}, Uber Chorus~\cite{chorus}, or Tumult Analytics~\cite{tumult}, to boost the number of linear queries that can be answered accurately under a fixed global DP guarantee. The filled components indicate our additions to the DP SQL engine, while the transparent components are standard in DP SQL engines.
Here is how a typical DP SQL engine works {\em without \sysname}. Analysts issue queries against the engine, which is trusted to enforce a global $(\epsilon_G, \delta_G)$-DP guarantee. The engine executes the queries using a DP query executor, which adds noise to query results with the Laplace/Gaussian mechanism and consumes a part of the global privacy budget.  A budget accountant tracks the consumed budget; when it runs out, the DP SQL engine either stops responding to new queries (as do Chorus and Tumult Analytics) or sacrifices privacy by ``resetting'' the budget (as does LinkedIn Audience Insights). We assume the former.

\sysname intercepts the queries before they go into the DP query executor and performs a very proactive form of caching for them, reusing prior results as much as possible to avoid consuming privacy budget for new queries.
\sysname's architecture is organized in two types of components: {\em caching objects} (denoted in light-orange background in \F\ref{fig:architecture}) and {\em functional components} that act upon them (denoted in grey background).

\heading{Caching objects.}
\sysname maintains several types of caching objects.
First, the {\em \exactcache} stores previous queries and their DP results, allowing for direct retrieval of the result without consuming any privacy budget when the same query is seen again on the same database version.
Second, the {\em \pmwbypass} is an improved version of PMW that reduces privacy budget consumption during the training phase of its histogram (\S\ref{sec:pmw-bypass}). Given a query, \pmwbypass uses an effective heuristic to judge whether the histogram is sufficiently trained to answer the query accurately; if so, it uses it, thereby spending no budget.
Critically, \pmwbypass includes a mechanism to {\em externally update} the histogram even when bypassing it, to continue training it for future, free-budget queries.

\sysname aims to enable parallel composition for workloads that benefit from it, such as timeseries or streaming workloads, by supporting database partitioning. In theory, partitions could be defined by attributes with public values that are typically queried by range, such as time, age, or geo-location. In this paper, we will focus on partitioning by time.  \sysname uses a {\em tree-structured \pmwbypass} caching object, consisting of multiple histograms organized in a binary tree, to support linear range queries over these partitions effectively (\S\ref{sec:tree-structured-caching}). This approach conserves more privacy budget and enables larger workloads to be run when queries access only subsets of the partitions, compared to alternative methods.

\heading{Functional components.}
When \sysname receives a linear query through the DP SQL engine's query parser, it applies its caching objects to the query. If the database is partitioned, \sysname splits the query into multiple sub-queries based on the available tree-structured caching objects. Each sub-query is first passed through an \exactcache, and if the result is not found, it is forwarded to a \pmwbypass, which selects whether to execute it on the histogram or through direct Laplace/Gaussian.  For sub-queries that can leverage histograms, the answer is supplied directly without execution or budget consumption. For sub-queries that require execution with Laplace/Gaussian, the $(\epsilon, \delta)$ parameters for the mechanism are computed based on the $(\alpha, \beta)$ accuracy parameters, using the ``calibrate $(\epsilon, \delta)$ for $(\alpha, \beta)$'' functional component in \F\ref{fig:architecture}. Then, each sub-query and its privacy parameters are passed to the DP query executor for execution.

\sysname combines all sub-query results obtained from the caching objects to form the final result, ensuring that it is within $\alpha$ of the true result with probability $1-\beta$ (functional component ``combine results''). New results computed with fresh noise are used to update the caching objects (functional component ``update histograms and Exact-Caches''). Additionally, \sysname includes cache management functionality, such as ``warm-start of histograms,'' which reuses trained histograms from previous partitions to warm-start new histograms when a new partition is created (\S\ref{sec:warm-start}).
This mechanism is effective in streams where the data's distribution and query workload are stable across neighboring partitions. 
Theoretical and experimental analyses show that external histogram updates and warm-starting give convergence properties similar to, but slightly slower than, vanilla PMW.                 

\section{Detailed Design}
\label{sec:detailed-design}

We next detail the novel caching objects and mechanisms in \sysname, using different use cases from \S\ref{sec:use-cases} to illustrate each concept.  We describe \pmwbypass in the static, \nonpartitioneddb, then introduce partitioning for the \pmwbypasstree, followed by the addition of streaming to discuss warm-start procedures.
We focus on the Laplace mechanism and basic composition, thus only discussing pure $(\epsilon, 0)$-DP and ignoring $\delta$.
We also assume $\beta$ is small enough for \sysname results to count as $\alpha$-accurate w.h.p.
Appendix~\ref{sec:appendix:gaussian_rdp} extends all our theoretical results to $(\epsilon,\delta)$-DP, non-zero $\beta$, the Gaussian mechanism, and R\'enyi composition; in theory, all these should help to further conserve privacy budget, so we speculate they will be important for practice, but we leave their implementation and evaluation for future work.

\subsection{Notation}
\label{sec:notation}

Our algorithms require some notation.
Given a data domain $\mathcal{X}$, a database $x$ with $n$ rows can be represented as a histogram $h \in \mathbb{N}^\mathcal{X}$ as follows: for any data point $v \in \mathcal{X}$, $h(v)$ denotes the number of rows in $x$ whose value is $v$.
$h(v)$ is the {\em bin} corresponding to value $v$ in the histogram.
We denote $N=|\mathcal{X}|$ the size of the data domain and $n$ the size of the database.
When $\mathcal{X}$ has the form $\{0,1\}^d$, we call $d$ the data domain dimension.
Example: a database with $3$ binary attributes has domain $\mathcal{X} = \{0,1\}^3$ of dimension $d=3$ and size $N=8$; $h(0,0,1)$ is the number of rows that are equal to $(0,0,1)$.
\S\ref{sec:running-example} exemplifies a database, its dimensions, and its histogram.

We define {\em linear queries} as SQL queries that can be transformed or broken into the following form:

\noindent\begin{minipage}{.5\textwidth}
    \begin{lstlisting}[frame=none]
  SELECT AVG(*) FROM ( SELECT q(A, B, C, ...) FROM Table ),
\end{lstlisting}
\end{minipage}
where {\small \tt q} takes $d$ arguments (one for each attribute of {\small \tt Table}, denoted $A, B, C, ...$) and outputs a value in $[0,1]$. When {\small \tt q} has values in $\{0,1\}$, a query returns the {\em fraction of rows satisfying predicate {\small \tt q}}. To get raw counts, we multiply by $n$, which we assume is public information.
PMW (and hence \sysname) is designed to support only linear queries.
Examples of {\em non-linear} queries are: maximum, minimum, percentiles, top-k.

\subsection{Running Example}
\label{sec:running-example}

\begin{figure}[t]
    \includegraphics[width=\linewidth]{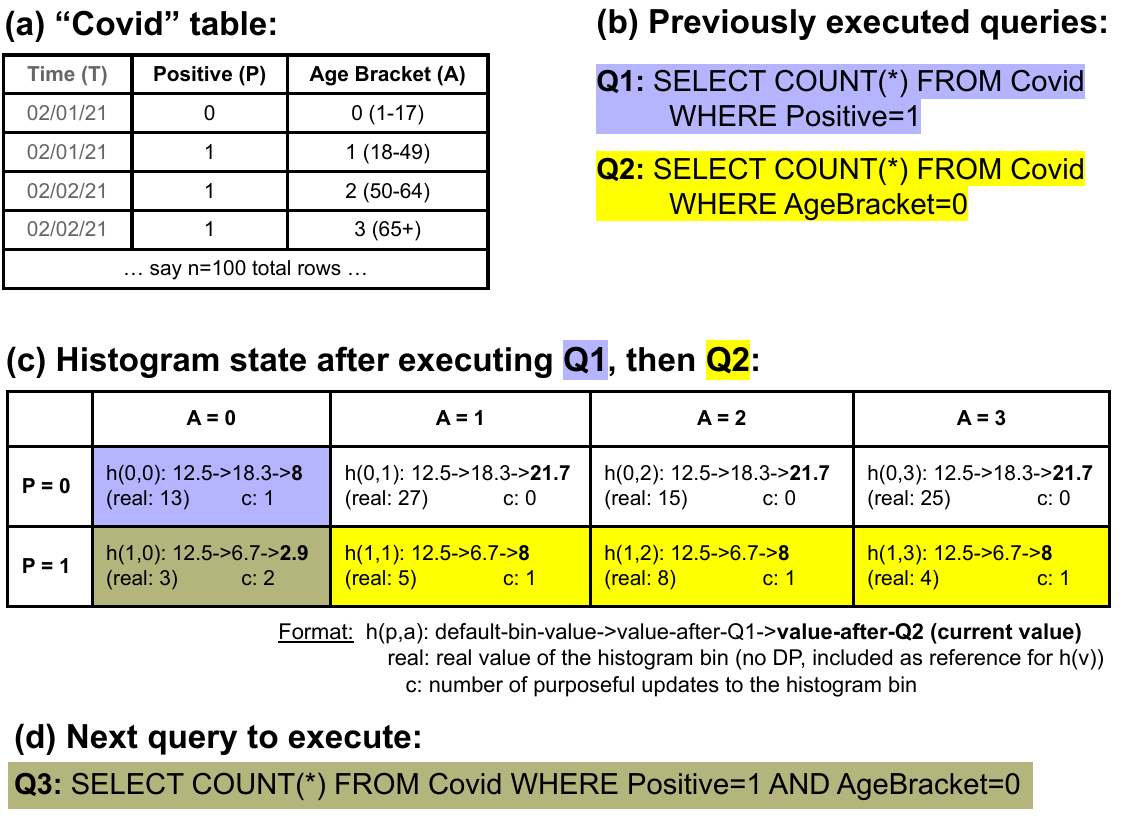}
    \caption{{\bf Running example}. (a) Simplified \covid19 tests dataset with $n=100$ rows and data domain size $N=8$ for the two non-time attributes, test outcome $P$ and subject's age bracket $A$. (b) Two queries that were previously run.
        (c) State of the histogram as queries are executed.
        (d) Next query to run.}
    \label{fig:running-example}
\end{figure}

\F\ref{fig:running-example} gives a running example inspired by our evaluation \covid19 dataset.
Analysts run queries against a database consisting of Covid test results over time.
\F\ref{fig:running-example}(a) shows a simplified version of the database, with only three attributes: the test's date, T; the outcome, P, which can be 0 or 1 for negative/positive; and subject's age bracket, A, with one of four values as in the figure.
The database could be either static or actively streaming in new test data. Initially, we assume it static and ignore the T attribute.
Our example database has $n=100$ rows and data domain size $N=8$ for P and A.

\F\ref{fig:running-example}(b) shows two queries that were previously executed.
While queries in \sysname return the {\em fraction} of entries satisfying a predicate, for simplicity we show raw counts.  $Q1$ requests the positivity rate and $Q2$ the fraction of tested minors. \F\ref{fig:running-example}(c) illustrates the histogram representation corresponding to the dataset, as estimated by the PMW algorithm, whose execution we discuss shortly.
\F\ref{fig:running-example}(d) shows the next query that will be executed, $Q3$, requesting the fraction of positive minors.
$Q3$ is not identical to either $Q1$ or $Q2$, but it is {\em correlated} with both, as it accesses data that overlaps with both queries.
Thus, while neither $Q1$'s nor $Q2$'s DP results can be used to directly answer $Q3$, intuitively, they both should help.
That is the insight that \pmw (and \pmwbypass) exploits through its query-by-query build-up of a DP histogram representation of the database that becomes increasingly accurate in bins that are accessed by more queries.

\F\ref{fig:running-example}(c) shows the state of the histogram after executing $Q1$ and $Q2$ but before executing $Q3$.
Each bin in the histogram stores an {\em estimation} of the number of rows equal to $(p,a)$.
This is the $h(p,a)$ field in the figure, for which we show the sequence of values it has taken following updates due to $Q1$ and $Q2$.
Initially, $h(p,a)$ in all bins is set assuming a uniform distribution over $P \times A$; in this case the initial value was $n/N = 12.5$.
The figure also shows the real (non-private) count for each bin (denoted {\em real}), which is {\em not} part of the histogram, but we include it as a reference.
As queries are executed, $h(p,a)$ values are updated with DP results, depending on which bins are accessed.
$Q1$ and $Q2$ have already been executed, and both are assumed to have resorted to the Laplace mechanism, so they both contributed DP results to specific bins (we specify the update algorithm later when discussing \A\ref{alg:pmw-bypass}).
$Q1$ accessed, and hence updated, data in the $P=1$ bins (the bottom row of the histogram).
$Q2$ did so in the $A=0$ bins (the left column of the histogram).
Through a renormalization step, t hese queries have also changed the other bins, though not necessarily in a query-informed way.
The $c$ variable in each bin shows the number of queries that have purposely updated that bin.
We can see that estimates in the $c>0$ bins are a bit more accurate compared to those in the $c=0$ bins.
The only bin that has been updated twice is $(P=1,A=0)$, as it lies at the intersection of both queries; that bin has diverged from its neighboring, singly-updated bins and is getting closer to its true value.
(Bin $(P=1,A=2)$, updated only once, is even more accurate purely by chance.)

Our last query, $Q3$, which accesses $(P=1,A=0)$, may be able to leverage its estimation ``for free,'' assuming the estimation's error is within $\alpha$ w.h.p.
Assessing that the error is within $\alpha$ -- privately, and without consuming privacy budget if it is -- is the purview of the SV mechanism incorporated in a \pmw.
The catch is that the SV consumes privacy budget, in copious amounts, if this test fails.
This is what makes vanilla \pmw impractical, a problem that we address next.
\subsection{\pmwbypass}
\label{sec:pmw-bypass}

\pmwbypass addresses practical inefficiencies of \pmw, which we illustrate with simple demonstration.

\heading{Demo experiment.}
Using a four-attribute \covid19 dataset with domain size 128 (so a bit larger than in our running example), we generate a query pool of over 34K unique queries by taking all possible combinations of values over the four attributes.
From this pool, we sample uniformly with replacement 35K queries to form a workload; there is therefore some identical repetition of queries but not much.
This workload is not necessarily realistic, but it should be an {\em ideal showcase} for \pmw: there are many unique queries relative to the small data domain size (giving the \pmw ample chance to train), and while most queries are unique, they tend to overlap in the data they touch (giving the \pmw ample chance to reuse information from previous queries).
We evaluate the cumulative privacy budget spent as queries are executed, comparing the case where we execute them through \pmw vs. directly with Laplace, with and without an exact-match cache.
\setlength{\columnsep}{2pt}
\begin{wrapfigure}{r}{0.66\linewidth}
	\vspace{-0.3cm}
	\begin{center}
		\includegraphics[width=\linewidth]{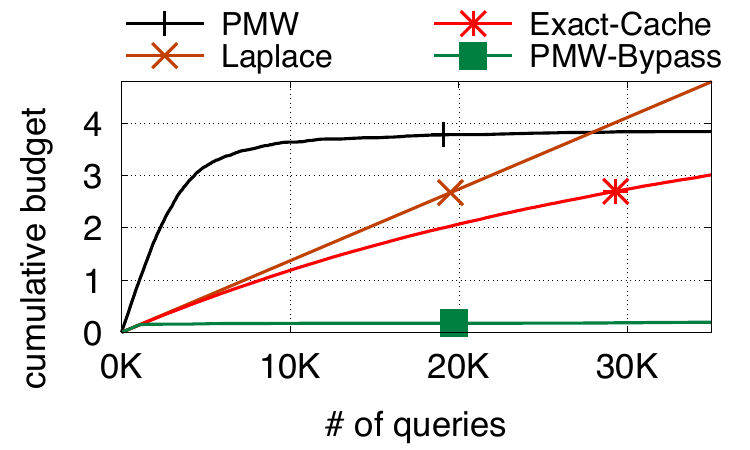}
	\end{center}
	\caption{{\bf Demo experiment.}} \label{fig:pmw-laplace-demo}
\end{wrapfigure}
\F\ref{fig:pmw-laplace-demo} shows the results.
As expected for this workload, the \pmw works, as it converges after roughly the first 10K queries and consumes very little budget afterwards.
However, before converging, the PMW consumes enormous budget.
In contrast, direct execution through Laplace grows linearly, but more slowly compared to PMW's beginning.
The \pmw eventually becomes better than Laplace, but only after $\approx 27K$ queries.

Moreover, if instead of always executing with Laplace, we trivially cached the results in an exact-match cache for future reuse if the same query reappeared -- a rare event in this workload -- then the \pmw would {\em never} become notably better than this simple baseline!
This happens for a workload that should be ideal for PMW.
\S\ref{sec:evaluation} shows that for other workloads, less favorable for PMW but more realistic, the outcome persists: {\em PMWs underperform even the simplest baselines in practice}.

We propose {\em \pmwbypass}, a re-design for PMWs that releases their power and makes them {\em very effective}.
We make multiple changes to PMWs, but the main one involves {\em bypassing} the PMW while it is training (and hence expensive) and instead executing directly with Laplace (which is less expensive).
Importantly, we do this while still updating the histogram with the Laplace results so that eventually the PMW becomes good enough to switch to it and reap its zero-privacy query benefits.
The \pmwbypass line in \F\ref{fig:pmw-laplace-demo} shows just how effective this design is in our demo experiment: \pmwbypass follows the low, direct-Laplace curve instead the PMW's up until the histogram converges, after which it follows the flat shape of PMW's convergence line.
In this experiment, as well as in others in \S\ref{sec:evaluation}, the outcome is the same: {\em our changes make PMWs very effective}.
We thus believe that \pmwbypass should replace PMW in most settings where the latter is studied, not just in our system's design.

\heading{\pmwbypass.}
\F\ref{fig:pmw-bypass} shows the functionality of \pmwbypass, with the main changes shown in blue and bold.
Without our changes, a vanilla PMW works as follows.
Given a query $Q$, \pmw first estimates its result using the histogram ($R1$) and then uses the SV protocol to test whether it is $\alpha$-accurate w.h.p.  The test involves comparing $R1$ to the {\em exact result} of the query executed on the database.  If a noisy version of the absolute error between the two is within a threshold comfortably far from $\alpha$, then $R1$ is considered accurate w.h.p. and outputted directly.  This is the good case, because the query need not consume {\em any privacy}.  The bad case is when the SV test fails. First, the query must be executed directly through Laplace, giving a result $R2$, whose release costs privacy.  But beyond that, the SV must be {\em reset}, which consumes privacy. In total, if the Laplace execution costs $\epsilon$, then releasing $R2$ costs $4*\epsilon$! This is what causes the extreme privacy consumption during the training phase for vanilla PMW, when the SV test mostly fails. Still, in theory, after paying handsomely for this histogram ``miss,'' $R2$ can be used to update the histogram (the arrow denoted ``update (R2)'' in \F\ref{fig:pmw-bypass}), in hopes that future correlated queries ``hit'' in the histogram.

\begin{figure}[t]
	\centering
	\includegraphics[width=\linewidth]{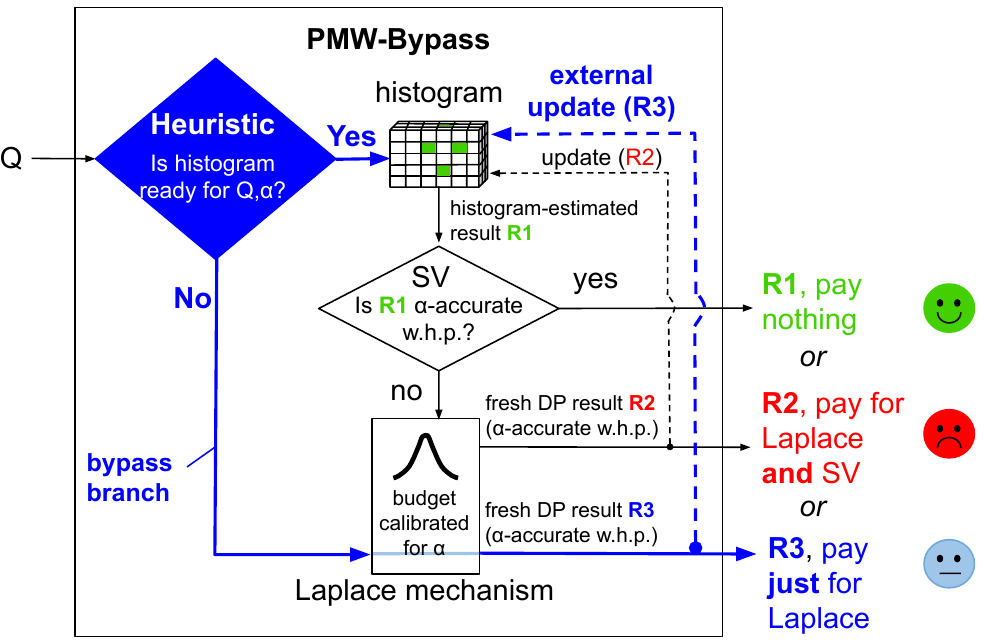}
	\captionof{figure}{{\bf \pmwbypass.} New components over vanilla PMW are in blue/bold.}  \label{fig:pmw-bypass}
	\vspace{-0.4cm}
\end{figure}

\pmwbypass adds three components to PMW: (1) a {\em heuristic} that assesses whether the histogram is likely ready to answer $Q$ with the desired accuracy; (2) a {\em bypass branch}, taken if the histogram is deemed not ready and  direct query execution with Laplace instead of going through (and likely failing) the SV test; and (3) an {\em external update} procedure that updates the histogram with the bypass branch result.
Given $Q$, \pmwbypass first consults the heuristic, which only inspects the histogram, so its use is free. Two cases arise:

\underline{Case 1:} If the heuristic says the histogram is ready to answer Q with $\alpha$-accuracy w.h.p., then the PMW is used, $R1$ is generated, and the SV is invoked to test $R1$'s actual accuracy.
If the heuristic's assessment was correct, then this test will succeed, and hence the free, $R1$ output branch will be taken.
Of course, no heuristic that lacks access to the raw data can guarantee that $R1$ will be accurate enough, so if
the heuristic was actually wrong, then the SV test will fail and the expensive $R2$ path is taken.
Thus, a key design question is whether there exist heuristics good enough to make \pmwbypass effective.
We discuss heuristic designs below, but the gist is that simple and easily tunable heuristics work well, enabling the significant privacy budget savings in \F\ref{fig:pmw-laplace-demo}.

\underline{Case 2:} If the heuristic says the histogram is not ready to answer Q with $\alpha$-accuracy w.h.p., then the bypass branch is taken and Laplace is invoked directly, giving result $R3$.  
Now, \pmwbypass must pay for Laplace, but because it bypassed the PMW, it does not risk an expensive SV reset.
A key design question here is whether we can still reuse $R3$ to update the histogram, even though we did not, in fact, consult the SV to ensure that the histogram is truly insufficiently trained for Q.
We prove that performing the same kind of update as the PMW would do, from outside the protocol, would break its theoretical convergence guarantee.
Thus, for \pmwbypass, we design an {\em external update} procedure that {\em can} be used to update the histogram with $R3$ while preserving the PMW's \theoreticalconvergence, albeit at slower speed.

\begin{algorithm}[t]
	\centering
	\begin{algorithmic}[1]
		\State \textbf{Cfg.:} $\textproc{PrivacyAccountant}$, $\textproc{Heuristic}$, accuracy params $(\alpha, \beta)$, histogram convergence params $\lr, \tau$, database $\textproc{Data}$ with $n$ rows.
		%
		\vspace{0.1cm}
		\Function{Update}{$h, q, s$}
		\State Update estimated values: $\forall v \in \mathcal{X}, g(v) \gets h(v)e^{s * q(v)}$
		\State Renormalize: $\forall v \in \mathcal{X}, h(v) \gets g(v) / \sum_{w \in \mathcal{X}} g(w)$
		\State \Return $h$
		\EndFunction
		\vspace{0.1cm}
		\Function{CalibrateBudget}{$\alpha, \beta$}
		\State \Return $\frac{4\ln(1/\beta)}{n\alpha}$
		\EndFunction
		\vspace{0.1cm}
		\State Initialize histogram $h$ to uniform distribution on $\mathcal{X}$
		\State $\epsilon \gets \textproc{CalibrateBudget}(\alpha, \beta)$\label{line:calibrate}
		\State $\textproc{PrivacyAccountant.pay}(3 \cdot \epsilon)$ \ \ \ \ \ \ \ \ // Pay to initialize first SV
		\While{$\textproc{PrivacyAccountant.HasBudget()}$}
		\State $\hat \alpha \gets \alpha/2 + \lap{1/\epsilon n}$ \ \ // SV reset \label{line:sv_reset}
		\State $SV \gets \textproc{NotConsumed}$
		\While{$SV == \textproc{NotConsumed}$}
		\State Receive next query $q$
		\If{$\textproc{Heuristic.IsHistogramReady}(h, q, \alpha, \beta)$} \label{line:heuristic}
		\State {\bf // Regular PMW branch:}
		\If{$|q(\textproc{data}) - q(h)| + \lap{1/\epsilon n} < \hat \alpha$} \ \ // SV test \label{line:sv_test}
		\State Output $R1 = q(h)$  \ \ \ \ \ \ \ \ \ \ \ \ \ \ \ \ \ \ \ \ \ \ \  \ \ \ \ \ \ \ \ \ {\textcolor{green}{$\rightarrow${\bf R1}, pay nothing}}
		\Else{}
		\State $\textproc{PrivacyAccountant.pay}(4*\epsilon)$ \ \textcolor{red}{$\rightarrow$ {\bf R2}, pay for}
		\State Output $R2 = q(\textproc{data}) + \lap{1/\epsilon n}$\ \ \ \ \ \ \ \textcolor{red}{Laplace, SV} \label{line:R2}
		\State // Update histogram (R2):
		\State $s \gets \begin{cases}
				\lr  & \text{if}\ R2 > q(h) \\
				-\lr & \text{if}\ R2 < q(h)
			\end{cases}$
		\State $h \gets \textproc{Update}{(h, q, s)}$
		\State $SV \gets \textproc{Consumed}$ \ \ // force SV reset
		\State $\textproc{Heuristic.Penalize}(q, h)$
		\EndIf
		\Else{}
		\State {\bf // Bypass branch:}
		\State $\textproc{PrivacyAccountant.pay}(\epsilon)$  \ \ \ \ \ \ \ \ \ \ \ \ {\textcolor{blue}{$\rightarrow$ {\bf R3}, pay for}}
		\State Output $R3 = q(\textproc{data}) + \lap{1/\epsilon n}$ \ \ \ \ \ \ \ \ \ \ \ \textcolor{blue}{Laplace} \label{line:R3}
		\State // External update of histogram (R3):
		\State $s \gets \begin{cases}
				\lr  & \text{if}\ R3 > q(h) + \tau\alpha                         \\
				-\lr & \text{if}\ R3 < q(h) - \tau\alpha                         \\
				0    & \text{otherwise\ \ // no updates if we're not confident!}
			\end{cases}$ \label{line:external_update}
		\State $h \gets \textproc{Update}{(h, q, s)}$
		\EndIf
		\EndWhile
		\EndWhile
	\end{algorithmic}
	\caption{\footnotesize\bf \pmwbypass algorithm.}
	\label{alg:pmw-bypass}
\end{algorithm}

\heading{Heuristic {$\textproc{IsHistogramReady}$}.}
One option to assess if a histogram is ready to answer a query accurately is to check if it has received at least $C$ updates, for some global threshold $C$.
However, this approach is often imprecise as it fails to detect histogram regions that might still be untrained.
Thus, we use a separate threshold value per bin, raising the question of how to configure all these thresholds.
To keep configuration easy (goal {\bf (G6)}), we use an {\em adaptive per-bin threshold}.
For each bin, we initialize its threshold $C$ with a value $C_0$ and increment $C$ by an additive step $S_0$ every time the heuristic errs (\ie, predicts it is ready when it is in fact not ready for that query).
While the threshold is too small, the heuristic gets penalized until it reaches a threshold high enough to avoid mistakes.
For queries that span multiple bins, we only penalize the least-updated bins to prevent a single, inaccurate bin from setting back the histogram from queries using accurate bins only.
With these thresholds, we only configure initial parameters $C_0$ and $S_0$, which we find experimentally easy to do (\S\ref{sec:evaluation:non-partitioned-database}).

\heading{External updates.}
While we want to bypass the PMW when the histogram is not ``ready'' for a query, we still want to update the histogram with the result from the Laplace execution (R3); otherwise, the histogram will never get trained.
That is the purpose of our external updates (lines 33-34 in \A\ref{alg:pmw-bypass}).
They follow a similar structure as a regular PMW update (lines 24-25 in \A\ref{alg:pmw-bypass}), with a key difference.
In vanilla PMW, the histogram is updated with the result $R2$ from Laplace {\em only when} the SV test fails.
In that case, \pmw updates the relevant bins in one direction or another, depending on the sign of the error $R2 - q(h)$.  For example, if the histogram is underestimating the true answer, then R2 will likely be higher than the histogram-based result, so we should increase the value of the bins (case $R2 > q(h)$ of line 24 in \A\ref{alg:pmw-bypass}).

In \pmwbypass, external updates are performed not just when the authoritative SV test finds the histogram estimation inaccurate, but also when our heuristic predicts it to be inaccurate even though it may actually be accurate. In the latter case, performing external updates in the same way as PMW updates would add bias into the histogram and forfeit its convergence guarantee.  To prevent this, in \pmwbypass, external updates are executed only when we are quite confident, based on the direct-Laplace result $R3$, that the histogram overestimates or underestimates the true result.  Line 33 shows the change: the term $\tau\alpha$ is a {\em safety margin} that we add to the comparison between the histogram's estimation and $R3$, to be confident that the estimation is wrong and the update warranted. This lets us prove \theoreticalconvergence akin to \pmw.
Finally, like regular PMW updates, external updates reuse the already DP result $R3$, hence they do not consume any additional privacy budget beyond what was already consumed to generate $R3$.

\heading{Learning rate.}
In addition to the bypass option, we make another key change to PMW design for practicality.
When updating a bin, we increase or decrease the bin's value based on a learning rate parameter, $\lr$, which determines the size of the update step taken (line 3 in \A\ref{alg:pmw-bypass}).
Prior PMW works fix learning rates that minimize theoretical convergence time, typically $\alpha/8$~\cite{complexity}.
However, our experiments show that larger values of $\lr$ can lead to much faster convergence, as dozens of updates may be needed to move a bin from its uniform prior to an accurate estimation.
However, increasing $\lr$ beyond a certain point can impede convergence, as the updates become too coarse.
Taking cue from deep learning, \pmwbypass uses a scheduler to adjust $lr$ over time.
We start with a high $\lr$ and progressively reduce it as the histogram converges.

\heading{Guarantees.}
{\bf (G1)} {\em Privacy:} \pmwbypass preserves $\epsilon_G$-DP across the queries it executes (\Thm\ref{thm:pmwbypass:privacy}). 
	{\bf (G2)} {\em Accuracy:} \pmwbypass is $\alpha$-accurate with $1-\beta$ probability for each query (\Thm\ref{thm:pmwbypass:accuracy}).
This property stems from how we calibrate Laplace budget $\epsilon$ to $\alpha$ and $\beta$.
This is function $\textproc{CalibrateBudget}$ in \A\ref{alg:pmw-bypass} (lines 6-7).
For $n$ datapoints, setting $\epsilon = \frac{4\ln(1/\beta)}{n\alpha}$ ensures that each query is answered with error at most $\alpha$ with probability $1-\beta$. 
	{\bf (G3)} {\em \Theoreticalconvergence:} 
If $\lr/\alpha < \tau \le 1/2$, then \whp \pmwbypass needs to perform at most
$\frac{\ln |\mathcal{X}|}{ \lr(\tau\alpha-\lr)/2}$
updates (\Thm\ref{thm:pmwbypass:convergence}).
\pmwbypass's worst-case convergence is thus similar to \pmw's, but roughly $1/2\tau$ times slower.
\S\ref{sec:evaluation:non-partitioned-database} confirms this empirically.

\subsection{Tree-Structured \pmwbypass}
\label{sec:tree-structured-caching}

We now switch to the partitioned-database use cases, focusing on time-based partitions, as in timeseries databases, whether static or dynamic.
Rather than accessing the entire database, analysts tend to query specific time windows, such as requesting the Covid positivity rate over the past week, or the fraction of minors diagnosed with Covid in the two weeks following school reopening.
This allows the opportunity to leverage DP's parallel composition: the database is partitioned by time (say a week's data goes in one partition), and privacy budget is consumed at the partition level. Queries can run at finer or coarser granularity, but they will consume privacy against the partition(s) containing the requested data.
With this approach, a system can answer more queries under a fixed global $(\epsilon_G, \delta_G)$-DP guarantee compared to not partitioning~\cite{privatekube,sage,pinq,roy2010airavat}.
We implement support for partitioning and parallel composition in \sysname through a new caching object called a {\em \pmwbypasstree}.

\begin{figure}[t]
    \includegraphics[width=\linewidth]{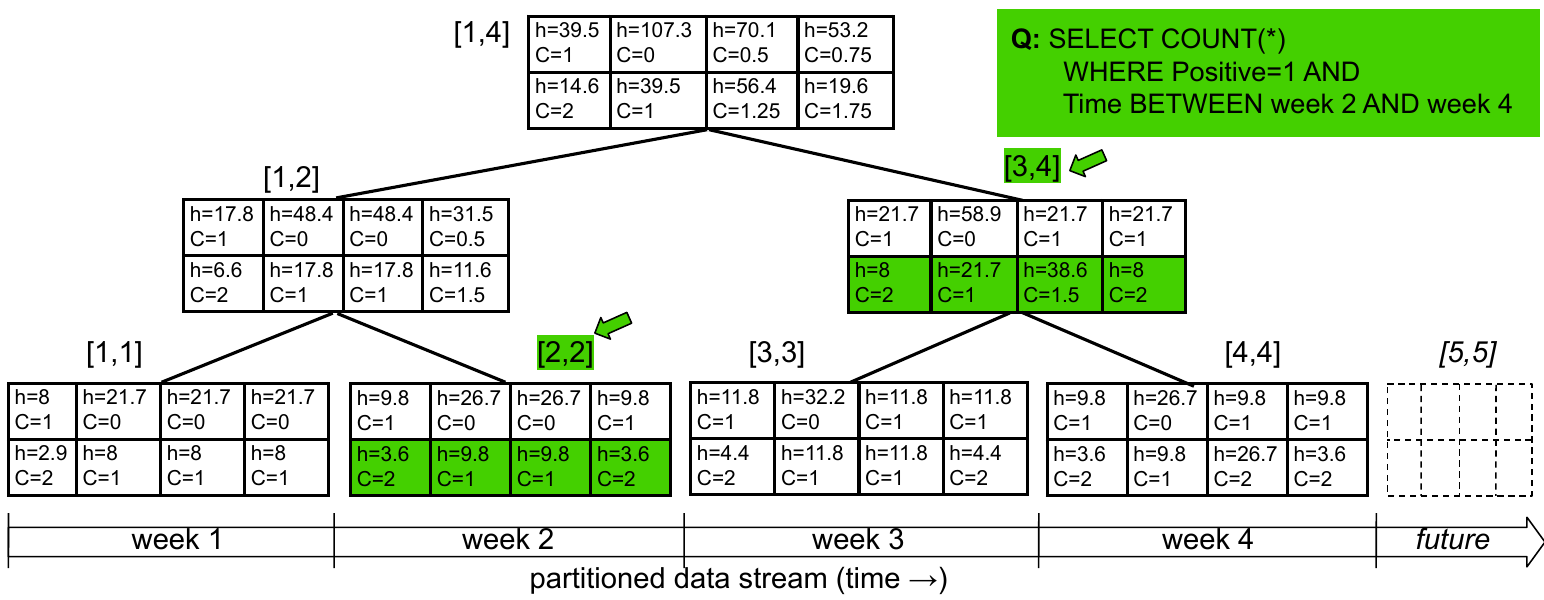}
    \caption{{\bf Example of tree-structured histograms.}
    }
    \label{fig:running-example-histogram-tree}
\end{figure}

\heading{Example.}
\F\ref{fig:running-example-histogram-tree} shows an extension of the running example in \S\ref{sec:running-example}, with the database partitioned by week. Denote $n_i$ the size of each partition.
A new query, $Q$, asks for the positivity rate over the past three weeks.
How should we structure the histograms we maintain to best answer this query?
One option would be to maintain {\em one histogram per partition} (i.e., just the leaves in the figure).
To resolve $Q$, we query the histograms for weeks 2, 3, 4.
Assume the query results in an update. Then, we need to update histograms, computing the answer with DP within our $\alpha$ error tolerance.
Updating histograms for weeks 2, 3, and 4 requires querying the result for each of them with parallel composition.
Given that $\textrm{Laplace}(1/n\epsilon)$ has standard deviation $\sqrt{2}/n\epsilon$, for week 4 for instance, we need noise scaled to $1/n_4\epsilon$.  Thus, we consume a fairly large $\epsilon$ for an accurate query to compensate for the smaller $n_4$.
Another option would be to use {\em one histogram per range} (\ie set of contiguous partitions), but that involves maintaining a large state that grows quadratically in the number of partitions.

Instead, our approach is to maintain a {\em binary-tree-structured set of histograms}, as shown in \F\ref{fig:running-example-histogram-tree}.  For each partition, but also for a binary tree growing from the partitions, we maintain a separate histogram.  To resolve $Q$, we split the query into two sub-queries, one running on the histogram for week 2 ([2,2]) and the other running on the histogram for the range week 3 to week 4 ([3,4]).  That last sub-query would then run on a larger dataset of size $n_3 + n_4$, requiring a smaller budget consumption to reach the target accuracy.

\begin{figure}[t]
    \includegraphics[width=0.8\linewidth]{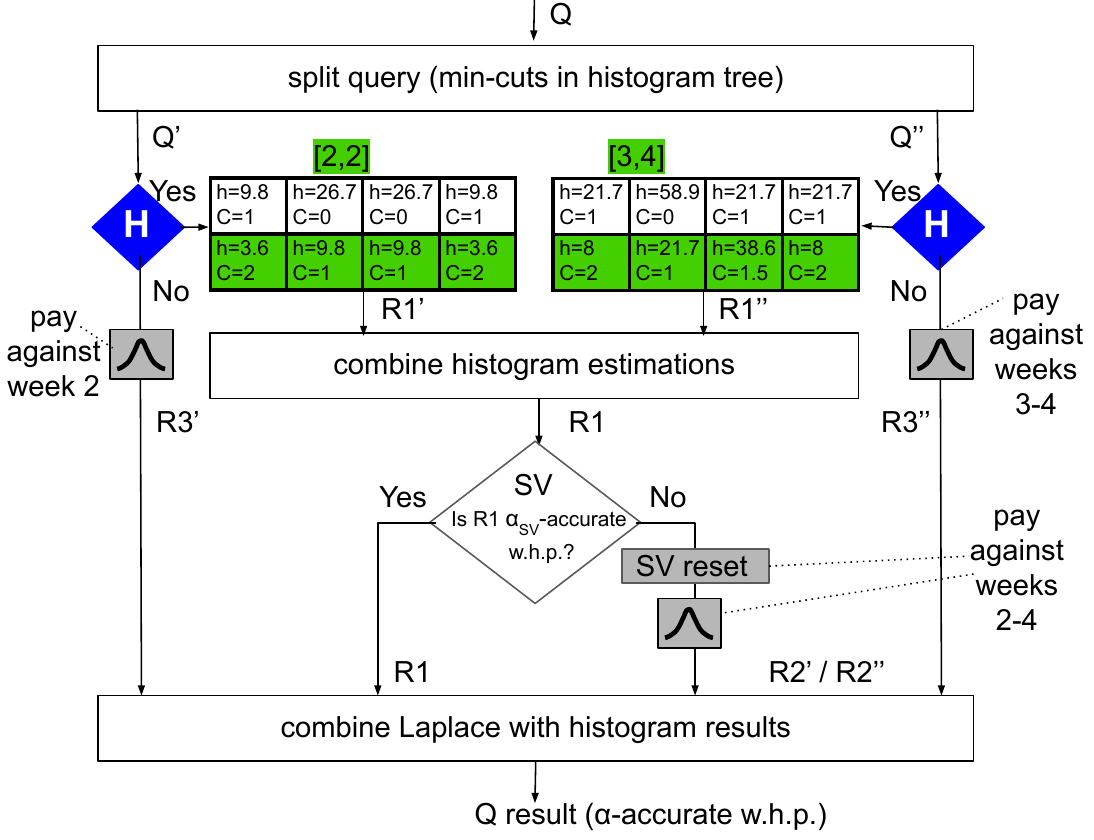}
    \caption{{\bf \Pmwbypasstree.}
    }
    \label{fig:tree-structured-pmw-bypass}
\end{figure}

\heading{Design.}
\F~\ref{fig:tree-structured-pmw-bypass} shows our design.
Given a query $Q$, we split it into sub-queries based on the histogram tree, applying the min-cuts algorithm to find the smallest set of nodes in the tree that covers the requested partitions.
In our example, this gives two sub-queries, $Q'$ and $Q''$, running on histograms [2,2] and [3,4], respectively.
For each sub-query, we use our heuristic to decide whether to use the histogram or invoke Laplace directly.
If both histograms are ``ready,'' we compute their estimations and combine them into one result, which we test with an SV against an accuracy goal.
In our example, there are only two sub-queries, but in general there can be more, some of which will use Laplace while others use histograms.
We adjust the SV's accuracy target to an $(\alpha_{SV}, \beta_{SV})$ calibrated to the aggregation that we will need to do among the results of these different mechanisms.
We pay for any Laplace's and SV resets against the queried data partitions and finally combine Laplace results with histogram-based results.
Each subquery updates the corresponding histograms of the tree (details in \A\ref{alg:pmwbypass-tree}) and increments $c$ for updated nodes.

\heading{Guarantees.}
{\bf (G1)} {\em Privacy} and {\bf (G2)} {\em accuracy} are unchanged (\Thm\ref{thm:pmwbypasstree:privacy},~\ref{thm:pmwbypasstree:accuracy}).
    {\bf (G3)} {\em \Theoreticalconvergence:} For $T$ partitions, if $\lr/\alpha < \tau \le 1/2$, then
\whp we perform at most $\frac{2T (\lceil\log T\rceil+1)\ln |\mathcal{X}|}{\eta ( \tau \alpha - \eta)/2}$ updates (\Thm\ref{thm:pmwbypasstree:convergence}).

\subsection{Histogram Warm-Start}
\label{sec:warm-start}

An opportunity exists in streams to warm-start histograms from previously trained ones to converge faster.  
Prior work on PMW initialization~\cite{pmwpub} only justifies using a public dataset close to the private dataset to learn a more informed initial value for histogram bins than a uniform prior.
We prove that warm-starting a histogram by copying an entire, trained histogram preserves the worst-case convergence.
In \sysname, we use two procedures: for new leaf histograms, we copy the previous partition's leaf node; for non-leaf histograms, we take the average of children histograms.
We also initialize the per-bin thresholds and update counters of each node.

\heading{Guarantees.}
{\bf (G1)} {\em Privacy} and {\bf (G2)} {\em accuracy} guarantees are unchanged.
    {\bf (G3)} {\em \Theoreticalconvergence:}
If there exists $\lambda \ge 1$ such that the initial histogram $h_0$ in \A\ref{alg:pmw-bypass} satisfies $\forall x \in \mathcal{X}, h_0(x) \ge \frac{1}{\lambda |\mathcal{X}|}$, then we show that each \pmwbypass converges, albeit at a slower pace (\Thm\ref{thm:warm-start}).
The same properties hold for the tree.

\section{Prototype Implementations}
\label{sec:implementation}

\begin{figure*}[t]
	\includegraphics[width=.8\linewidth]{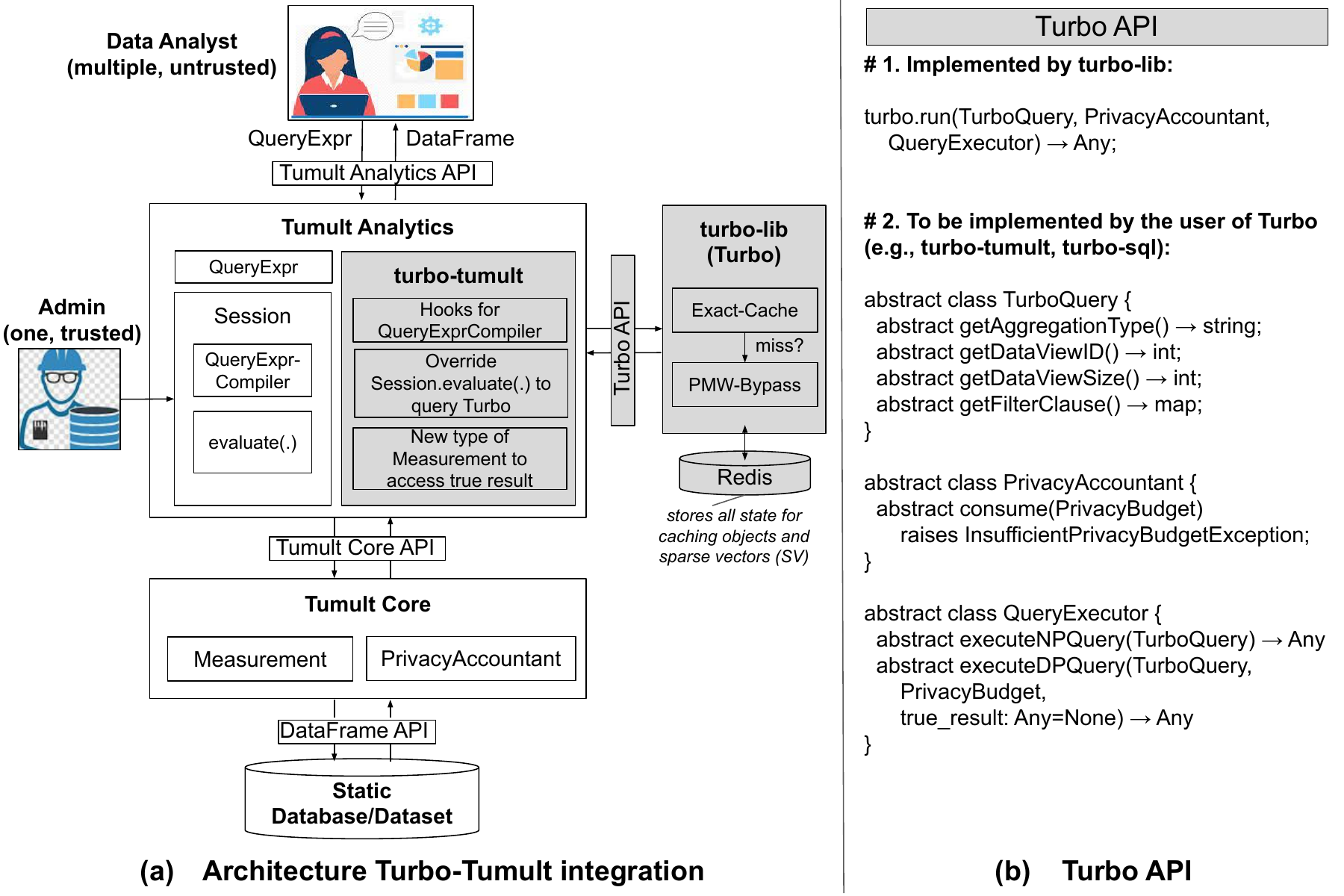}
	\caption{{\bf (a) \sysname integration into Tumult. (b) Turbo API.}
	}
	\label{fig:turbo-tumult-prototype}
\end{figure*}

We prototype \sysname in three components that we release open-source: (1) {\em turbo-lib}, a library that contains \sysname-specific functionality, notably the caching objects and functional components in the \sysname architecture (\F\ref{fig:architecture}); (2) {\em turbo-tumult}, a library that connects turbo-lib with Tumult Analytics, to add caching functionality into that existing DP system; and (3) {\em turbo-sql}, a basic standalone library to run a select subset of DP SQL queries through turbo-lib directly against a traditional, non-DP database, such as TimescaleDB or PostgreSQL.  The reason for both (2) and (3) is that Tumult provides a more complete database query engine, supporting a wide variety of Spark-SQL-like queries while having significant limitations with respect to parallel composition on partitioned databases.  Our integration with Tumult (2) shows that \sysname can be integrated with a real, existing DP system, while our standalone querying library (3) can let us experiment with both non-partitioned and partitioned databases, in both static- and streaming-DB settings.  We use a version of (3) (released through the SOSP'23 artifact) throughout our evaluation, but describe here predominantly our integration with Tumult, which can serve as a blueprint for integration with other existing DP systems in the future.  Finally, we separately release the artifact that we used in our evaluation and which was evaluated by the SOSP'23 artifact evaluation committee.  All are available from the repository: \url{https://github.com/columbia/turbo}.

\F\ref{fig:turbo-tumult-prototype}(a) shows the architecture of our Turbo-Tumult integration.
The grey boxes are Turbo-specific while the clear boxes are unchanged Tumult components.

\heading{Tumult overview.}
Without Turbo, Tumult functions as follows.  It consists of two main components: (1) Tumult Core, a library that implements primitive DP mechanisms and privacy accounting; and (2) Tumult Analytics, a layer on top of Tumult Core that exports a higher-level, Spark-SQL-like query interface on top of one or more static datasets or databases.
Tumult Core exports a low-level API consisting of a privacy accountant and a {\em measurement} abstraction, which is the Tumult terminology for a DP computation.  It implements the necessary methods to ``evaluate'' a measurement on top of a dataset and deduct its privacy budget against the accountant.  Tumult Analytics implements two main abstractions: (1) a {\em session}, which represents the context against which Tumult will enforce a global privacy budget across all queries issued against this session and (2) a Spark-SQL-like interface for analysts to construct queries that consists of multiple {\em transformations} chained one after another (such as filters, projections, joins, etc.) against one or more datasets, followed by a single {\em aggregate function} (such as an average, count, sum, median, percentiles, stdev, etc.), with potential for splitting and grouping the results by one or more attributes.
Compared to other DP SQL databases, it is our impression that Tumult supports a fairly wide range of SQL that can be handled with DP.

For the purposes of this paper, we will assume that an administrator creates a session upfront, specifying a global privacy budget to be enforced and hosts this session as a service to guard analysts' accesses to a sensitive dataset (or datasets) underneath.
Analysts, which can be many and are untrusted, send their query expressions for execution against the session.
The session is then responsible for executing each query by first compiling it into a measurement and then evaluating it through the Tumult Core, which will deduct the necessary privacy budget.
While the measurement abstraction is a quite general representation of a DP computation, Tumult Core and Analytics assume a Spark DataFrame-based API for interacting with the dataset(s) underneath.
Thus, measurements compiled through Tumult Analytics, will be Spark DataFrame queries -- to be executed through Spark -- in which Tumult Analytics transparently includes an additional operation that adds an appropriately scaled amount of noise to the result of the aggregation. A Tumult measurement encapsulates this compiled Spark DataFrame query, along with information regarding the privacy budget it is programmed to expend upon its execution.  Tumult Analytics hands over this measurement for Tumult Core, which executes the DataFrame query through Spark and deducts the measurement's reported privacy consumption through its privacy accountant.

\heading{Turbo-Tumult.}
The preceding describes Tumult and its main abstractions (relevant for this paper) {\em without Turbo}.  Tumult itself has no caching capabilities, so our integration aims to add Turbo as a caching layer in Tumult.  The integration consists of two components, denoted in grey in \F\ref{fig:turbo-tumult-prototype}(a).  First is turbo-lib, which contains the core Turbo functionality we described in this paper.  Turbo-lib exposes an API, Turbo API, consisting of the functions Turbo exports to and requires from any user of Turbo, such as turbo-tumult and turbo-sql.
Second is turbo-tumult, a small library that incorporates Turbo into Tumult by invoking and implementing different parts of the Turbo API.

Turbo-tumult takes a {\em light-touch approach} to incorporating Turbo into Tumult, which ensures that our system is easily adoptable. It manifests in two ways. First, we only extend, but do not modify, certain classes within Tumult Analytics and implement new types of measurements to extend, but not change, Tumult Core functionality.  Specifically, turbo-tumult  provides one type of externally visible abstraction: a new type of session, called {\em TurboSession}, which overrides the original's query evaluation function to: (1) incorporate a set of hooks into the query compiler such that certain information necessary for Turbo is extracted from the query, such as the dataset ID, the type of aggregate function, and the filtering conditions; and (2) if the query can be handled by Turbo, TurboSession passes it through turbo-lib instead of executing it directly on Tumult Core.  Turbo-lib then checks its own caching objects for an answer, but resorts to Tumult Core -- which it accesses back through the Turbo API, discussed shortly -- for execution of the query and for privacy budget deduction in the Tumult Core accountant. 

Second, we take a fail-to-Tumult approach for all queries.  Turbo supports a small subset of all queries supported by Tumult: e.g., we do not support joins, medians, percentiles, and a number of transformation functions allowed in Tumult.  Moreover, Turbo aims to control accuracy of the queries, and presently that accuracy must be specified upfront, when the cache (e.g., through TurboSession) service is created. Yet, analysts may wish to vary their accuracy targets per query, and in some cases may wish to specify privacy budgets rather than accuracies for a query.
Finally, we support only certain types of DP mechanisms and definitions in our prototype, specifically those relying on Laplace, whereas Tumult supports more.
Our approach to address these limitations without restricting analysts' interaction with Tumult is to consult the Turbo caches only when the queries exhibit properties we can handle, while resorting to Tumult-based execution when they do not.  As a result, an analyst interacting with a TurboSession will not be restricted in terms of their queries or accuracy demands compared to interacting with a vanilla Tumult session, but Turbo will only conserve privacy budget for those queries that it can handle.

The preceding two approaches for light-touch integration of Turbo into Tumult ensure that our system can be easily adopted.

\heading{Turbo API.}  
The Turbo API is the central component for integrating Turbo into real DP systems. Shown in \F\ref{fig:turbo-tumult-prototype}(b), it consists of two components: functionality that Turbo-lib implements and users invoke to take advantage of its caches (specifically, the {\small \tt run} function) and (2) several classes that users implement to provide Turbo with services it needs from the DP system with which it is integrated.  Turbo needs three types of services from the DP system. First, it needs the ability to extract certain information about the query, such as: the type of aggregation and filter chain; a unique ID for the dataset (or partition or view over the dataset or partition) on which the query is run, as Turbo's state is tied to a dataset/partition/view; and the number of records in that dataset (recall that our design assumes that the dataset size is public information).  This information is supplied by implementing the {\small \tt TurboQuery} interface, which wraps the original, DP-system-specific query structure into one that supplies the necessary information.  For example, our turbo-tumult library wraps query expressions into a TurboQuery with this enhanced functionality.

Second, Turbo is {\em not} a query engine, so it needs the ability to execute a query through the original DP system.  This is provided by implementing the QueryExecutor interface.  A peculiarity of Turbo in this context, which was easy to implement in Tumult but which we anticipate may be non-trivial to implement in other DP systems, is that Turbo needs not only the ability to execute the query in a DP way, but also the ability to execute it without DP.  Recall that Turbo's SV checks compares the histogram-based result to the {\em true result} of invoking the same query on the data without DP.  Turbo thus needs access to this true result, a piece of functionality that typically DP databases do not offer publicly, for good, safety-related reasons.  Still, in Tumult, due to its highly modular structure, we find that this functionality can be implemented without having to modify its code base.  Specifically, turbo-tumult implements {\small \tt QueryExecutor.executeNPQuery(.)} by defining a special type of measurement that does not, in fact, incorporate randomness into its aggregate and which reports as zero the privacy budget being used.  This measurement is executed against the Tumult Core and returns the true result of the query.  In turbo-lib, we take care to only leverage this sensitive result internally during the SV check in a DP way.  Moreover, to optimize query execution in the case that the SV fails, turbo-tumult implements {\small \tt QueryExecutor.executeDPQuery(.)} with the optional ability to reuse a non-private, true result previously obtained for the SV check.  This is achieved by implementing another type of measurement, which, when executed by Tumult Core, will only apply the randomness operation to the given {\small \tt true\_result} and report the appropriate amount of privacy budget to be deducted by the Tumult Core's accountant.

Third, Turbo needs the ability to deduct the privacy budget consumed by the SV reset.  This is supplied by implementing the Turbo PrivacyAccountant interface, with one function: {\small \tt consume}.  In turbo-tumult, we implement this interface by defining a third type of measurement, which does not perform any computation but just consumes privacy.  We believe that DP systems should export this kind of functionality to more naturally support extensions.

\heading{Turbo-lib.}
Turbo-lib implements the Turbo design described in this paper, with some notable restrictions.  First, we do not yet support partitioning in the turbo-lib implementation, though that support exists in our SOSP artifact release, as used in our evaluation.  Second, our implementation only supports count queries presently, although our histograms and exact-match caches can be extended to support other types of linear aggregations, such as sums, averages, standard deviation.
Third, we use Redis to store all state in Turbo, including the exact-match caches, PMW histograms, and SV state.
Redis can be replaced with a persistent, consistent and durable storage service for production use.

\heading{Turbo-sql.}
In addition to incorporating Turbo into the Tumult Analytics engine, we are also creating a basic, standalone, SQL DP database ourselves, which only supports the types of queries that Turbo supports, but which can support streaming and partitioning.  At the time of this writing, the most mature version of this library can be found in our SOSP artifact release, but we are working on a more modular version of this library that presently lacks support for streaming and partitioning.  The full-featured version of this library, which is what we use in our evaluation, receives simple linear SQL queries as strings, parses them, implements the Turbo API to first check for answers to them in the Turbo cache, and execute the queries -- DP through Laplace or non-DP (as needed by Turbo) -- using TimescaleDB, a streaming version of PostgreSQL.

\section{Evaluation}
\label{sec:evaluation}

We evaluate \sysname using the SOSP artifact version of our own, dedicated DP SQL database with Turbo support incorporated in it.
We use two public timeseries datasets -- \covid19 and \citibike\ -- to evaluate \sysname in the three use cases from \S\ref{sec:use-cases}.
Each use case lets us do {\em system-wide evaluation}, answering the critical question: {\em Does \sysname significantly improve privacy budget consumption compared to reasonable baselines for each use case?}
This corresponds to evaluating our \S\ref{sec:goals} design goals {\bf (G5)} and {\bf (G6)}.
In addition, each setting lets us evaluate a different set of caching objects and mechanisms:

\heading{(1) \Nonpartitioneddb:} We configure \sysname with a single \pmwbypass and \exactcache, letting us evaluate the \pmwbypass object, including its \empiricalconvergence and the impact of its heuristic and learning rate parameters.

\heading{(2) \Partitionedstaticdb:} We partition the datasets by time (one partition per week) and configure \sysname with the tree-structured \pmwbypass and \exactcache.  This lets us evaluate the tree-structured cache. 

\heading{(3) \Partitionedstreamingdb:} We configure \sysname with the tree-structured \pmwbypass, \exactcache, and histogram warm-up, letting us evaluate warm-up.

\begin{figure*}[t]
	\subfigure{
		\includegraphics[width=0.5\textwidth]{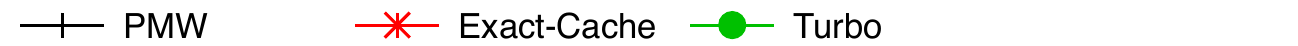}
	}
	\setcounter{subfigure}{0}
	\\ \vspace{-0.7cm}
	\subfigure[\sysname on \covid19 $\kzipf=0$]{
		\includegraphics[width=0.23\textwidth]{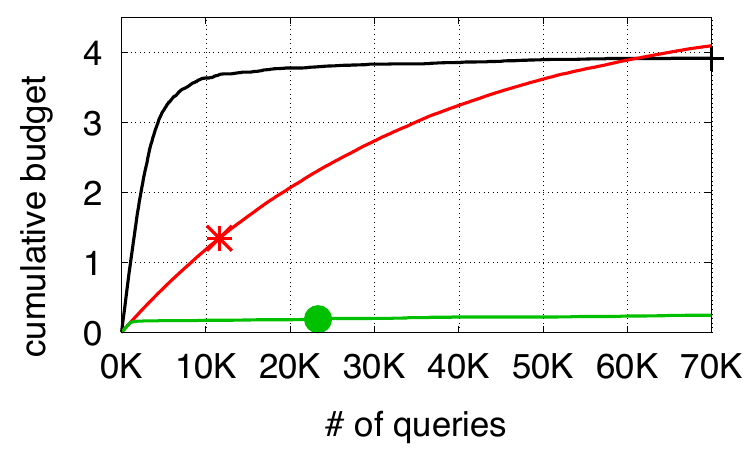}
		\label{fig:monoblock_covid_system_eval_zipf0}
	}
	\hfill
	\subfigure[\sysname on \covid19 $\kzipf=1$]{
		\includegraphics[width=0.23\textwidth]{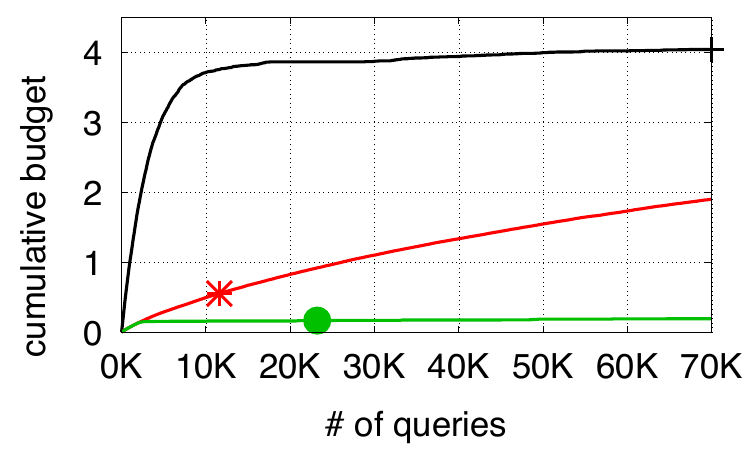}
		\label{fig:monoblock_covid_system_eval_zipf1}
	}
	\hfill
	\subfigure[\sysname on \citibike $\kzipf=0$]{
		\includegraphics[width=0.23\linewidth]{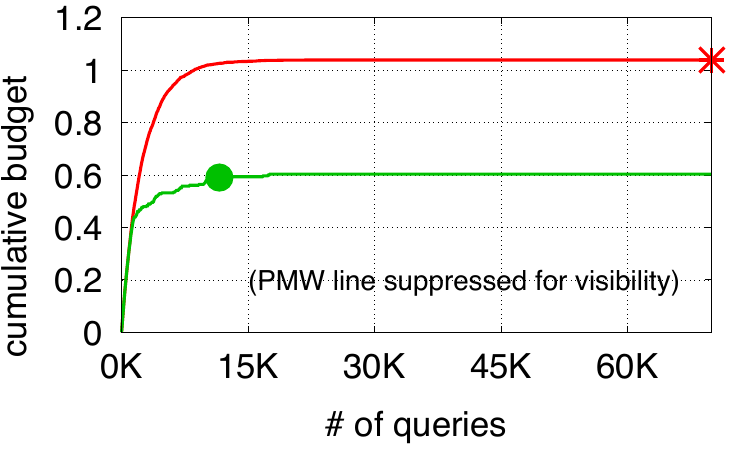}
		\label{fig:monoblock_citibike_system_eval_zipf0}
	}
	\subfigure[\Empiricalconvergence]{
		\includegraphics[width=0.23\linewidth]{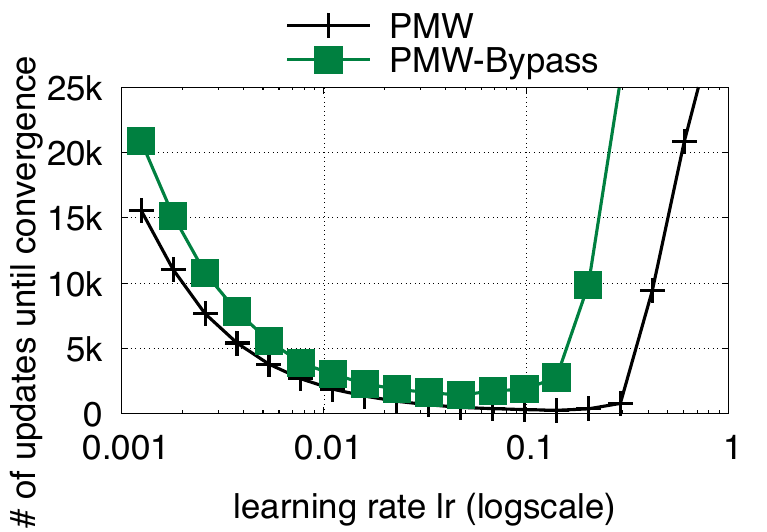}
		\label{fig:monoblock_covid_empirical_convergence}
	}
	\caption{{\bf \Nonpartitioneddb: (a-c) system-wide evaluation (Question 1); (d) \empiricalconvergence for \pmwbypass vs. \pmw (Question 2).} (a-c) \sysname, instantiated with one \pmwbypass and \exactcache, significantly improves budget consumption compared to both baselines. (d) Uses \covid19 $\kzipf=1$. \pmwbypass has similar \empiricalconvergence to \pmw, and both converge faster with much larger $\lr$ than anticipated by \theoreticalconvergence.
	}
	\label{fig:monoblock_eval}
\end{figure*}

\vspace{0.2cm}
As highlighting, our results show that \pmwbypass unleashes the power of \pmw, enhancing privacy budget consumption for linear queries well beyond the conventional approach of using an exact-match cache (goal {\bf (G5)}).
Moreover, \sysname as a whole seamlessly applies to multiple settings, with its novel \pmwbypasstree structure scoring significant benefit for timeseries workloads where database can be partitioned to leverage parallel composition (goal {\bf (G6)}).
Configuration of our objects and mechanisms is straightforward (goal {\bf (G7)}), and we tune them based on \empiricalconvergence rather than theoretical convergence, boosting their practical effectiveness (goal {\bf (G4)}).
Finally, we provide a basic runtime and memory evaluation, which shows that while \sysname performs reasonably for our datasets, further research is needed for larger-domain data.

\subsection{Methodology}
\label{sec:evaluation:methodology}

For each dataset, we create query workloads by (1) generating a pool of linear queries and (2) sampling queries from this pool based on a Zipfian distribution.  \covid19 uses a completely synthetic query pool. \citibike uses a pool based on real-user queries from prior \citibike analyses. We use the former as a microbenchmark, the latter as a macrobenchmark.

\heading{\covid19.}
{\em Dataset:} We take a California dataset of Covid-19 tests from 2020 that provides daily {\em aggregate information} of the number of Covid tests and their positivity rates for various demographic groups defined by age $\times$ gender $\times$ ethnicity. We combine this data with US Census data to generate a synthetic dataset that contains $n=50,426,600$ per-person test records, each with the date and four attributes: positivity, age, gender, and ethnicity. These attributes have domain sizes of 2, 4, 2 and 8, respectively, so the dataset domain size is $N=128$. The dataset spans $50$ weeks, so in partitioned use cases we have up to 50 partitions.
	{\em Query pool:} We create a synthetic and rich pool of correlated queries comprising all possible count queries that can be posed on \covid19.  This gives $34,425$ unique queries, plenty for us to microbenchmark \sysname.

\heading{\citibike.} {\em Dataset:} We take a dataset of NYC bike rentals from 2018-2019, which includes information about individual rides, such as start/end date, start/end geo-location, and renter's gender and age.
The original data is too granular with 4,000 geo-locations and 100 ages, making it impractical for PMWs.
Since all the real-user analyses we found consider the data at coarser granularity (\eg broader locations and age brackets), we group geo-locations into ten neighborhoods and ages into four brackets.
This yields a dataset with $n=21,096,261$ records, domain size $N=604,800$, and spanning 50 weeks.
	{\em Query pool:} We collect a set of pre-existing \citibike analyses created by various individuals and made available on Public Tableau~\cite{tableau}.  An example is here~\cite{tableau-citibike}.  We extract 30 distinct queries, most containing `GROUP BY' statements that we decompose into multiple {\em primitive queries} that can interact with \sysname histograms. This gives us a pool of $2,485$ queries, which is smaller than \covid19's but more realistic and suitable as a macrobenchmark.

\heading{Workload generation.}
As is customary in caching literature~\cite{ycsb,twitter-cache,atikoglu2012workload}, we use a Zipfian distribution to control the skewness of query distribution, which affects hit rates in the exact-match cache.
From a pool of $Q$ queries, a query of type $x \in [1,Q]$ is sampled with probability $\propto x^{-\kzipf}$, where $\kzipf\ge0$ is the parameter that controls skewness.
We evaluate with several $\kzipf$ values but report only results for $\kzipf=0$ (uniform) and $\kzipf=1$ (skewed) for \covid19. For \citibike, we evaluate only for $\kzipf=0$ to avoid reducing the small query pool further with skewed sampling.
For streaming, queries arrive online with arrival times following a Poisson process; they request a window of certain size over recent timestamps.

\heading{Metrics.}
$\bullet$ {\em Average cumulative budget}: the average budget consumed across all partitions.
$\bullet$ {\em Systems metrics:} traditional runtime, process RAM.
$\bullet$ {\em \Empiricalconvergence}: We periodically evaluate the quality of \sysname's histogram by running a validation workload sampled from the same query pool. We measure the accuracy of the histogram as the fraction of queries that are answered with error $\ge \alpha/2$ by the histogram. We define {\em \empiricalconvergence} as the number of histogram updates necessary to reach 90\% validation accuracy.

\heading{Default parameters.}
Unless stated otherwise, we use the following parameter values:
privacy $(\epsilon_G=10, \delta_G=0)$; accuracy $(\alpha=0.05, \beta=0.001)$;
for \covid19: \{learning rate $lr$ starts from $0.25$ and decays to $0.025$, heuristic $(C_0=100,S_0=5)$, external updates $\tau=0.05$\};
for \citibike: \{learning rate $lr=0.5$, heuristic $(C_0=5,S_0=1)$, external updates $\tau=0.01$\}.
\begin{figure}[t]
	\centering
	\subfigure[Impact of heuristic $C_0$ ($S_0=1$)]{
		\includegraphics[width=0.47\linewidth]{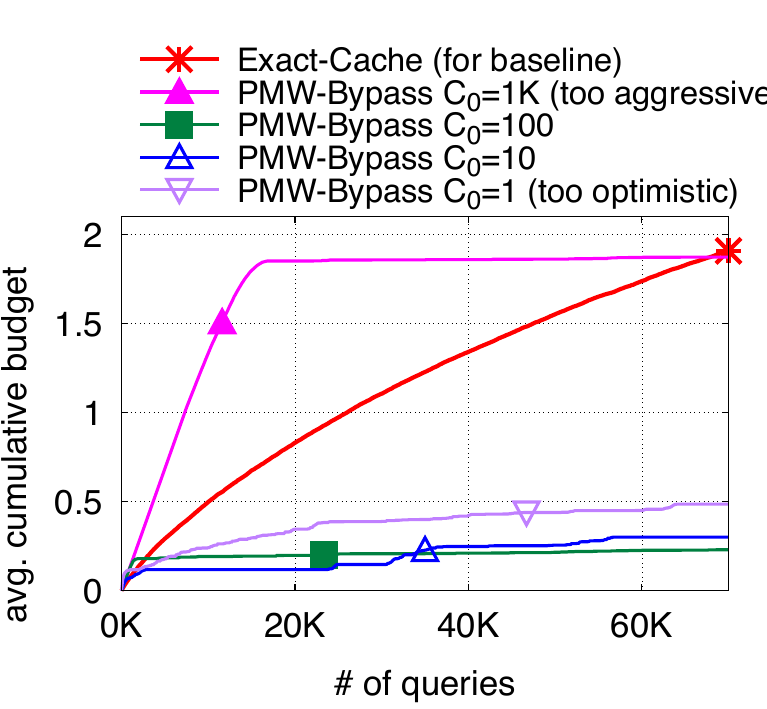}
		\label{fig:monoblock_covid_heuristics}
	}
	\subfigure[Impact of learning rate $\lr$]{
		\includegraphics[width=0.47\linewidth]{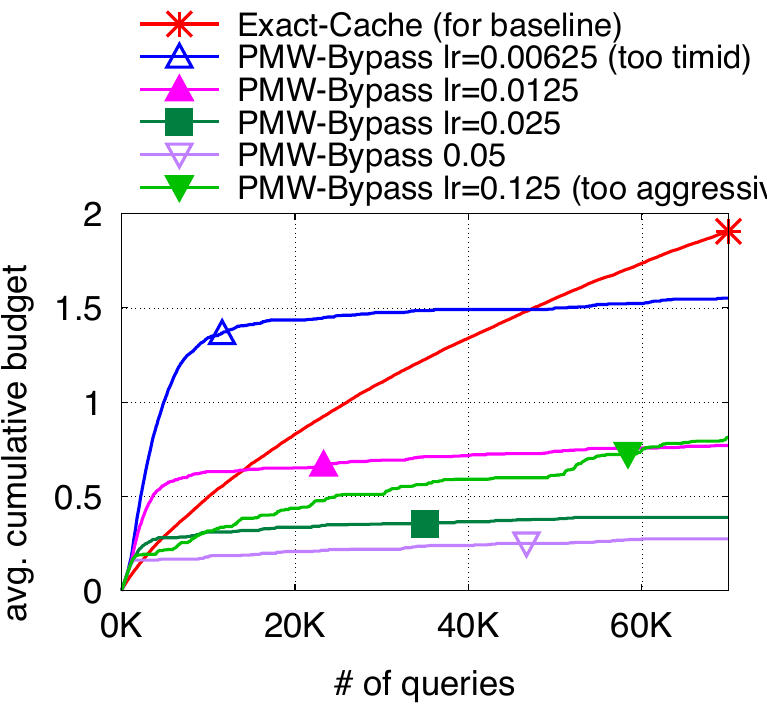}
		\label{fig:monoblock_covid_learning_rates}
	}
	\caption{{\bf Impact of parameters (Question 3).} Uses \covid19 $\kzipf=1$.
		Being too optimistic or pessimistic about the histogram's state (a), or too aggressive or timid in learning from each update (b), give poor performance. 
	} \label{fig:pmwbypass_params_eval}
\end{figure}

\begin{figure*}[tp]
	\begin{minipage}{.70\linewidth}
		\centering
		\subfigure{
			\includegraphics[width=0.7\textwidth]{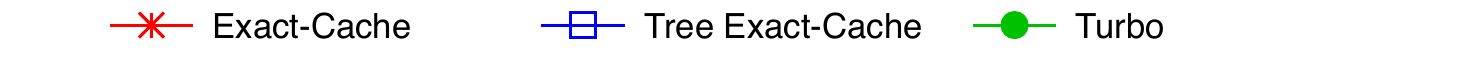}
		}\vspace{-0.5cm}
		\setcounter{subfigure}{0}
		\subfigure[\sysname on \covid19 $\kzipf=0$]{
			\includegraphics[width=0.31\textwidth]{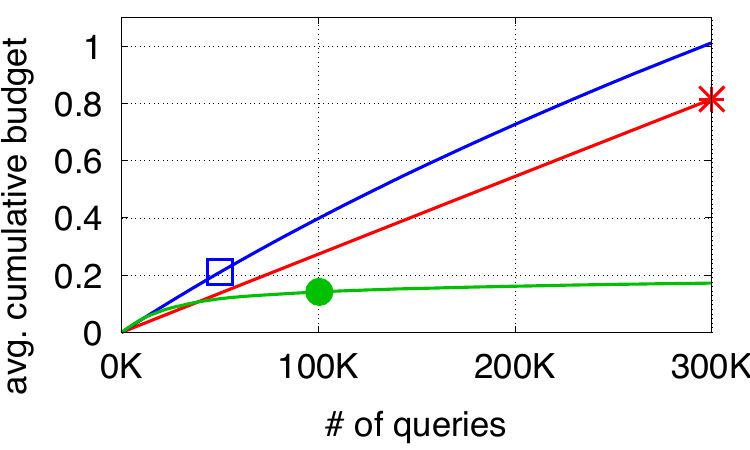}
			\label{fig:multiblock_static_covid_cumulative_budget_consumption_zipf0}
		}
		\subfigure[\sysname on \covid19 $\kzipf=1$]{
			\includegraphics[width=0.31\textwidth]{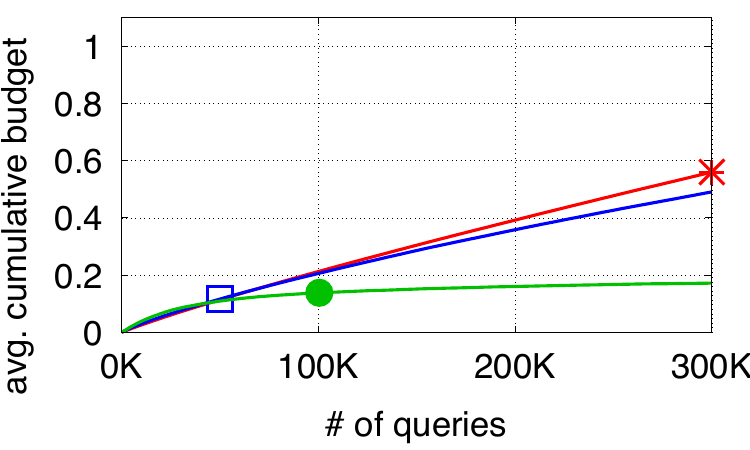}
			\label{fig:multiblock_static_covid_cumulative_budget_consumption_zipf1}
		}
		\subfigure[\sysname on \citibike $\kzipf=0$]{
			\includegraphics[width=0.31\textwidth]{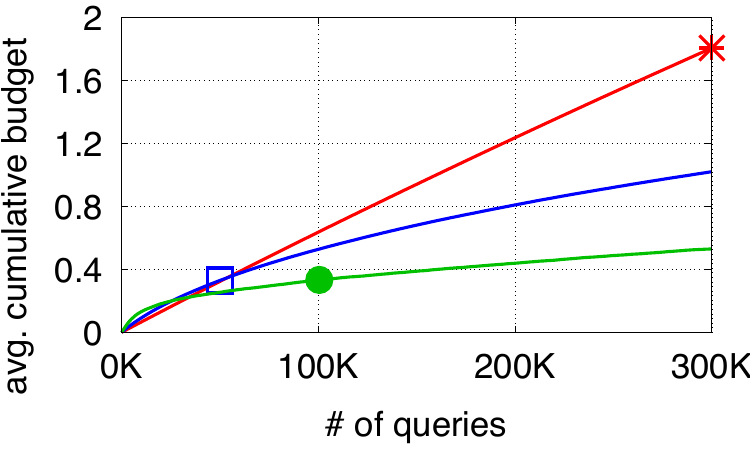}
			\label{fig:multiblock_static_citibike_cumulative_budget_consumption_zipf0}
		}
	\end{minipage}
	\hfill
	\begin{minipage}{.27\linewidth}
		\vspace{-0.5cm}
		\caption{{\bf \Partitionedstaticdb: system-wide evaluation (Question 5).}
			\sysname is instantiated with tree-structured \pmwbypass and \exactcache.
			\sysname significantly improves budget consumption compared to both a single \exactcache and a tree-structured set of Exact-Caches.
		}
		\label{fig:multiblock_static_eval}
	\end{minipage}
\end{figure*}

\subsection{Use Case (1): Non-partitioned Database}
\label{sec:evaluation:non-partitioned-database}

\heading{System-wide evaluation.} {\em Question 1: In a non-partitioned database, does \sysname significantly improve privacy budget consumption compared to vanilla PMW and a simple \exactcache?}
\F\ref{fig:monoblock_covid_system_eval_zipf0}-\ref{fig:monoblock_citibike_system_eval_zipf0} show the cumulative privacy budget used by three workloads as they progress to $70K$ queries. Two workloads correspond to \covid19, one uniform ($\kzipf = 0$) and one skewed ($\kzipf = 1$), and one uniform workload for \citibike.
\sysname surpasses both baselines across all three workloads. The improvement is enormous when compared to vanilla \pmw: $15.9-37.4\times$! \pmw's convergence is rapid but consumes lots of privacy; \sysname uses little privacy during training and then executes queries for free.
Compared to just an \exactcache, the improvement is less dramatic but still significant.
The greatest improvement over \exactcache is seen in the uniform \covid19 workload: $16.7\times$ (\F\ref{fig:monoblock_covid_system_eval_zipf0}). Here, queries are relatively unique, resulting in low hit rate for the \exactcache. That hit rate is higher for the skewed workload (\F\ref{fig:monoblock_covid_system_eval_zipf1}), leaving less room for improvement for \sysname: $9.7\times$ better than \exactcache.
For \citibike (\F\ref{fig:monoblock_citibike_system_eval_zipf0}), the query pool is much smaller ($<2.5K$ queries), resulting in many exact repetitions in a large workload, even if uniform.
Nevertheless, \sysname gives a $1.7\times$ improvement over \exactcache.
And in this workload, \sysname outperforms \pmw by $37.4\times$ (omitted from figure for visualization reasons).
Overall, then, \sysname significantly reduces privacy budget consumption in non-partitioned databases, achieving \nonpartitioneddbimprovement improvement over the best baseline for each workload (goal {\bf (G5)}).

\heading{\pmwbypass evaluation.}
Using \covid19 $\kzipf=1$, we microbenchmark \pmwbypass to understand the behavior of this key \sysname component.
	{\em Question 2: Does \pmwbypass converge similarly to \pmw in practice?}
Through theoretical analysis, we have shown that \pmwbypass achieves similar worst-case convergence to PMW, albeit at slower speed (\S\ref{sec:pmw-bypass}).
\F\ref{fig:monoblock_covid_empirical_convergence} compares the {\em \empiricalconvergence} (defined in \S\ref{sec:evaluation:methodology}) of \pmwbypass vs. \pmw, as a function of the learning rate $lr$.
We make three observations, two of which agree with theory, and the last differs.
First, the results confirm the theory that (1) \pmwbypass and \pmw converge similarly, but (2) for ``good'' values of $lr$, vanilla \pmw converges slightly faster:
e.g., for $lr = 0.025$, \pmwbypass converges after 1853 updates, while PMW after 944. 
Second, as theory suggests, very large values of lr (e.g., $lr \ge 0.4$) impede convergence in practice.
Third, although theoretically, $lr = \alpha/8 = 0.00625$ is optimal for worst-case convergence, and it is commonly hard-coded in PMW protocols~\cite{complexity}, we find that empirically, larger values of $lr$ (e.g., $lr = 0.05$, which is $8\times$ larger) give much faster convergence.
This is true for both \pmw and \pmwbypass, and across all our workloads.
This justifies the need to adapt and tune mechanisms based on not only theoretical but also empirical behavior (goal {\bf (G4)}).

	{\em Question 3: How do \pmwbypass heuristic, learning rate, and external update parameters impact consumed budget?}
We experimented with all parameters 
and found that the two most impactful are (a) $C_0$, the initial threshold for the number of updates each bin involved in a query must have received to use the histogram, and
(b) the learning rate. \F\ref{fig:pmwbypass_params_eval} shows their effects.
	{\em Heuristic $C_0$ (\F\ref{fig:monoblock_covid_heuristics}):} Higher $C_0$ results in a more pessimistic assessment of histogram readiness. If it's too pessimistic ($C_0=1K$), PMW is never used, so we follow a direct Laplace.
If it's too optimistic ($C_0=1$), errors occur too often, and the histogram's training overpays.
$C_0 = 100$ is a good value for this workload.
	{\em Learning rate $\lr$ (\F\ref{fig:monoblock_covid_learning_rates}):} Higher $lr$ leads to more aggressive learning from each update.
Both too aggressive ($lr = 0.125$) and too timid ($lr = 0.00625$) learning slow down convergence. Good values hover around $lr = 0.025$.
Overall, only a few parameters affect performance, and even for those,
performance is relatively stable around good values, making them easy to tune (goal {\bf (G7)}).

{\em Question 4: How does \sysname's adaptive, per-bin heuristic compare to alternatives?}
We experimented with three alternative {$\textproc{IsHistogramReady}$} designs that forgo either (1) the per-bin granular thresholds, or (2) the adaptivity property, or (3) both. We make two observations.
First, the {\em coarse-grained heuristics} consume more privacy budget than the fine-grained heuristics, especially on more skewed workloads, such as $\kzipf=1.5$, which have less diversity so they tend to train histogram bins less uniformly.
For example, a coarse-grained heuristic that uses a histogram-level count of the number of updates, with a threshold $C_0$ to determine when the histogram is ready to receive {\em any} query, consumes at best $0.7$ global privacy budget on a \covid19 workload with $\kzipf=1.5$; this is achieved when $C_0$ is optimally configured to a value of 2070 updates.
In contrast, a fine-grained heuristic, which uses a per-bin update count with a threshold $C_0$ for each bin, consumes at best $0.44$ global privacy budget, achieved when $C_0$ is set to 160 updates.
Second, the {\em adaptive heuristics} consume similar budget as the optimally-configured, non-adaptive ones, but the former are much easier to configure, as they offer stable performance around wide ranges of the $C_0$ parameter.
For example, when $C_0$ varies in range $[20,200]$, the non-adaptive per-bin heuristic's budget consumption varies in range $[0.44,0.81]$ for the $\kzipf=1.5$ workload, and in range $[0.31,0.76]$ for $\kzipf=1$ workload.
In contrast, \sysname's adaptive, per-bin heuristic's budget consumption varies in much tighter ranges under the same circumstances: $[0.44,0.52]$ and $[0.28,0.48]$ for the $\kzipf=1.5$ and $\kzipf=1$ workload, respectively.
Thus, \sysname's heuristic is the best of these options.

\subsection{Use Case (2): Partitioned Static Database}
\label{sec:evaluation:partitioned-static-database}

\heading{System-wide evaluation.}  {\em Question 5: In a partitioned static database, does \sysname significantly improve privacy budget consumption, compared to a single \exactcache and a tree-structured set of Exact-Caches?}
We divide each database into 50 partitions and select a random contiguous window of 1 to 50 partitions for each query.
We adjust the $(C_0, S_0)$ heuristic parameters to $(50, 1)$ for Covid and $(1,1)$ for \citibike.
\F\ref{fig:multiblock_static_covid_cumulative_budget_consumption_zipf0}-\ref{fig:multiblock_static_citibike_cumulative_budget_consumption_zipf0} show the average budget consumed per partition up to 300K queries. Compared to the static case, \sysname can now support more queries under $\epsilon_G = 10$ thanks to parallel composition: each query only consumes privacy from the accessed partitions.
\sysname further divides privacy budget consumption by \partitionedstaticdbimprovement compared to the best-performing baseline for each workload, demonstrating its effectiveness as a caching strategy for the static partitioned use case.

\begin{figure*}[tp]
	\centering
	\begin{minipage}{.745\linewidth}
		\subfigure{
			\includegraphics[width=\textwidth]{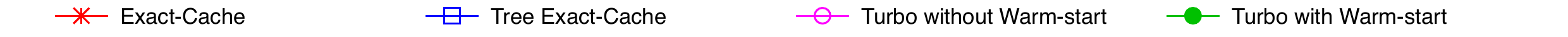}
		}\vspace{-0.4cm}
		\setcounter{subfigure}{0}
		\subfigure[\sysname on \covid19 $\kzipf=0$]{
			\includegraphics[width=0.31\textwidth]{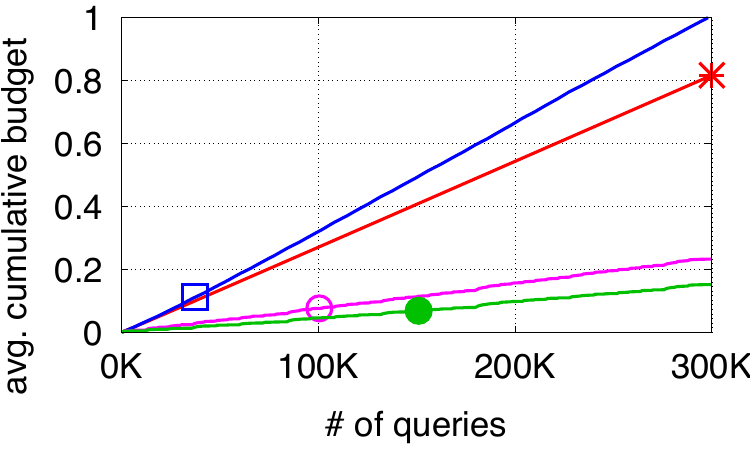}
			\label{fig:multiblock_streaming_covid_cumulative_budget_consumption_zipf0}
		}
		\subfigure[\sysname on \covid19 $\kzipf=1$]{
			\includegraphics[width=0.31\textwidth]{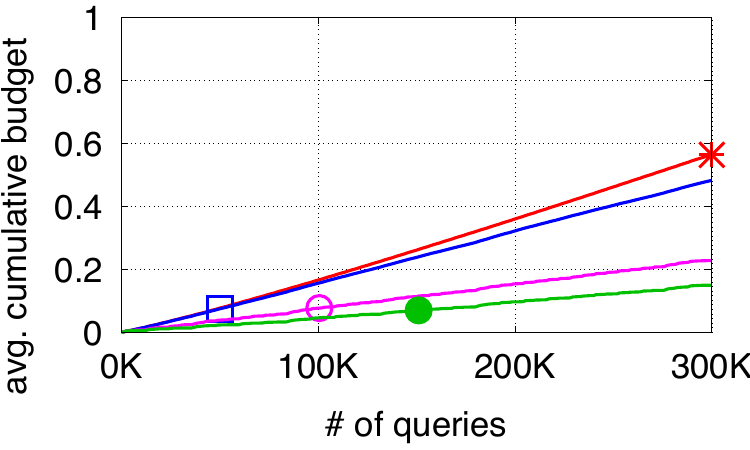}
			\label{fig:multiblock_streaming_covid_cumulative_budget_consumption_zipf1}
		}
		\subfigure[\sysname on \citibike $\kzipf=0$]{
			\includegraphics[width=0.31\textwidth]{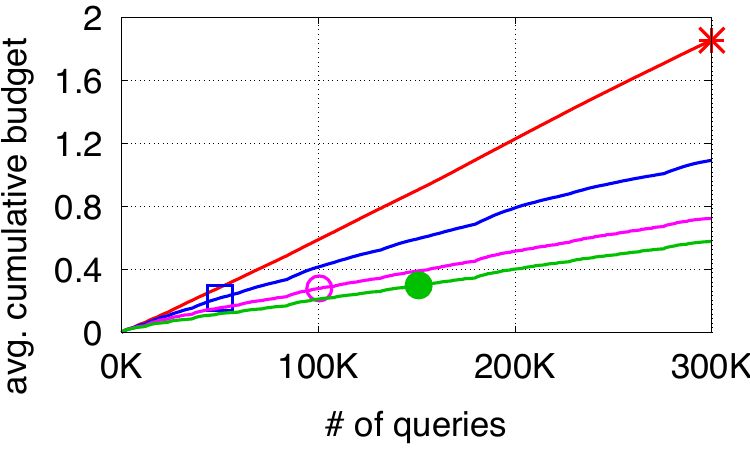}
			\label{fig:multiblock_streaming_citibike_cumulative_budget_consumption_zipf0}
		}
	\end{minipage}
	\hfill
	\begin{minipage}{.235\linewidth}
		\subfigure[Runtime evaluation]{
			\includegraphics[width=\textwidth]{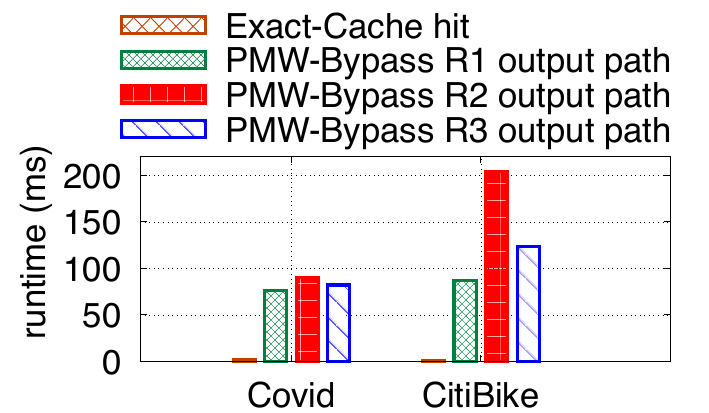}
			\label{fig:runtime_eval}
		}
	\end{minipage}
	\caption{{\bf (a-c) \Partitionedstreamingdb: system-wide consumed budget (Question 7); (d) \pmwbypass runtime in non-partitioned setting (Question 8).}
		(a-c) \sysname is instantiated with tree-structured \pmwbypass and \exactcache, with and without warm-start.  
		(d) Uses \covid19 $\kzipf=1$ and one \exactcache and \pmwbypass. Shows execution runtime for different execution paths. Most expensive is when the SV test fails.  
	}
	\label{fig:multiblock_streaming_system_eval}
\end{figure*}

\heading{\bf Tree structure evaluation.}
{\em Question 6: When does the tree structure for histograms outperform a flat structure that maintains one histogram per partition?}
We vary the average size of the windows requested by queries from $1$ to $50$ partitions based on a Gaussian distribution with std-dev $5$.
We find the tree structure for histograms is beneficial when queries tend to request more partitions ($25$ partitions or more). Because the tree structure maintains more histograms than the flat structure, it fragments the query workload more, resulting in fewer histogram updates per histogram and more use of direct-Laplace. The tree's advantage in combining fewer results makes up for this privacy overhead caused by histogram maintenance when queries tend to request larger windows of partitions, while the linear structure is more justified when queries tend to request smaller windows of partitions.

\subsection{Use Case (3): Partitioned Streaming Database}
\label{sec:evaluation:partitioned-streaming-database}

\heading{System-wide evaluation.}  {\em Question 7: In streaming databases partitioned by time, does \sysname significantly improve privacy budget consumption compared to baselines?  Does warm-start help?}
\F\ref{fig:multiblock_streaming_covid_cumulative_budget_consumption_zipf0}-\ref{fig:multiblock_streaming_citibike_cumulative_budget_consumption_zipf0} show \sysname's budget consumption compared to the baselines.
The experiments simulate a streaming database, where partitions arrive over time and queries request the latest $P$ partitions, with $P$ chosen uniformly at random between $1$ and the number of available partitions.
\sysname outperforms both baselines significantly for all workloads, particularly when warm-start is enabled. Without warm-start, \sysname improves performance by $1.5-3.5\times$ at the end of the workload.
With warm-start, \sysname gives \partitionedstreamingdbimprovement improvement over the best baseline for each workload, showing its effectiveness for the streaming use case.
When there is a large variety of unique queries the \exactcachetree has a significantly better hit-rate than the \exactcache baseline and performs better (\F\ref{fig:multiblock_streaming_covid_cumulative_budget_consumption_zipf0}).
In \F\ref{fig:multiblock_streaming_covid_cumulative_budget_consumption_zipf1} and \ref{fig:multiblock_streaming_citibike_cumulative_budget_consumption_zipf0} the query pool is considerably smaller.
Both baselines have a good enough hit-rate while the \exactcachetree  needs to consume more privacy budget to compensate for the aggregation error which makes it perform worse.
This concludes our evaluation across use cases ({goal \bf (G6)}).

\subsection{Runtime and Memory Evaluation}
\label{sec:evaluation:runtime-memory}

{\em Question 8: What are \sysname's runtime and memory bottlenecks?}
We evaluate \sysname's runtime and memory consumption to identify areas of improvement. 
\F\ref{fig:runtime_eval} shows the average runtime of \sysname's main execution paths in a non-partitioned database.  The \exactcache hit path is the cheapest and the other paths are more expensive. 
Histogram operations are the bottlenecks in \citibike due to the larger domain size ($N$), while query execution in TimescaleDB is the bottleneck in \covid19 due to the larger database size ($n$). 
The $R1$ path is similar across the two datasets because their distinct bottlenecks compensate. 
Failing the SV check (output path $R2$) is the costliest path for both datasets due to the extra operations needed to update the heuristic's per-bin thresholds. 
We also conduct an experiment in the partitioned streaming case and find the same bottlenecks: TimescaleDB for \covid19, histogram operations for \citibike. 
Finally, we report \sysname's memory consumption in the streaming case with 50 partitions: $5.21$MB for \covid19 and $1.43$GB for \citibike.  For context, the raw datasets occupy on disk 600MB and 795MB, respectively.  Thus, \sysname's memory overhead is significant and it is caused by the PMWs.  The next section discusses this limitation and proposes potential directions to address it.

\section{Discussion}
\label{sec:limitations}

We discuss several of \sysname's strengths and weaknesses.
\sysname provides benefits when queries overlap in the data they access, i.e., new queries access histogram bins that have been accessed by past queries.  The functions computed atop these bins can differ among queries (e.g., the new query can compute an average while all the past ones computed count fractions).  If there is no data overlap in the queries, then \sysname does not give any benefit and comes with memory/computational costs.  This is typical for caching systems: they only help if the workload has some level of locality.

A key strength in \sysname is its support for dynamic workloads, both new queries and new data arriving in the system.
First, \sysname adapts seamlessly to changing queries. In the worst case, the new queries will access completely ``untrained'' regions within a histogram. Our heuristic will detect this and trigger a new cycle of external updates. In more moderate cases, the workload will touch a mix of ``trained'' and ``untrained'' regions. This will yield a mix of hits and misses in the heuristic, and Turbo will use just the right amount of privacy budget to adapt to these slower workload changes.
Second, thanks to histogram warm-start, \sysname adapts to new data partitions arriving into the system with minimal privacy budget consumption: as new partitions arrive, their histograms are initialized from past ones and then fine-tuned for the new data by a few external updates.  This way, the new histograms will quickly start serving query answers for free, conserving privacy budget.
Still, there is a limitation: while we support new data arriving in the system, we do not support updates on past data; such updates would result in our heuristics predicting less accurately when the histogram can answer a query, and thus in more expensive SV failures.  

By far, \sysname's biggest limitation is the memory consumed to maintain the PMW histograms.
Each histogram is a RedisAI vector whose size grows with data domain size $N$, i.e., exponentially in data domain dimension $d$ ($N$ and $d$ are defined in Section~\ref{sec:notation}).
With $T$ partitions and $k$ queries, \sysname maintains a binary tree of such histograms, which means it stores $\approx 2 T N$ scalar values.  By comparison, the Tree Exact-Cache baseline stores at most $\log(T) k$ scalar values, a much lower memory consumption.
This impacts not only the scale of the datasets that can be handled with \sysname, but also the runtime performance of \sysname-mediated queries.
Indeed, as shown in the preceding section, histogram operations for \citibike are the bottleneck in runtime due to the relatively high domain size.
Some techniques have previously been proposed to address this rather fundamental challenge for PMW~\cite{mwem}.
However, for even larger-scale deployments, we believe that it will be worth considering PMW alternatives that may not offer as compelling convergence guarantees as PMW but which are much more lightweight.
One example may be the relaxed adaptive projection (RAP)~\cite{relaxed_adaptive_projection},
which builds a lightweight representation of the dataset by learning a small subset of representative data points using gradient-descent.
One would have to be willing to forfeit the theoretical convergence guarantees to use this mechanism, and to develop an adaptive version of RAP to support realistic systems settings involving dynamic workloads and data.
Even so, some of the core concepts we have proposed in this paper may transfer to this new design, including passing RAP-based estimations through an SV to ensure result accuracy while incorporating a heuristic-based bypass to avoid expensive failures in the SV.

We also touch on several potential vulnerabilities.
First, an adversary may craft queries that consume budget by generating cache misses.
The convergence proofs in \S\ref{sec:appendix:bounded_budget} provide a bound on how much such queries can affect budget consumption when a straightforward cutoff parameter is configured upfront.
Second, response time can be a side-channel, which we leave out of scope but should be addressed in the future.
Third, $n$, the number of elements in the database (or in each partition), is considered public knowledge.
This can leak information and should be addressed by consuming some of the budget to compute $n$ privately, as done in~\cite{sage}.

Regarding integration of Turbo with a real system, Tumult, we find that it can be done with ease, thanks to Tumult Core's extensible measurement API.  We anticipate that such integration will not be as easy or ``light touch'' in other DP systems we have seen, and in general we see a gap in the core primitives that DP systems (SQL or not) should implement to support extensions such as Turbo; these might include providing direct access to the privacy accountant, decoupling the accountant from the query executor, and others. We encourage the community to work to articulate this set of key primitives, which we suspect will be useful in other extensions beyond Turbo.

\section{Related Work}
\label{sec:related-work}

This paper presents the first design, implementation, and evaluation for a {\em general, effective, and accurate DP-caching system for interactive DP-SQL systems}.
In computer systems, caching is a heavily-explored topic, with numerous algorithms and implementations~\cite{lhd, twitter_caching, web_cache}, some pervasively used in processors, operating systems, databases, and more.  However, traditional forms of caching differ significantly from DP caching, justifying the need for a specialized approach for DP.  The primary purposes of traditional caching are to conserve CPU and to improve throughput and latency; for these purposes, existing caches can be readily reused in DP systems.  However, DP caching aims to conserve privacy budget, which requires a new design to be truly effective.  For example, layering Redis on a DP database to cache query results would save CPU, but for privacy it would be equivalent to the ``Exact-Cache'' baseline that our evaluation shows is less effective than \sysname.  This paper thus builds upon general traditional caching concepts -- such as the two-layer design, the principle of generality in supporting multiple workloads -- but develops a cache specialized in conserving DP budget.

To our knowledge, no existing DP system incorporates such a specialized caching system.
Most DP systems do not incorporate caching capabilities at all~\cite{pinq,flex,tumult,rogers2020linkedin,orchard,plume};~\cite{zetasql} explicitly leaves the design of an effective DP cache for future work.
Some DP systems incorporate what amounts to an Exact-Cache by deterministically generating the same noise upon the arrival of the same query.
Three systems consider more sophisticated mechanisms for DP result reuse: PrivateSQL~\cite{privatesql}, Chorus~\cite{chorus}, and CacheDP~\cite{cachedp}. But the result reuse components in these systems suffer from such significant limitations that they cannot be considered general and effective caching designs.
    {\bf PrivateSQL}~\cite{privatesql} takes a batch of ``representative'' offline queries and precomputes a private synopsis that answers them all. If new queries arrive (online), PrivateSQL uses the synopsis to answer them in a best-effort way, without accuracy guarantees. It does not learn on-the-fly from them, so it is unsuited for online workloads and does not support data streams.
    {\bf Chorus}~\cite{chorus} provides a trivialized implementation of MWEM, a variant of PMW, however the implementation only works for databases with a {\em single attribute}. 
The paper does not evaluate the MWEM-based implementation, nor integrates it as a caching layer.
    {\bf CacheDP}~\cite{cachedp} is an interactive DP query engine and has a built-in DP cache that answers queries using the Matrix Mechanism~\cite{matrix_mechanism}. Our experience with the CacheDP code suggests that it is not a general, effective, or accurate caching layer for DP databases.
First, CacheDP's implementation only scales to a few attributes and does not support parallel composition on data partitions; this suggests that it is not general enough to support a variety of workloads.
Second, the ``Tree Exact-Cache'' baseline with which we compare in evaluation matches, to our understanding, the CacheDP design while scaling to the higher-dimension datasets and streaming workloads we evaluate against.
Our evaluation shows Turbo more effective than Tree Exact-Cache.

While DP caching are under-explored in systems, the topic of optimizing global privacy budget for a query workload is heavily explored in theory.
Approaches include generating synthetic datasets or histograms that can answer certain classes of queries, such as linear queries, with accuracy guarantees and no further privacy consumption~\cite{pmw,smalldb,VietriTBSW20,relaxed_adaptive_projection,mwem,pmwpub}; and optimizing privacy consumption over a batch of queries by adapting the noise distribution to properties of the queries~\cite{matrix_mechanism, hdmm, dawa}.
Apart from PMW~\cite{pmw}, all these methods operate in the offline setting, where queries are known upfront.
This setting is unrealistic, as discussed in \S\ref{sec:use-cases}.

All of the theory works cited above, including PMW, suffer from another limitation: they operate on static datasets and do not support new data arriving into the system.
PMWG~\cite{pmwg} is an extension of PMW for dynamic ``growing'' databases, but operates in a setting where all queries request the {\em entire database}.
This precludes the use of parallel composition for queries that access less than the entire database, such as queries over windows of time.
Other algorithms focus on continuously releasing specific statistics over a stream, such as the streaming counter~\cite{chan2010} that inspired our tree structure, and extensions to top-k and histogram queries~\cite{streaming_histograms}.
These works do not support arbitrary linear queries, and they answer all predefined queries at every time step  while we only pay budget for queries that are actually posed by analysts.

\section{Conclusion}
\label{sec:conclusions}

\sysname is a caching layer for differentially-private databases that increases the number of linear queries that can be answered accurately with a fixed privacy guarantee. It employs a PMW, which learns a histogram representation of the dataset from prior query results and can answer future linear queries at no additional privacy cost once it has converged. To enhance the practical effectiveness of PMWs, we bypass them during the privacy-expensive training phase and only switch to them once they are ready. This transforms PMWs from ineffective to very effective compared to simpler cache designs.
Moreover, \sysname includes a tree-structured set of histograms that supports timeseries and streaming use cases, taking advantage of fine-grained privacy budget accounting and warm-starting opportunities to further increase the number of answered queries.


\section{Acknowledgments}
\label{sec:acks}

We thank our shepherd Andreas Haeberlen and the anonymous reviewers.
We thank Junfeng Yang for feedback throughout the project.
We thank Heeyun Kim, Hailey Onweller, Sally Wang, Yucheng Wu, Chris Yoon for indirect contributions to prototype and evaluation.
The work was supported by NSF EEC-2133516, CNS-2104292, Sloan, Microsoft, and Google faculty fellowships, and Google Cloud credits.  


\bibliographystyle{abbrv}
\bibliography{bib/abbrev,bib/conferences,bib/refs}

\clearpage \newpage

\setcounter{page}{1}
\appendix
\noindent {\em Note: This appendix has not been peer-reviewed.}
\section{Theorems and proofs}

\subsection{Notation}
\label{sec:appendix:notation}

In this section we introduce the following notation, in addition to the notation from Section \ref{sec:detailed-design}:
\begin{itemize}
    \item     For two distributions $p, h$ over $\mathcal{X}$, we note $D(p||h)$ their relative entropy:

          $$D(p||h) := \sum_x p(x)\ln (p(x)/h(x)) $$

    \item For a linear query $q$, $q(x)\in[0,1]$ is the result of $q$ on a datapoint $x \in \mathcal{X}$. For a histogram or distribution $h$ we note $q \cdot h$ the dot product:
          $$q \cdot h := \sum_{x \in \mathcal{X}} q(x)h(x)$$

          This is also the result of the query $q$ on a {\em normalized} database histogram $h$, so with a slight abuse of notation we alternatively write $q(h)$ for $q \cdot h$.
    \item Given a true distribution $p$, an estimate $h$ and a query $q$, we generally note $q^* := q \cdot p$ the true value of the query, $q(h)$ the estimate returned by the histogram, and $\tilde q$ a random variable denoting the answer to $q$ returned by a randomized algorithm such as PMW.
    \item We use $\eta$ as a shorthand for $\lr$ and $\eta_i$ for variable learning rates, where $i$ is the index of an update.
    \item For streaming databases, we use the standard definition of DP on streams \cite{streaming_counter, pmwg}: two streams of rows $\mathcal{X}^\mathbb{N}$ are adjacent if they differ exactly at one time (index) $t \in \mathbb{N}$.
          Note that since time is a public attribute we can group indices by timestamp.
    \item In \A\ref{alg:pmw-bypass}, we note $X$ the noise added to the threshold at Line~\ref{line:sv_reset}, $Y$ the noise added to the error at Line~\ref{line:R2}, $Z$ the noise potentially added to the true result in the PMW branch at Line~\ref{line:R2} and $Z'$ the noise added to the true result in the Bypass branch at Line~\ref{line:R3}.
          We have $X \sim Y \sim Z \sim Z' \sim Lap(b)$ with $b = 1/\epsilon n$.
\end{itemize}

\subsection{\pmwbypass}
\label{sec:appendix:pmwbypass}

\begin{theorem} \pmwbypass preserves $\epsilon_G$-DP, for a global privacy budget set upfront in the $\textproc{PrivacyAccountant}$.
    \label{thm:pmwbypass:privacy}
\end{theorem}
\begin{proof}
    First, we clarify that each call to \textsc{PrivacyAccount-}\textsc{ant.pay} in Algorithm \ref{alg:pmw-bypass} is a call to a privacy filter \cite{rogers2016privacy} operating with pure differential privacy and basic composition.
    More precisely, we use a type of privacy filter that is suitable for interactive mechanisms such as the SV protocol. We provide a definition in \S\ref{sec:appendix:filters} and a privacy proof in \Thm\ref{thm:rdp_filter_adaptive_concurrent_corollary}.
    The filter takes a sequence of adaptively chosen DP mechanisms. At each turn $i \in \mathbb{N}$, it receives a new mechanism with budget $\epsilon_i$ and uses a stopping rule to decide whether it should run the mechanism.
    In our case, the filter stops answering when $\sum_{j=0}^i \epsilon_j > \epsilon_G$, where $\epsilon_G$ is the predetermined global privacy budget.

    Second, we show that each of the three calls to $\textproc{pay}$ spends the right budget before running each DP mechanism.
    \begin{enumerate}
        \item We use the sparse vector mechanism from \cite{lyu_sv} on queries with sensitivity $\Delta = 1/n$, with cut-off parameter $c=1$ and with $\epsilon_1 = \epsilon, \epsilon_2 = 2\epsilon, \epsilon_3 = 0$. In this case, each SV mechanism is $3\epsilon$-DP.
        \item After a hard query, we initialize a new SV mechanism (for $3\epsilon$) followed by a Laplace mechanism \cite{dwork2014algorithmic} scaled to the right sensitivity (costing $\epsilon$). The composition of these two steps is $4\epsilon$-DP, thanks to the basic composition theorem for pure DP \cite{dwork2014algorithmic}.
        \item In the bypass branch, the Laplace mechanism is also $\epsilon$-DP.
    \end{enumerate}
\end{proof}

\begin{lemma}[Per-query accuracy for vanilla PMW]
    \label{thm:pmw-accuracy}
    \label{lem:pmw-accuracy}
    Consider a vanilla PMW, \ie \A\ref{alg:pmw-bypass} where
    \textsc{Heuristic.IsHisto-}\textsc{gramReady} always return \textsc{True}, that receives a query with true answer $q^*$ and estimate $q(h)$.
    Using $X,Y,Z$ introduced in \S\ref{sec:appendix:notation}, we note $S := \mathds{1}[ |q(h) - q^*| + Y < \alpha/2 + X]$, i.e. $S = 1$ if the sparse vector thinks the query is accurate, otherwise $S = 0$.
    Then the error of this PMW answer $\tilde q$ is $\tilde q - q^* = (q(h) - q^*) S + Z (1 - S)$, and:

    $$
        \Pr[|\tilde q - q^*| > \alpha] < \exp(-\alpha n \epsilon) +  (\frac{1}{2} + \frac{\alpha n\epsilon}{8})\exp(-\frac{\alpha n\epsilon}{2})
    $$

\end{lemma}

\begin{proof}
    Consider two cases depending on the value of the PMW estimate:
    \begin{itemize}
        \item If $|q(h) - q^*| \le \alpha$, then $|\tilde q - q^*| > \alpha \iff S = 0 \cap |Z| > \alpha$. Hence $\Pr[|\tilde q - q^*| > \alpha] \le \exp(-\alpha/b)$.
        \item Now suppose $|q(h) - q^*| > \alpha$. We have $|\tilde q - q^*| > \alpha \iff (S = 0 \cap |Z| > \alpha) \sqcup (S = 1)$, so $\Pr[|\tilde q - q^*| > \alpha] \le \Pr[|Z| > \alpha] + \Pr[|q(h) - q^*| + Y < \alpha/2 + X] \le \Pr[|Z| > \alpha] + \Pr[\alpha + Y < \alpha/2 + X] = \Pr[|Z| > \alpha] + \Pr[\alpha/2 < Y' + X]$ with $Y' = -Y \sim Lap(b)$.

              Fix $z >0$. Note $p$ the pdf of $\lap{b}$. The pdf of the convolution of $X$ and $Y'$ is:

              \begin{align*}
                  p_{X+Y'}(z) & = \int_{-\infty}^{+\infty} p(z-t)p(t)dt                                                       \\
                              & = \frac{1}{4b^2} [ \int_{-\infty}^0 e^{-(z-t)/b}e^{t/b} + \int_{0}^z e^{-(z-t)/b}e^{-t/b}     \\
                              & +\int_{z}^\infty e^{(z-t)/b}e^{-t/b}]                                                         \\                                                                                                  \\
                              & = \frac{1}{4b^2} \left[ e^{-z/b}\frac{b}{2} + ze^{-z/b} + e^{z/b} \frac{b}{2}e^{-2z/b}\right] \\
                              & = \frac{ze^{-z/b}}{4b^2} + \frac{e^{-z/b}}{4b}                                                \\
              \end{align*}

              Hence:

              \begin{align*}
                  \Pr[X + Y' > \alpha/2] & = \int_{\alpha/2}^{\infty} \left(\frac{ze^{-z/b}}{4b^2} + \frac{e^{-z/b}}{4b}\right)dz                         \\
                                         & = \frac{1}{4}\exp(-\frac{\alpha}{2b}) + \frac{1}{4b} \left[-z e^{-z/b} - be^{-z/b} \right]^{\infty}_{\alpha/2} \\
                                         & = \frac{1}{4}\exp(-\frac{\alpha}{2b}) + \frac{1}{4b}(\frac{\alpha}{2} + b) \exp(-\frac{\alpha}{2b})            \\
                                         & = (\frac{1}{2} + \frac{\alpha}{8b})\exp(-\frac{\alpha}{2b})
              \end{align*}

    \end{itemize}

    In both cases we have:
    $$\Pr[|\tilde q - q^*| > \alpha] \le \exp(-\alpha/b) +  (\frac{1}{2} + \frac{\alpha}{8b})\exp(-\frac{\alpha}{2b})$$

\end{proof}

\begin{theorem}
    \label{thm:pmwbypass:accuracy}
    If $\epsilon = \frac{4\ln(1/\beta)}{n\alpha}$, \pmwbypass is $(\alpha,\beta)$-accurate for each query it answers.

    We can also achieve $(\alpha,\beta)$-accuracy with the following slightly smaller $\epsilon$, computable with a binary search:

    $$\min \{\epsilon > 0 : \exp(-\alpha n \epsilon) +  (\frac{1}{2} + \frac{\alpha n\epsilon}{8})\exp(-\frac{\alpha n\epsilon}{2}) \le \beta \}$$

\end{theorem}
\begin{proof}
    Consider an incoming query $q$. We show that the output $\tilde q$ of each branch of \A\ref{alg:pmw-bypass} respects the accuracy guarantees.

    First, we simplify the expression for $\epsilon$ from Lemma \ref{lem:pmw-accuracy}. Since $1 + x \le \exp(x)$ for $x>0$, for any $\epsilon > 0$ we have $\exp(-\alpha \epsilon n ) + \frac{1}{2}(1 + \frac{\alpha\epsilon n}{4})\exp(-\frac{\alpha\epsilon n}{2}) \le \exp(-\alpha\epsilon n) +$\\
    $\frac{1}{2} \exp(-\alpha\epsilon n/4) \le \exp(-\alpha\epsilon n/4)$.
    Hence, taking $\epsilon = \frac{4\ln(1/\beta)}{n\alpha}$ gives $\exp(-\alpha \epsilon n ) + \frac{1}{2}(1 + \frac{\alpha\epsilon n}{4})\exp(-\frac{\alpha\epsilon n}{2}) \le \beta$.

    \begin{itemize}
        \item If the heuristic routes the query to the regular PMW branch, Lemma \ref{lem:pmw-accuracy} gives $\Pr[|\tilde q - q^*| > \alpha] < \beta$.
        \item If the heuristic routes the query to the bypass branch, we have $\tilde q = q(h) + Z'$ with $Z' \sim \lap{1/\epsilon n}$. By a Laplace tail bound, $\Pr[|\tilde q - q^*| > \alpha] < \exp{(-\alpha n \epsilon)} < \exp(-\alpha \epsilon n ) + \frac{1}{2}(1 + \frac{\alpha\epsilon n}{4})\exp(-\frac{\alpha\epsilon n}{2}) \le \beta$.
    \end{itemize}
\end{proof}

\begin{theorem}[Convergence of \pmwbypass (\A\ref{alg:pmw-bypass}), adapted from \cite{pmw}]
    \label{thm:convergence-detailed}
    \label{thm:pmwbypass:convergence}

    Consider a histogram sequence $(h_t)$ with $h_0$ the uniform distribution over $\mathcal{X}$. The histogram approximates a true distribution $p$.
    At time $t \ge 0$, we receive a linear query $q_t$ and we perform a multiplicative weight update with learning rate $\eta$. This happens in one of two cases:
    \begin{itemize}
        \item If the query was answered by the PMW branch, and SV found it hard; or
        \item If the query was answered by the Bypass branch, and the noisy answer was far enough from the histogram
    \end{itemize}

    Let $\betaconv \in (0,1)$ be a parameter for per-update failure probability.
    Suppose that $\eta/2\alpha + \ln(1/\betaconv)/n \epsilon \alpha < \tau < 1/2$.
    Then:

    \begin{enumerate}
        \item After each update we have, with probability $1-\betaconv$:

              $$D(p\|h_{t+1}) - D(p\|h_t)\le- \eta (\tau \alpha - \ln(1/\betaconv)/n\epsilon) + \eta^2/2$$

        \item Moreover, for $k$ an upper bound on the number of queries, with probability $1-k\betaconv$ we perform at most $\frac{\ln |\mathcal{X}|}{ \eta(\tau\alpha - \ln(1/\betaconv)/n \epsilon-\eta/2)}$ updates. This is about $1/2\tau$ times more updates than a \pmw with comparable parameters, which only performs at most $\frac{\ln |\mathcal{X}|}{ \eta(\alpha/2- \ln(1/\betaconv)/n\epsilon-\eta/2)}$ updates.
        \item If we set $\betaconv = \exp{(-\tau n\alpha\epsilon/2)}$, the left condition on $\tau$ becomes $\eta/2\alpha + \tau/2 < \tau$ \ie $\eta/\alpha < \tau$, and we perform at most $\frac{\ln |\mathcal{X}|}{ \eta(\tau\alpha - \eta)/2}$ updates. Moreover, since $\epsilon \le \frac{4\ln(1/\beta)}{n\alpha}$ and $\tau \le 1/2$, we have $\betaconv \le \beta$.
        \item In particular, if we set $\eta = \alpha/8$ and $\tau = \frac{1}{4}$ we have $1/2 > \tau > \eta/2\alpha + \ln(1/\betaconv)/n \epsilon$ when $n$ is large enough, and \whp we perform at most $O(\frac{\ln |\mathcal{X}| }{\alpha^2})$ updates.
        \item Finally, if each update $i$ uses custom parameters $\eta_i, \tau_i$ such that $\eta_i/2\alpha + \ln(1/\betaconv)/n \epsilon < \tau_i < 1/2$ then \whp we have at most $t_{\max}$ updates with $t_{\max} = \max\{t \ge 1 : \sum_{i=1}^t \eta_i (\tau_i \alpha  - \ln(1/\betaconv)/n\epsilon - \eta_i/2) \le \ln |\mathcal{X}| \}$.
    \end{enumerate}
\end{theorem}

\begin{proof}
    We follow a standard potential argument used in PMW analyses, but with variable parameters.
    We note $s_t$ the sign of the multiplicative weight update: $s_t = 1$ if $R_2 > q(h)$ or $R_3 > q(h) + \tau \alpha$, and $s_t = -1$ otherwise.
    We note $\nu_t$ the normalizing factor: $\nu_t := \sum_x h_t(x)\exp(s_t\eta q(x))$
    Then, we have $\forall x \in \mathcal{X}, h_{t+1}(x) = h_t(x) \exp(s_t \eta q(x)) / \nu_t$. Hence the decrease in potential at update $t$ is:
    \begin{align*}
        D(p\|h_{t+1}) & - D(p\|h_t) = \sum_x p(x) [\ln h_{t}(x) - \ln h_{t+1}(x)]            \\
                      & = \sum_x p(x) [\ln h_t(x) - (\ln h_t(x) + s_t\eta q(x) - \ln \nu_t)] \\
                      & = -s_t\eta q\cdot p +  \ln \nu_t
    \end{align*}

    We first upper bound $\ln(\nu_t)$ depending on the sign $s_t$ of the update.

    \begin{itemize}
        \item     First, suppose $s_t = -1$. For $a>0$, we have $\exp(-a) \le 1 - a + a^2/2$. We also have $q(x)^2 \le 1$ so:

              \begin{align*}
                  \nu_t & \le \sum_x h_t(x) (1 -  \eta q(x) + \eta^2 q(x)^2/2) \\
                        & \le 1 + \eta^2/2 - \eta q \cdot h_t
              \end{align*}

              Since $\ln(1+a) \le a \text{ for } a > -1 \text{, and } \eta^2/2 -  \eta q\cdot h_t \ge \eta^2/2 - \eta = (\eta - 1)^2/2 - 1/2 > -1$ we have:

              $$
                  \ln \nu_t \le \eta^2/2 - \eta q \cdot h_t = \eta^2/2 + s_t \eta q \cdot h_t
              $$

        \item Now, suppose $s_t = 1$. $\forall x, q(x) \in [0,1]$ so $1-q(x) > 0$. After multipliying by $\exp{(-\eta)}$ we can reuse the same reasoning as above:
              \begin{align*}
                  \nu_t \exp{(-\eta)} & =  \sum_x h_t(x)\exp(s_t\eta q(x) - \eta)                                  \\
                                      & = \sum_x h_t(x)\exp(- \eta (1 - q(x)))                                     \\
                                      & \le \sum_x h_t(x) (1 -  \eta (1 - q(x)) + \eta^2(1 - q(x))^2/2)            \\
                                      & \le 1 + \eta^2/2 - \eta + \eta q \cdot h_t \text{ since } \sum_x h_t(x) =1
              \end{align*}

              Since $\eta^2/2 - \eta + \eta q \cdot h_t > \eta^2/2 - \eta > -1$ we have $\ln \nu_t - \eta \le \eta^2/2 - \eta + \eta q \cdot h_t$, \ie $ \ln \nu_t \le  \eta^2/2 + s_t \eta q \cdot h_t$.

    \end{itemize}

    Both cases lead to the same upper-bound on the decrease in relative entropy:

    \begin{align}
        \label{eq:potential_decrease}
        D(p\|h_{t+1}) - D(p\|h_t)\le -s_t\eta(q \cdot p - q \cdot h_t) + \eta^2/2
    \end{align}

    Noticing that the first term includes the error in estimating query $q$ on $h_t$ compared to the true distribution $p$, which determines the sign of the multiplicative weights update, we can show that, \whp, $s_t(q \cdot p - q \cdot h_t) > \tau\alpha - \frac{\ln(1/\betaconv)}{n\epsilon}$.

    \begin{itemize}
        \item For queries answered with the PMW branch:  for any $\gamma_1 > 0$, with probability $1-\exp(-3\gamma_1)$ we have $|L| < \gamma_1/n\epsilon$ for $L \in \{X,Y,Z\}$.
              Let's take $\gamma_1 = \ln(1/\betaconv)/3$ to have $1-\exp(-3\gamma_1) = 1 - \betaconv$.
              Since we are performing an update, we have: $|q \cdot p - q \cdot h| + Y > \alpha/2 + X$, so $|q \cdot p - q \cdot h| > \alpha/2 + X - Y > \alpha/2 - 2\gamma_1/n\epsilon$.

              \begin{itemize}
                  \item              First, suppose that $s_t = 1$, \ie $q \cdot p + Z > q \cdot h$ (we need to increase the estimate $h$ to get closer to $p$).
                        Since $Z$ is small \whp, we have $q \cdot p - q \cdot h > -\gamma_1/n\epsilon$.
                        By assumption on $\betaconv$, we have $\ln(1/\betaconv)/n\epsilon\alpha < \eta/2\alpha + \ln(1/\betaconv)/n\epsilon\alpha < 1/2$, so $\alpha/2 > \ln(1/\betaconv)/n\epsilon = 3\gamma_1/n\epsilon$.
                        Thus, $q \cdot p - q \cdot h > - (\alpha/2 - 2\gamma_1/n\epsilon)$.
                        Since $|q \cdot p - q \cdot h| > \alpha/2 - 2\gamma_1/n\epsilon$, we must have $q \cdot p - q \cdot h > \alpha/2 - 2\gamma_1/n\epsilon > \alpha/2 - \ln(1/\betaconv)/n\epsilon$. In particular, $s_t(q \cdot p - q \cdot h) >\tau \alpha- \ln(1/\betaconv)/n\epsilon$ since $\tau < 1/2$.
                  \item Similarly, if $s_t = -1$ we have $q \cdot p - q \cdot h < -Z < \gamma_1/n\epsilon < \alpha/2 - 2\gamma_1/n\epsilon$ so we must have $q \cdot p - q \cdot h < -( \alpha/2 - 2\gamma_1/n\epsilon) < - \alpha/2 + 3\gamma_1/n\epsilon$.
                        Thus $s_t (q \cdot p - q \cdot h) > \tau \alpha - \ln(1/\betaconv)/n\epsilon$.
              \end{itemize}

        \item For queries answered with the Bypass branch: taking the same failure probability as the PMW branch, we note $\gamma_2 := \ln(1/\betaconv)$.
              With probability $1-\betaconv$ we have $|Y| < \gamma_2/n\epsilon$.
              \begin{itemize}
                  \item If $s_t = 1$ (\ie $q \cdot p + Y > q \cdot h + \tau\alpha$) we have $q \cdot p - q \cdot h > \tau\alpha - \gamma_2/n\epsilon$ \ie $s_t (q \cdot p - q \cdot h) > \tau \alpha - \ln(1/\betaconv)/n\epsilon$.
                  \item If $s_t = -1$ (\ie $q \cdot p + Y < q \cdot h - \tau\alpha$) we have $q \cdot p - q \cdot h < -\tau\alpha - \gamma_2/n\epsilon$ \ie $s_t (q \cdot p - q \cdot h) > \tau \alpha  - \ln(1/\betaconv)/n\epsilon$.
              \end{itemize}
    \end{itemize}

    Thus in each case, with probability $1-\betaconv$ we have:

    \begin{align}
        \label{eq:update_sign}
        s_t(q \cdot p - q \cdot h_t) > \tau\alpha - \frac{\ln(1/\betaconv)}{n\epsilon}
    \end{align}

    and thus, since $\tau\alpha - \frac{\ln(1/\betaconv)}{n\epsilon} >0$:
    \begin{align}
        \label{eq:potential_decrease_without_sign}
        D(p\|h_{t+1}) - D(p\|h_t)\le- \eta (\tau \alpha - \ln(1/\betaconv)/n\epsilon) + \eta^2/2
    \end{align}

    Finally, we take a union bound over all the random variables in the system, \ie at most $k$ queries, not all of them giving updates.
    After $t$ updates, with probability $1-k\betaconv$ we have:

    \begin{align*}
        D(p\|h_{t+1}) & = D(p\|h_{t+1}) - D(p\|h_t) + \dots  - D(p\|h_0) + D(p\|h_0)               \\
                      & \le  -t\eta(\tau \alpha - \ln(1/\betaconv)/n\epsilon - \eta/2) + D(p\|h_0)
    \end{align*}

    Since $0 \le D(p\|h_{t+1})$, we have $t\eta(\tau\alpha- \ln(1/\betaconv)/n\epsilon - \eta/2) \le  D(p\|h_0)$. If $\tau\alpha- \ln(1/\betaconv)/n\epsilon - \eta/2 > 0$, we can upper bound $t$:

    \begin{align}{ \label{eq:convergence_sum}}
        t \le \frac{D(p\|h_0)}{ \eta(\tau\alpha- \ln(1/\betaconv)/n\epsilon-\eta/2)}
    \end{align}

    $h_0$ is the uniform distribution and $p(x) \ln (p(x)) \le 0$, so:
    \begin{align*}
        D(p\|h_0) & =  \sum_x p(x)\ln (p(x)/h_0(x))                         \\
                  & = \sum_x p(x) [\ln(p(x)) + \ln(|\mathcal{X}|)]          \\
                  & \le \sum_x p(x) \ln(|\mathcal{X}|) = \ln(|\mathcal{X}|)
    \end{align*}

    Finally,

    $$t\le \frac{\ln |\mathcal{X}|}{ \eta(\tau\alpha- \ln(1/\betaconv)/n\epsilon-\eta/2)}$$
\end{proof}

\subsection{\Pmwbypasstree}

\newcommand{\graycomment}[1]{\textcolor{gray}{// \em #1}}

\begin{algorithm}[ht]
    \centering
    \begin{algorithmic}[1]
        \State $S \gets \emptyset$ \ \ \ \ \ \ \graycomment{Set of initialized SVs}
        \While{$\textproc{PrivacyAccountant.HasBudget()}$}
        \State Receive next query $q$
        \Statex \ \ \ \ \ \  \graycomment{Divide into subqueries following the tree-structured cache}
        \State $I \gets \textproc{SplitQuery}(q)$
        \Statex \ \ \ \ \ \ \graycomment{Allocate nodes to the SV branch or Laplace branch}
        \State $I_{SV} \gets \emptyset$
        \For{$i \in I$}
        \If{$\textproc{Heuristic.IsHistogramReady}(h_i,q_i, \alpha, \beta)$}
        \State $I_{SV} \gets I_{SV} \cup \{i\}$
        \EndIf
        \EndFor
        \Statex \ \ \ \ \ \ \graycomment{Proofs need one SV on contiguous window}
        \State $I_{SV} \gets \textproc{LargestContiguousSubset}(I_{SV})$
        \State $\epsilon_{SV} \gets \textproc{CalibrateBudgetSV}(\alpha, \beta, I_{SV})$
        \State $n_{SV} \gets \sum_{i \in I_{SV}} n_i$
        \Statex \ \ \ \ \ \ \graycomment{Initialize new SV if necessary}
        \If{$I_{SV} \not\in S$}
        \State $\textproc{PrivacyAccountant.pay}(I_{SV}, 3\epsilon_{SV})$
        \State $\tilde{\alpha}_{SV} \gets \alpha/2 + \lap{1/\epsilon_{SV} n_{SV}} $
        \State $S \gets S \cup \{I_{SV}\}$
        \EndIf
        \Statex \ \ \ \ \ \ \graycomment{Compute the histogram estimate}
        \State $r_H \gets \textproc{Agg}(\{q_i(h_i), i \in I_{SV}\})$
        \State $r_{SV}^* \gets \textproc{Agg}(\{q_i(\textproc{Data}_i), i \in I_{SV}\})$
        \Statex \ \ \ \ \ \ \graycomment{SV-check the estimate with noisy threshold $\tilde{\alpha}_{SV}$}
        \If{$|r_{SV}^* - r_{SV}| + \lap{1/\epsilon_{SV} n_{SV}} < \tilde{\alpha}_{SV}}$
        \State $r_{SV} \gets r_H$ \ \ \ \ \ \graycomment{Pay nothing}
        \Else
        \State $\textproc{PrivacyAccountant.pay}(I_{SV}, \epsilon_{SV})$
        \State $S \gets S \setminus \{I_{SV}\}$  \ \ \ \ \graycomment{Will pay initialization fee next time}
        \State $r_{SV} \gets r_{SV}^* + \lap{1/\epsilon_{SV} n_{SV}}$
        \Statex \ \ \ \ \ \ \ \ \ \ \graycomment{Update all the sub-histograms in the same direction}
        \State $\eta \gets \textproc{Sign}(r_{SV} - r_H) \cdot \lr$
        \For{$i \in I_{SV}$}
        \State $h_i \gets \textproc{Update}(h_i, q_i, \eta)$
        \EndFor
        \EndIf
        \Statex \ \ \ \ \ \ \graycomment{Compute the remaining subqueries with Laplace}
        \State $I_{Lap} \gets I \setminus I_{SV}$

        \State $\epsilon_{Lap} \gets \textproc{CalibrateBudgetLaplace}(\alpha, \beta, I_{Lap})$
        \For{$i \in I_{Lap}$}
        \State $\textproc{PrivacyAccountant.pay}(i, \epsilon_{Lap})$
        \State $r_i \gets q_i(\textproc{Data}_i) +\lap{1/\epsilon_{Lap} n_i}$
        \State $\eta \gets \textproc{Sign}(r_i - q_i(h_i))\cdot \mathds{1}[|r_i - q_i(h_i)| > \tau \alpha]  \cdot \lr $
        \State $h_i \gets \textproc{Update}(h_i, q_i, \eta)$
        \EndFor
        \State $r_{Lap} \gets \textproc{Agg}(\{r_i, i \in I_{Lap}\})$.
        \State Output $R = \textproc{Agg}(\{r_{SV}, r_{Lap}\})$
        \EndWhile


    \end{algorithmic}
    \caption{\bf \Pmwbypasstree algorithm}
    \label{alg:pmwbypass-tree}
\end{algorithm}


Algorithm \ref{alg:pmwbypass-tree} formalizes our \pmwbypasstree caching object and its use. At a high level, we build a binary tree over our data stream, in which each leaf is a partition, and each node maintains a histogram responsible for all the data it spans (which is a superset of nodes lower in the hierarchy). A query is split into sub-queries to minimize the number of nodes, and hence histograms, to query (l. 4). We first find the largest consecutive {\em subset} of nodes with histograms ready for guesses (that we will not bypass), ll. 5-9. This subset will be handled with a \pmwbypass query, at a DP cost that depends on its size (l. 10, see \S\ref{sec:tree-structured-caching} for details), and ``using'' half the failure probability for accuracy ($\beta/2$ failure probability). This is done ll. 11-25. The remaining queries are computed as a bypass query for each node independently (including an histogram update if warranted), combining to $\alpha$ accuracy with the remaining $\beta/2$ failure probability, ll. 26-33. Finally all results are aggregated to answer the original query with $\alpha$ accuracy with probability at least $1-\beta$.

\heading{Notation.} We define functions and data structures used in \A\ref{alg:pmwbypass-tree}:
\begin{itemize}
    \item For $T$ timestamps, we note $\mathcal{I}$ the tree structure introduced in Section~\ref{sec:tree-structured-caching}: $\mathcal{I} := \{(a,b) \in [0, T-1]^2 | \exists k \in \mathbb{N} : b - a + 1 = 2^k \wedge a \equiv 0 \mod 2^k\}$.
    \item For each node $i = (a,b) \in \mathcal{I}$ we maintain a histogram $h_i$ estimating the data distribution $\textproc{Data}_i$ for the window $[a,b]$.
    \item We note $n_i$ the number of datapoints with timestamp in the $i = [a,b]$ range. We use the same definition of DP for streams as \cite{pmw, streaming_counter} introduced in \S\ref{sec:appendix:notation}, so $n_i$ is public.
    \item For a query $q$ requesting a window $[a,b]$, $\textproc{SplitQuery}(q)$ returns $I \subseteq \mathcal{I}$ the smallest set of nodes covering $[a,b]$.
    \item For a set of nodes $I_{SV}$, $\textproc{LargestContiguousSubset}(I_{SV})$ returns the largest set $J \subseteq I_{SV}$ such that there exists $K \in \mathbb{N}$ and a strictly increasing sequence $u \in \mathbb{N}^K$ verifying $J = \{(u_k, u_{k+1} - 1), k \in [K]\}$.
    \item For a set of results over some nodes $I$,  $\textproc{Agg}(\{r_i, i \in I\})$ returns the weighted average $\sum_{i \in I} \frac{n_i}{ \sum_{j \in I} n_j} r_i$.
    \item In particular, if $(r_i^*)$ is a set of true results, $\forall i,  \in I |r_i^*- r_i| \le \alpha \implies |\textproc{Agg}(\{r_i^*, i \in I\}) - \textproc{Agg}(\{r_i, i \in I\})| \le \alpha$.
    \item $\textproc{PrivacyAccountant.pay}(I, \epsilon)$ spends budget for each timestamp in $I$ using block composition \cite{sage-tr}.
    \item  $\textproc{CalibrateBudgetLaplace}(\alpha, \beta, I_{Lap})$ uses a binary search over a Monte Carlo simulation with parameter $N$ and failure probability $\beta_{MC}(N)$ to find:

          $$\epsilon_{Lap} := \min \{ \epsilon: \Pr[ |\sum_{i=1}^{|I_{Lap}|} \lap{1/\epsilon}| > n_{Lap} \alpha ] < \beta/2 - \beta_{MC}(N) \}$$




    \item $\textproc{CalibrateBudgetSV}(\alpha, \beta, I_{SV})$ returns:
          $$\epsilon_{SV} = 4\ln(2/\beta)/n_{SV}\alpha$$
\end{itemize}

\begin{theorem}[Privacy of \Pmwbypasstree]
    \label{thm:pmwbypasstree:privacy}
    \Pmwbypasstree preserves $\epsilon_G$-DP across the queries it executes.
\end{theorem}
\begin{proof}
    \Pmwbypasstree uses the same DP mechanisms as \pmwbypass, but on different and potentially overlapping subsets of the data.
    To compose these mechanisms we use block composition \cite{sage-tr} instead of a simple privacy filter.
\end{proof}

\begin{theorem}[Accuracy of \Pmwbypasstree]
    \label{thm:pmwbypasstree:accuracy}
    \Pmwbypasstree is $(\alpha,\beta)$-accurate for each query it answers.
\end{theorem}
\begin{proof}


    The binary search for $\epsilon_{Lap}$ gives that with probability $1- \beta_{MC}(N)$ over the Monte Carlo randomness we have:

    $$\Pr[ |\sum_{i=1}^{|I_{Lap}|} \lap{1/\epsilon_{Lap}}| > n_{Lap} \alpha ] < \beta/2 - \beta_{MC}(N)$$

    Since $X \sim \lap{b} \implies kX \sim \lap{kb}$ for $k,b > 0$ we have, with the triangle inequality:

    \begin{align*}
        * & := \Pr[|r_{Lap}^* - r_{Lap}| > \alpha]                                                                \\
          & = \Pr[ |\textproc{Agg}(\{r_i^*, i \in I_{Lap}\}) - \textproc{Agg}(\{r_i, i \in I_{Lap}\})| > \alpha ] \\
          & = \Pr[|\sum_{i \in I_{Lap}} \frac{n_i}{n_{Lap}} (r_i^*- r_i)| > \alpha]                               \\
          & = \Pr [|\sum_{i \in I_{Lap}} \frac{n_i}{n_{Lap}} \lap{1/\epsilon_{Lap}n_i}| > \alpha ]                \\
          & = \Pr[ |\sum_{i=1}^{|I_{Lap}|} \lap{1/\epsilon_{Lap}}| > n_{Lap} \alpha ]                             \\
          & < \beta/2 - \beta_{MC}(N)
    \end{align*}

    A union bound gives that with probability $1- (\beta_{MC}(N) + \beta/2 - \beta_{MC}(N))$ over the simulation and the Laplace randomness we have:

    $$ \Pr[|r_{Lap}^* - r_{Lap}| > \alpha]    < \beta/2$$
    Thanks to Theorem \ref{thm:pmwbypass:accuracy}, for $\epsilon_{SV} = 4\ln(2/\beta)/n\alpha$ we have $\Pr[|r_{SV} - r_{SV}^*| > \alpha] < \beta/2$.

    Finally, a union bound and the fact that $n_{SV} + n_{Lap} = n$ gives:

    \begin{align*}
        * & := \Pr[|R^* - R| > \alpha]                                                                        \\
          & = \Pr[| \frac{n_{Lap}}{n} (r_{Lap}^* - r_{Lap}) +\frac{n_{SV}}{n}(r_{SV}^* - r_{SV}) | >\alpha  ] \\
          & \le \Pr[ \frac{n_{Lap}}{n} |r_{Lap}^* - r_{Lap}| +\frac{n_{SV}}{n}|r_{SV}^* - r_{SV}|> \alpha]    \\
          & \le \Pr[ \{|r_{Lap}^* - r_{Lap}| > \alpha\} \cup \{|r_{SV}^* - r_{SV}|> \alpha\}]                 \\
          & \le \beta/2 + \beta/2 = \beta
    \end{align*}

\end{proof}

\begin{theorem}[Convergence of \Pmwbypasstree on a bounded number of partitions]
    \label{thm:pmwbypasstree:convergence_bounded}
    Consider a partitioned dataset with $T = 2^m$ timestamps, with $m \in \mathbb{N}^*$, containing $n$ datapoints.
    We note $n_{\min}$ the size of the smallest partition.
    At each update $t$, each node $i = (a,b) \in \mathcal{I}$ of the \Pmwbypasstree contains a histogram $h_t^i$ approximating the true distribution $p_i$ of the data on timestamp range $[a,b]$.
    For $i \in \mathcal{I}$ we note $\lambda_i := \frac{n_i}{(m+1)n}$.

    Let $\betaconv \in (0,1)$ be a parameter for per-update failure probability.
    Suppose that $\eta/2\alpha + 2m \frac{\ln(1/\betaconv) }{\ln(2/\beta)} < \tau \le 1/2$.
    Then:

    \begin{enumerate}
        \item After each update we have, with probability $1-2m \betaconv$:

              \begin{align*} \sum_{i \in \mathcal{I}} \lambda_{i} \big\{ D(p^{i} \parallel h_{t+1}^{i}) - D(p^{i} \parallel h_{t}^{i}) \big\} \\
                  \le -   \frac{\eta n_{\min}}{(m+1)n}\left(\tau \alpha - 2m \frac{\ln(1/\betaconv)\alpha }{\ln(2/\beta)} -  \eta /2 \right)
              \end{align*}

        \item Moreover, for $k$ an upper bound on the number of queries over the whole tree, with probability $1-2km \betaconv$ the number of updates we perform is at most:
              $$\frac{(m+1) n\ln |\mathcal{X}|}{\eta n_{\min} ( \tau \alpha  - 2m \frac{\ln(1/\betaconv)\alpha }{\ln(2/\beta)} - \eta/2)}$$
        \item If $n_{\min} = n/T$ (partitions with equal size), we can set $\betaconv$ such that if $\eta/\alpha < \tau \le 1/2$ we have at most $\frac{(m+1) T \ln |\mathcal{X}|}{\eta ( \tau \alpha - \eta)/2}$ updates.
    \end{enumerate}

\end{theorem}

\begin{proof}

    First, we have $\sum_{i \in \mathcal{I}} \lambda_i = \frac{\sum_i n_i}{(m+1)n} = 1$ because each of the $m+1$ layers of the binary tree contains nodes covering the $n$ datapoints of the range $[0,T-1]$.

    Consider an update $t$. Any node $i$ belongs to one of three sets:
    \begin{enumerate}
        \item If $i$ hasn't been updated, $h_{t+1}^{i} = h_{t}^{i}$ and  $D(p^{i} \parallel h_{t+1}^{i}) =  D(p^{i} \parallel h_{t}^{i})$.
        \item If $i$ has been updated by a Bypass branch with sign $s_t^{i}$, Equation \ref{eq:potential_decrease_without_sign} shows that with probability $1- \betaconv$ we have:
              \begin{align}
                  D(p^{i} \parallel h_{t+1}^{i}) & -  D(p^{i} \parallel h_{t}^{i}) \le -\eta (\tau \alpha - \ln(1/\betaconv)/n_i\epsilon_{Lap}^t) + \eta^2/2                         \nonumber  \\
                                                 & \le -\eta\tau \alpha + \eta \frac{\ln(1/\betaconv)}{n_i\frac{\ln(2/\beta)}{n_{Lap} \alpha}} + \eta^2/2                             \nonumber \\
                                                 & = -\eta\tau \alpha + \eta \frac{n_{Lap} \ln(1/\betaconv)\alpha }{n_i\ln(2/\beta)} + \eta^2/2 \label{eq:potential_decrease_laplace}
              \end{align}

              where $\frac{\ln(2/\beta)}{n_{Lap} \alpha}$ is a lower bound on $\epsilon_{Lap}$ because \\
              $ \Pr[\sum_{i \in \mathcal{I}_{Lap}} |\lap{ \frac{n_{Lap} \alpha}{\ln(2/\beta)}}| >  n_{Lap} \alpha] \ge \\ \Pr[ |\lap{ \frac{n_{Lap} \alpha}{\ln(2/\beta)}}| >  n_{Lap} \alpha] = \beta /2$.

        \item Otherwise, $i$ has been updated by the SV branch.  We detail this case below.
    \end{enumerate}

    Note $s_t^{(SV)}$ the sign of the global update.
    For each $i \in I_{SV}$, Equation \ref{eq:potential_decrease} gives:

    \begin{align}
        \label{eq:potential_decrease_sv_node}
        D(p^{i} \parallel h_{t+1}^{i}) \le  D(p^{i} \parallel h_{t}^{i}) - s_t^{(SV)} \eta (q \cdot p^{i} - q \cdot h_t^{i})  + \eta^2/2
    \end{align}

    However, we can't reuse Equation \ref{eq:update_sign} here to show that the potential $ D(p^{i} \parallel h_{t}^{i})$ of every single node decreases.
    Indeed, a single SV is used to update all the nodes in $I_{SV}$ in the same direction, so $s_t^{(SV)} (q \cdot p^{i} - q \cdot h_t^{i})$ can have either sign.
    In other words, if one node $i \in I_{SV}$ is close to the answer but others are far from it, $i$ might witness an {\em increase} in potential because the SV check considers only the aggregated state of the nodes.
    Instead, we show that the combined potential of $I_{SV}$ decreases.

    We note $\lambda := \sum_{i \in I_{SV}} \lambda_i$, $\forall i \in I_{SV}, \lambda_i' := \lambda_i/\lambda$, $p^{SV} :=  \sum_{i \in I_{SV}} \lambda_i' p^i$ and $h_t^{SV} :=  \sum_{i \in I_{SV}} \lambda_i' h_t^i$.
    Thanks to Equation \ref{eq:potential_decrease_sv_node} we have:
    \begin{align*}
        * & := \sum_{i \in I_{SV}} \lambda_i' D(p^{i} \parallel h_{t+1}^{i})                                                                              \\
          & \le  \sum_{i \in I_{SV}}\lambda_i'  \left(D(p^{i} \parallel h_{t}^{i}) - s_t^{(SV)} \eta (q \cdot p^{i} - q \cdot h_t^{i})  + \eta^2/2\right) \\
          & = \left(\sum_{i \in I_{SV}}\lambda_i'  D(p^{i} \parallel h_{t}^{i})\right) - s_t^{(SV)} \eta (q \cdot p^{SV} - q \cdot h_t^{SV}) +\eta^2/2
    \end{align*}

    because $\sum_{i \in I_{SV}} \lambda_i'  = 1$ and $q$ is a linear query.

    Since $q\cdot p^{SV}$ is the true result of the query on $I_{SV}$, $q \cdot h_t^{SV}$ is the combined histogram estimate and $\tau <1/2$, we can apply Equation \ref{eq:update_sign}.
    With probability $1-\betaconv$ we have:
    $
        \sum_{i \in I_{SV}} \lambda_i' D(p^{i} \parallel h_{t+1}^{i}) \le \left(\sum_{i \in I_{SV}}\lambda_i'  D(p^{i} \parallel h_{t}^{i})\right) -\eta (\tau \alpha - \ln(1/\betaconv)/n_{SV}\epsilon_{SV}^t) + \eta^2/2
    $
    \ie:

    \begin{align}
        \label{eq:potential_decrease_sv_combined}
        \sum_{i \in I_{SV}} \lambda_i' D(p^{i} \parallel h_{t+1}^{i}) - \left(\sum_{i \in I_{SV}}\lambda_i'  D(p^{i} \parallel h_{t}^{i})\right) \nonumber \\ \le -\eta \tau \alpha + \frac{\eta \ln(1/\betaconv) \alpha}{4\ln(2/\beta)} + \eta^2/2
    \end{align}

    since $\epsilon_{SV} = 4\ln(2/\beta)/n_{SV}\alpha$.\\

    Now we can bound the global drop in potential.
    Equations \ref{eq:potential_decrease_sv_combined} and \ref{eq:potential_decrease_laplace} followed by a union bound show that, with probability $1 - (|I_{Lap}| + 1) \betaconv$ we have:

    \begin{align}
        * & := \sum_{i \in \mathcal{I}} \lambda_{i} D(p^{i} \parallel h_{t+1}^{i}) -  \sum_{i \in \mathcal{I}} \lambda_{i} D(p^{i} \parallel h_{t}^{i})                          \nonumber \\
          & \le \sum_{i \in I_{Lap}} \lambda_{i}(-\eta\tau \alpha + \eta \frac{n_{Lap} \ln(1/\betaconv)\alpha }{n_i\ln(2/\beta)} + \eta^2/2) +                                   \nonumber \\
          & \sum_{i \in I_{SV}}\lambda \cdot  \lambda_i' D(p^{i} \parallel h_{t+1}^{i}) - \left(\sum_{i \in I_{SV}}\lambda \cdot \lambda_i'  D(p^{i} \parallel h_{t}^{i})\right) \nonumber \\
          & \frac{n_{SV}}{n(m+1)}\left(-\eta \tau \alpha + \frac{\eta \ln(1/\betaconv) \alpha}{4\ln(2/\beta)} + \eta^2/2\right)
        \label{eq:potential_drop_sv_lap}
    \end{align}

    We have  $1/4 \le 2m$ and $|I_{Lap}| \le 2m$, because $2m$ is the maximum number of nodes required to cover a contiguous range of $[0,2^m-1]$.
    Since we have $\tau \alpha -  2m \frac{\ln(1/\betaconv)\alpha }{\ln(2/\beta)} - \eta/2 > 0$ by assumption,
    Equation \ref{eq:potential_drop_sv_lap} becomes:

    \begin{align*}
        * & \le -\eta \cdot \frac{n_{SV} + n_{Lap}}{n(m+1)} \cdot (\tau \alpha -  2m \frac{\ln(1/\betaconv)\alpha }{\ln(2/\beta)} - \eta/2) \\
          & \le -\eta \cdot \frac{n_{\min}}{n(m+1)}\cdot (\tau \alpha -  2m \frac{\ln(1/\betaconv)\alpha }{\ln(2/\beta)} - \eta/2)
    \end{align*}

    where the last inequality comes from $n_{SV} + n_{Lap} \ge n_{\min}$.

    Now, let's use this per-update potential drop to obtain a global bound on the number of updates.
    By convexity and positivity of the relative entropy \cite{entropy}, we have:

    $$0 \le D\left(\sum_{i} \lambda_{i} p^{i} \parallel \sum_{i} \lambda_{i} h^{i}_{t+1}\right) \le  \sum_{i} \lambda_{i} D(p^{i} \parallel h_{t+1}^{i}) $$

    If we have at most $k$ queries we have, with probability $1-k\cdot 2m \betaconv$:

    \begin{align*}
        t \le \frac{\ln |\mathcal{X}|}{\eta ( \tau \alpha  - 2m \frac{\ln(1/\betaconv)\alpha }{\ln(2/\beta)} - \eta/2)}\frac{n(m+1)}{n_{\min}}
    \end{align*}

    with the same reasoning as in Equation \ref{eq:convergence_sum} and the observation that $\sum_{i} \lambda_{i} D(p^{i} \parallel h_{0}^{i}) \le \ln |\mathcal{X}|$.

\end{proof}

\def\tmax{T}
\def\kmax{m}

\begin{theorem}[Convergence of \Pmwbypasstree on an unbounded number of partitions]
    \label{thm:pmwbypasstree:convergence}


    Suppose that we have only a bound $\tmax = 2^{\kmax}$ on the number of contiguous partitions a query can request (but not necessarily any bound on the total number of partitions in the database).
    Take $\nu_{\min}$ the smallest fraction of datapoints in one partition from any contiguous set of $2T$ partitions ($\nu_{\min} = 1/2T$ if all the partitions hold the same number of points).

    We change $\textproc{SplitQuery}$ and the histogram structure in \A\ref{alg:pmwbypass-tree} as follows:
    \begin{itemize}
        \item  For $\kappa \in \mathbb{N}$ we note $\mathcal{I}_\kappa$ the binary tree of depth $\kmax + 1$ whose leftmost leaf is $\kappa\tmax$ (thus covering $[\kappa \tmax, (\kappa + 2) \tmax -1]$):
              $\mathcal{I}_\kappa := \{(\kappa, \kappa \tmax + a, \kappa \tmax + b) \in \{\kappa \} \times [0, 2\tmax-1] \times  [0, 2\tmax-1] | \exists k \in [0,\kmax+1] : b - a + 1 = 2^k \wedge a \equiv 0 \mod 2^k\}$
        \item $\mathcal{I}$ becomes $\cup_{\kappa \in \mathbb{N}} \mathcal{I}_\kappa$. Note that these trees overlap: for $(a,b) \in [\kappa\tmax, \kappa\tmax + \tmax-1]^2$ with $\kappa \ge 1$ we have $(\kappa,a,b) \in \mathcal{I}$ but also $(\kappa -1, a, b) \in \mathcal{I}_{\kappa - 1}$.
        \item Consider a query requesting a range $(a,b) \in [\kappa\tmax, \kappa\tmax + \tmax-1] \times [\kappa \tmax, \kappa \tmax + 2\tmax -1]$ with $b-a+1 \le \tmax$. As noted above, $(a,b)$ might be covered by 2 trees (some windows are covered by a single tree, namely if $a < (\kappa+1)\tmax \le b$), so by convention we pick the rightmost tree. $\textproc{SplitQuery}$ returns $I \subseteq \mathcal{I}_{\kappa}$ the smallest set of nodes covering $[a,b]$.
    \end{itemize}

    Then, for all $\kappa \in \mathbb{N}$, for $k$ an upper bound on the number of queries allocated to $\mathcal{I}_\kappa$, with probability $1-2k(m+1) \betaconv$ the number of updates we perform on $\mathcal{I}_\kappa$ is at most:
    $$\frac{(m+2)\ln |\mathcal{X}|}{\eta \nu_{\min} ( \tau \alpha  - 2(m+1) \frac{\ln(1/\betaconv)\alpha }{\ln(2/\beta)} - \eta/2)}$$

\end{theorem}
\begin{proof}


    Consider $\kappa \in \mathbb{N}$.
    Consider the set of queries allocated to $\mathcal{I}_\kappa$.
    They all request data from a partitioned dataset with $2\tmax$ timestamps.
    Moreover, any update to a histogram in  $\mathcal{I}_\kappa$ must come from that set of queries.
    Hence we can apply \Thm\ref{thm:pmwbypasstree:convergence_bounded} on a tree of size $2\tmax$.
\end{proof}

\subsection{Warm-start}

\begin{theorem}[Warm-start]
    \label{thm:warm-start}

    Consider \A\ref{alg:pmw-bypass} using for histogram initialization a distribution $h_0$ instead of uniform.
    Suppose there exists $\lambda \ge 1$ such that $\forall x \in \mathcal{X}, h_0(x) \ge \frac{1}{\lambda |\mathcal{X}|}$.
    Then if  $\eta/\alpha < \tau \le 1/2$ the number of updates we perform is at most:

    $$\frac{\ln(\lambda |\mathcal{X}|)}{ \eta(\tau\alpha - \eta)/2}$$

    All the other convergence results from \Thm\ref{thm:pmwbypass:convergence}, \ref{thm:pmwbypasstree:convergence_bounded} and \ref{thm:pmwbypasstree:convergence} hold with $ \ln(\lambda |\mathcal{X}|)$ instead of $ \ln(|\mathcal{X}|)$.

\end{theorem}

\begin{proof}
    The initial relative entropy is:

    \begin{align*}
        D(p \| h_0) & = \sum_x p(x)\ln (p(x)/h(x)) \le  \sum_x p(x)\ln (p(x)\lambda |\mathcal{X}|)      \\
                    & = \ln(\lambda |\mathcal{X}|) + \sum_x p(x)\ln(p(x)) <  \ln(\lambda |\mathcal{X}|)
    \end{align*}

    We can use Equation~\ref{eq:convergence_sum} from \Thm\ref{thm:pmwbypass:convergence} with $ \ln(\lambda |\mathcal{X}|)$ instead of $ \ln(|\mathcal{X}|)$.
\end{proof}

\subsection{Bounding the total privacy budget}
\label{sec:appendix:bounded_budget}

\Thm\ref{thm:pmwbypass:convergence}'s convergence bound is expressed as a maximum number of {\em updates}, \ie number of queries that alter the state of the histogram.
However, this does not directly bound the total privacy budget used to answer a workload, because some queries can cost budget without updating the histogram.
Indeed, apart from the first SV initialization, there are two ways a query $q$ can cost budget: either $q$ goes through the histogram branch in \A\ref{alg:pmw-bypass} and triggers an SV reset, or $q$ goes through the Bypass branch and pays for a Laplace query. While an SV reset always yields an update, some Laplace queries cost budget without triggering an external update (\A\ref{alg:pmw-bypass}, l.\ref{line:external_update}).

Thus, it is theoretically possible to craft a workload that consumes an unbounded amount of budget (up to the privacy filter enforced maximum), by issuing queries that go through the Bypass branch and cost budget without yielding external updates.
We do not observe this phenomenon in our evaluation, but it is straightforward to prevent this problem by simply deactivating the Bypass branch after a predetermined number of queries.
More precisely, in \A\ref{alg:pmw-bypass}, l.\ref{line:heuristic}, we can modify \textsc{Heuristic.IsHistogramReady} to always return \textsc{True} after $k$ queries, for some parameter $k$.
Therefore, at most $k$ queries can cost budget without yielding an update, and the result from \Thm\ref{thm:pmwbypass:convergence} directly bounds the number of times we spend budget, up to a constant factor.
The same reasoning applies to \pmwbypasstree.

\subsection{Gaussian mechanism, R\'enyi DP and Approximate DP}
\label{sec:appendix:gaussian_rdp}

The Gaussian mechanism~\cite{privacybook} has desirable properties to answer workloads of DP queries.
We show how to modify \A\ref{alg:pmw-bypass} to use it as an alternative for the Laplace mechanism.
Since the Gaussian mechanism does not sastify $(\epsilon, 0)$-DP, we first need to introduce a more general privacy accountant.

\heading{RDP accounting.}
\A\ref{alg:pmw-bypass} can be modified to use R\'enyi DP (RDP) \cite{8049725}, a form of accounting that offers better composition than basic composition under pure $(\epsilon, 0)$-DP guarantees (even for workloads of pure DP mechanisms such as the Laplace or the SV mechanisms).
RDP accounting works particularly well for the Gaussian mechanism.
The modifications are as follows:

\begin{itemize}
    \item $\textproc{PrivacyAccountant}$ is an RDP filter (\Thm\ref{thm:rdp_filter_adaptive_concurrent_corollary}) instead of a pure DP filter.
    \item For a Laplace mechanism $\lap{1/\epsilon n}$ on a query with $\ell_1$ sensitivity $1/n$ (like our linear queries), instead of paying $\epsilon$, the budget for each RDP order $a >1$ is:

          $$ \frac{1}{a-1}\ln\left\{ \frac{a}{2a-1}\exp(\epsilon(a-1)) +\frac{a-1}{2a-1} \exp(-\epsilon a)  \right\}$$
          This comes directly from the RDP curve of the Laplace mechanism \cite{8049725}.
    \item For an SV initialization where each internal Laplace uses  $\lap{1/\epsilon n}$, the budget for each RDP order $a >1$ is:

          $$ \frac{1}{a-1}\ln\left\{ \frac{a}{2a-1}\exp(\epsilon(a-1)) +\frac{a-1}{2a-1} \exp(-\epsilon a)  \right\} + 2\epsilon$$

          This comes from the RDP analysis of the SV mechanism \cite{renyi_sv}. More precisely, we use Algorithm 2 of \cite{renyi_sv} with $c=1$, $\Delta =1/n$, $\epsilon_1 = \epsilon$ and $\epsilon_2 = 2\epsilon$.
          $M_\rho = \lap{\Delta/\epsilon_1}$ is $\epsilon_\rho(\alpha) = \epsilon_{\lambda_1}(\alpha)$-RDP for queries with sensitivity $\Delta$, with  $\epsilon_\lambda(\alpha)$ the Laplace RDP curve with $\lambda_1 = 1/\epsilon_1$.
          $M_\nu = \lap{2\Delta/\epsilon_2}$ is $\epsilon_\nu(\alpha) = \epsilon_{\lambda_2}(\alpha)$-RDP for queries with sensitivity $2\Delta$, with $\lambda_2 = 1/\epsilon_2$.
          Since $\epsilon_\rho(\infty) = \epsilon_2 < \infty$ we can use Point 3 of Theorem 8, so SV is $\epsilon_\nu(\alpha) +  \epsilon_2$-RDP.
    \item For a Gaussian mechanism adding noise from $\mathcal{N}(0, \sigma^2/n^2)$ on a query with $\ell_2$ sensitivity $1/n$, the budget for each RDP order $a >1$ is $\frac{a}{2\sigma^2}$.
          The default version of \A\ref{alg:pmw-bypass} does not use any Gaussian mechanism, but we will add some below.
\end{itemize}

Finally, we can use the RDP-to-DP conversion formula \cite{8049725} to obtain $(\epsilon, \delta)$-DP guarantees for \pmwbypass for $\epsilon, \delta >0$.

\heading{Gaussian mechanism.}
\A\ref{alg:pmw-bypass} can now be modified to use the Gaussian mechanism to answer queries directly. We keep the internal Laplace random variables used by the SV protocol, although there exists SV protocols with purely Gaussian noise~\cite{renyi_sv}.
The modifications are as follows:

\begin{itemize}
    \item Line \ref{line:calibrate} becomes  $\epsilon \gets \textproc{CalibrateBudget}(\alpha, \beta); \sigma \gets \textproc{CalibrateBudgetGaussian}(\alpha, n, \epsilon, \tau)$ for a function defined next.
    \item $\textproc{CalibrateBudgetGaussian}(\alpha, n, \epsilon, \tau)$ returns:
          $$\sigma = \frac{\tau \alpha}{\sqrt{18\ln 2 + 3 \tau n \alpha \epsilon}}$$
    \item In Lines \ref{line:R2} and \ref{line:R3}, we replace $\lap{1/\epsilon n}$ by $\mathcal{N}(0, \sigma^2/n^2)$
\end{itemize}

At a high level, we need to calibrate the noise of DP queries answered with the Gaussian mechanism so that they are compatible with the failure probabilities of the Laplace queries. That is, Gaussian mechanism queries need to have the same (or lower) error bound $\alpha$ with the same (or lower) failure probability $\betaconv$.
The following result shows the Gaussian noise variance $\sigma^2/n^2$ to use for this calibration:

\begin{lemma}[Gaussian tail bounds]
    \label{lem:tail-bounds}
    Consider $\alpha, \beta \in (0,1)^2$, $n\in \mathbb{N}^2$, $\epsilon > 0$ and $\tau \le 1/2$.
    Pose $\betaconv := \exp(-\tau n \alpha \epsilon /2)$, $\gamma_1 := \ln(1/\betaconv)/3$ and $\gamma_2 := \ln(1/\betaconv)$.
    Pose $\sigma := \frac{\tau \alpha}{\sqrt{18\ln 2 + 3 \tau n \alpha \epsilon}}$ and take $Z \sim \mathcal{N}(0, \sigma^2/n^2)$. Then, $Z$ satisfies the three following tail bounds:

    \begin{itemize}
        \item $\Pr[|Z| > \alpha] \le \exp(-\alpha n \epsilon)$
        \item $\Pr[|Z| > \gamma_1/n\epsilon] \le \exp(-\gamma_1)$
        \item $\Pr[|Z| > \gamma_2/n\epsilon] \le \exp(-\gamma_2)$
    \end{itemize}

\end{lemma}

\begin{proof}

    First, we recall that for any $t >0$ the upper deviation inequality gives:

    $$\Pr[|Z| > t] \le 2 \exp(-\frac{t^2}{2\sigma^2})$$

    Thus, to have $\Pr[|Z| > t] \le \exp(-t n \epsilon)$ it is sufficient to have $2\exp(-\frac{t^2}{2\sigma^2}) \le  \exp(-t n \epsilon)$ \ie $\ln{2} + tn\epsilon \le \frac{t^2}{2\sigma^2}$ \ie $\sigma^2 \le \frac{t^2}{2\ln{2} + 2tn\epsilon}$.
    Note $f: t \mapsto \frac{t^2}{2\ln{2} + 2tn\epsilon}$.
    We have $f'(t) = \frac{2t(2\ln 2 + 2 t n\epsilon) - t^2 2n\epsilon}{(2\ln{2} + 2tn\epsilon)^2} \ge 0$ so $f$ is monotonically increasing on $\mathbb{R}^+$.
    We also have $\ln(1/\betaconv) = \tau n \alpha \epsilon /2 \le n \alpha \epsilon$, thus $\gamma_1/n\epsilon < \gamma_2/n\epsilon \le \alpha$.

    Hence if $\sigma^2 \le f(\gamma_1/n\epsilon)$ then also $\sigma^2 \le f(\gamma_2/n\epsilon)$ and $\sigma^2 \le f(\alpha)$, which together imply $\Pr[|Z| > \alpha] \le \exp(-\alpha n \epsilon)$, $\Pr[|Z| > \gamma_1/n\epsilon] \le \exp(-\gamma_1)$ and $\Pr[|Z| > \gamma_2/n\epsilon] \le \exp(-\gamma_2)$.

    It is sufficient to take $\sigma^2 = \frac{(\gamma_1/n\epsilon)^2}{2\ln{2} + 2\gamma_1} = \frac{(\tau \alpha/3)^2}{2\ln2 + \tau n \alpha \epsilon / 3}$ \ie $\sigma = \frac{\tau \alpha}{\sqrt{18\ln 2 + 3 \tau n \alpha \epsilon}}$ to conclude.

\end{proof}

\def\gaussianpmwbybass{Gaussian \pmwbypass}
Finally, we prove that our modified version of \A\ref{alg:pmw-bypass}, that we call \gaussianpmwbybass, satisfies the same properties as the original (with RDP accounting instead of Pure DP accounting):

\begin{theorem}[\gaussianpmwbybass guarantees]
    \label{thm:gaussian-guarantees}
    ~
    \begin{enumerate}
        \item \gaussianpmwbybass preserves $(\epsilon_G, \delta_G)$-DP for a global DP budget set upfront in the $\textproc{PrivacyAccountant}$.
        \item \gaussianpmwbybass is $(\alpha,\beta)$-accurate for each query it answers.
        \item If $\eta/\alpha < \tau$, \gaussianpmwbybass performs at most $\frac{\ln |\mathcal{X}|}{ \eta(\tau\alpha - \eta)/2}$ updates.
    \end{enumerate}
\end{theorem}
\begin{proof}
    The key observation is that the proofs for \pmwbypass do not explicitly require the noise added Lines \ref{line:R2} and \ref{line:R3} to come from a Laplace distribution. We only use three tail bounds on that noise realization. Any random variable that satisfies the same tail bounds (with the same parameters) also works.
    \begin{enumerate}
        \item We use a RDP filter \cite{adaptive_rdp} instead of a pure DP filter \cite{rogers2016privacy} in \Thm\ref{thm:pmwbypass:privacy}. The proof of privacy for each step is sketched in the RDP accounting paragraph.
        \item We replace the Laplace tail bounds in Lemma~\ref{lem:pmw-accuracy} and \Thm\ref{thm:pmwbypass:accuracy} by identical tail bounds satisfied by Lemma \ref{lem:tail-bounds}: $\Pr[|Z| > \alpha] \le \exp(-\alpha n \epsilon)$ in both cases.
        \item In the PMW branch, we replace the tail bound on $Z$ by $\Pr[|Z| > \gamma_1/n\epsilon] \le \exp(-\gamma_1)$, and in the Bypass branch we use $\Pr[|Z| > \gamma_2/n\epsilon] \le \exp(-\gamma_2)$.
              Note that we only show convergence for a fixed (and reasonable) value of failure probability $\betaconv$.
              Indeed, we can't ask the Gaussian tail bound to match a Laplace tail bound simultaneously for all possible $\betaconv$.
              Given a failure probability though, we can set the Gaussian standard deviation to match the Laplace's error and ensure convergence.
    \end{enumerate}
\end{proof}

\section{Privacy filters for concurrent interactive mechanisms with adaptively chosen parameters}
\label{sec:appendix:filters}

Existing work on privacy filters \cite{rogers2016privacy, adaptive_rdp} provides bounds on the sequential composition of mechanisms with adaptively chosen parameters, \ie the list of mechanisms and their budget is not known ahead of time and instead chosen by the adversary depending on past requests.
This is not sufficient for \A\ref{alg:pmw-bypass}, that uses interactive mechanisms concurrently and not sequentially. A DP mechanism is interactive if the adversary can submit more than one request, as in the SV protocol.
Recent work on interactive DP \cite{lyu_interactive, vadhan_interactive} showed that we can compose such mechanisms concurrently when they are known ahead of time, \ie the interleaving of interactive DP protocols is also DP.

We now show how to build privacy filters for the concurrent composition of interactive mechanisms with adaptively chosen parameters, a direction left as future work by \cite{lyu_interactive}.
First, in \A\ref{alg:filter_game} we formalize how we extend the adaptive setting introduced in \cite{rogers2016privacy} by allowing the adversary to interact concurrently with multiple long-lived mechanisms, with an interaction protocol similar to \cite{vadhan_interactive}. Second, in \Thm\ref{thm:rdp_filter_adaptive_concurrent} we prove the privacy guarantees enforced by our filter.

\newcommand{\cA}{\mathcal{A}}
\newcommand{\oldgame}{\ensuremath{\mathtt{FixedParamComp}}}
\newcommand{\game}{\ensuremath{\mathtt{AdaptParamComp}}}
\newcommand{\filtgame}{\ensuremath{\mathtt{PrivacyFilterConComp}}}
\newcommand{\simgame}{\ensuremath{\mathtt{SimulatedComp}}}
\newcommand{\simul}{\ensuremath{\mathtt{SIM}}}
\newcommand{\filt}{\ensuremath{\mathtt{F}}}
\newcommand{\neigh}{\sim}
\newcommand{\bbx}{\mathbf{x}}
\newcommand{\cM}{\mathcal{M}}
\newcommand{\compfilt}{\ensuremath{\mathtt{COMP}}}
\newcommand{\halt}{\ensuremath{\mathtt{HALT}}}
\newcommand{\cont}{\ensuremath{\mathtt{CONT}}}
\newcommand{\breakloop}{\ensuremath{\mathtt{ExitWhileLoop}}}

\begin{algorithm}
    \caption{$\filtgame(\cA,k,b, \alpha, \epsilon_G)$}
    \label{alg:filter_game}
    \begin{algorithmic}
        \State Select coin tosses $r_\cA$ for $\cA$ uniformly at random.
        \State $\cM_1, \dots, \cM_k \gets \bot, \dots, \bot$
        \State $\epsilon_1, \dots, \epsilon_k \gets 0, \dots, 0$
        \State $\bbx^0_0, \bbx^1_0, \dots, \bbx^0_k, \bbx^1_k \gets \bot, \dots, \bot$
        \State $j \gets 0$  \ \ \ \ \ \ \graycomment{Current number of mechanisms}
        \State $i \gets 1$  \ \ \ \ \ \ \graycomment{Current message index}
        \State $m_0 \gets \cA(r_\cA)$
        \While{True}

        \Statex \ \ \ \ \ \ \graycomment{Parse the message into a query to a new or existing mechanism}

        \If {$m_{i-1}$ can be parsed as $(j', q')$ where $j' \in [j]$ and $q'$ is a query to $\cM_j$}
        \Statex \ \ \ \ \ \ \ \ \ \ \ \ \graycomment{Continue an existing mechanism}

        \State $j \gets j', q \gets q'$
        \Else{}
        \If{$m_{i-1}$ can be parsed as $(j',\bbx^{0'}, \bbx^{1'}, \cM', \epsilon', q')$ where $j' = j+1$ and $j' \le k$, $\cM'$ is $(\alpha,\epsilon')$-RDP and $\bbx^{0'}, \bbx^{1'}$ are neighboring inputs to $\cM'$}
        \Statex \ \ \ \ \ \ \ \ \ \ \ \ \ \ \ \ \ \ \graycomment{Check the filter before starting a new mechanism}

        \If{$\sum_{i=1}^{n} \epsilon_i + \epsilon' > \epsilon_G$} \label{line:rdp_stopping_rule}
        \State \breakloop  \ \ \ \ \ \ \graycomment{Out of budget}
        \EndIf
        \State $n \gets j +1 , j \gets j +1$
        \State $\bbx^0_j \gets \bbx^{0'}, \bbx^1_j \gets \bbx^{1'}, \cM_j \gets \cM', \epsilon_j \gets \epsilon, q \gets q'$
        \Else{}
        \State \breakloop  \ \ \ \ \ \ \graycomment{Invalid message}

        \EndIf
        \EndIf

        \Statex \ \ \ \ \ \ \graycomment{Execute the query against the mechanism, get the next message}
        \State Extract past interactions with $\cM_j$: $(m_0^j,\ldots,m_{t-1}^j)$
        \State{$\cA$ receives $m_i = \cM_j(\bbx^{j,b}, m_0^j,\ldots,m_{t-1}^j)$ }
        \State{$\cA(m_0, \dots, m_i ; r_\cA)$ sends $m_{i+1}$}
        \State{$i \gets i+2$}

        \EndWhile
        \Return view $V^b = (r_\cA, m_0, \cdots, m_i-1)$
    \end{algorithmic}
\end{algorithm}

\begin{theorem}[Single-order RDP filter]
    \label{thm:rdp_filter_adaptive_concurrent}
    The stopping rule at Line~\ref{line:rdp_stopping_rule} from \A\ref{alg:filter_game} is a valid single-order RDP privacy filter for concurrent composition of interactive mechanisms.
    That is, for all adversaries $\cA, k \in \mathbb{N}, b \in \{0,1\}, \alpha >0, \epsilon_G > 0$, \A\ref{alg:filter_game} defines an $(\alpha, \epsilon_G)$-RDP mechanism,
    \ie if we note $V^b := \filtgame(\cA,k,b,\alpha, \epsilon_G$) we have $D_\alpha(V^b \| V^{1-b}) \le \epsilon_G$ .

\end{theorem}
We recall that in practice $k$ can be taken arbitrarily large, as usual in privacy filter proofs \cite{rogers2016privacy}. If the adversary wants to run less than $k$ mechanisms, they can pass mechanisms that return an empty answer and use $\epsilon(\alpha) = 0$.

\begin{proof}
    Take $\alpha, \epsilon_G > 0$ and $k \in \mathbb{N}$.
    Let's note $\Psi$ the function that parses a view $v$ of the adversary and returns the $k$ mechanisms and their requested budget, ordered by start time: $\Psi(v) = (\cM_1, \epsilon_1, \dots, \cM_k, \epsilon_k)$.
    Note that $\epsilon_2$ can depend on the result of the first interactions with $\cM_1$ and so on, but once a view is fixed we can extract the underlying  privacy parameters.
    We have:

    \begin{align*}
        * & := \exp((\alpha - 1) D_\alpha(V^0 \| V^1))                                                                                               \\
          & = \E_{v \sim V^1} \left[\left(\frac{\Pr[V^0 = v]}{\Pr[V^1 = v]} \right)^\alpha \right]                                                   \\
          & = \E_{(\cM_1, \epsilon_1, \dots, \cM_k, \epsilon_k) \sim \Psi(V^1)}
        \\
          & \left[\E_{v \sim V^1} [(\frac{\Pr[V^0 = v]}{\Pr[V^1 = v]})^\alpha \vert \Psi(v) = (\cM_1, \epsilon_1, \dots, \cM_k, \epsilon_k)] \right]
    \end{align*}

    where the first equality comes from the definition of the R\'enyi divergence and the second from the law of iterated expectations.

    The innermost expectation is conditioned on fixed mechanisms with known privacy parameters, so we can apply the (non-adaptive) RDP concurrent composition theorem from \cite{lyu_interactive}:

    \begin{align}
        * & =  \E_{(\cM_1, \epsilon_1, \dots, \cM_k, \epsilon_k) \sim \Psi(V^1)} \left[  \exp(\epsilon_1 + \dots + \epsilon_k)\right]                      \label{eq:concurrent_composition}
    \end{align}

    Then, since any realization of $V^1$ must respect the RDP filter condition, we must have $\epsilon_1 + \dots + \epsilon_k \le \epsilon_G$.
    Finally, we conclude:

    \begin{align}
        * & \le \E_{(\cM_1, \epsilon_1, \dots, \cM_k, \epsilon_k) \sim \Psi(V^1)} \left[ \exp(\epsilon_G)\right] = \exp(\epsilon_G)
    \end{align}

    This fact also yields a single-order RDP filter~\cite{adaptive_rdp} when the mechanisms are non-interactive.


\end{proof}

\begin{theorem}[Pure DP and filter over the RDP curve]
    \label{thm:rdp_filter_adaptive_concurrent_corollary}
    \noindent
    \begin{enumerate}
        \item If we modify \A\ref{alg:filter_game} to take $\epsilon_G \in\mathbb{R^*_+}^\mathbb{A}$ for a set of RDP orders $A \subseteq \mathbb{R^*_+}$ instead of $(\alpha, \epsilon_G) \in \mathbb{R^*_+}^2$, and if we replace the condition in Line~\ref{line:rdp_stopping_rule} by $\forall \alpha \in A, \sum_{i=1}^{j} \epsilon_i(\alpha) + \epsilon_i'(\alpha) > \epsilon_G(\alpha)$, then \A\ref{alg:filter_game} verifies $\exists \alpha > 0 : D_\alpha(V^b \| V^{1-b}) \le \epsilon_G(\alpha)$.
        \item If we take $\alpha \to +\infty$ we obtain a filter for pure differential privacy.
    \end{enumerate}
\end{theorem}

\begin{proof}
    \noindent
    \begin{enumerate}
        \item The proof is the same as \cite{adaptive_rdp}.
        \item The R\'enyi difference can be extended by continuity to $\alpha = +\infty$, which corresponds to pure differential privacy \cite{8049725}. In that case, the additive composition rule for RDP becomes the basic composition theorem for pure DP.
    \end{enumerate}
\end{proof}


\section{Laplace Histogram baseline}
\label{sec:appendix:laplace_baseline}

We can consider another simple baseline, {\em Laplace Histogram}, that works as follows. First, we compute a noisy estimate for every single bin in $\mathcal{X}$. Then, we can answer any linear query by taking the combination of the noisy estimates. Consider a static, non-partitioned, dataset with $n$ datapoints for domain $\mathcal{X}$. Suppose that we want to answer linear queries with absolute error $\alpha$ with probability $1-\beta$. We are using Pure DP for simplicity.
\begin{itemize}
    \item {\em Direct Laplace.} If we answer each query separately like in the Laplace baseline of \F \ref{fig:pmw-laplace-demo}, a Laplace tail bound shows that we can take $\epsilon_{Direct} = \frac{\ln(1/\beta)}{\alpha n}$. Each query has to pay this cost.
    \item {\em Laplace Histogram.} Instead, we can use a single multidimensional Laplace query to get a noisy estimate of the count for every bin in the histogram. The L1 sensitivity of the histogram is $2$, so we can pay $\epsilon_{Histogram}$ to get $\tilde h(v) := h(v) + X_v$ with $X_v \sim \lap{2/\epsilon_{Histogram}}$ for all $v\in |\mathcal{X}|$. To answer a linear query $q$, we compute $\tilde q := \frac{1}{n} \sum_{v \in |\mathcal{X}| } q_v \tilde h(x)$ where $q_v \in [0,1]$. By post-processing, we can pay $\epsilon_{Histogram}$ only once to answer as many queries $\tilde q$ as we want. To have error below $\alpha$, we need $|\frac{1}{n} \sum_{v \in |\mathcal{X}| } q_v X_v| < \alpha$. With Chebyshev's inequality, this happens with probability $\frac{ \sum_{v \in |\mathcal{X}| } q_v^2 \mathbb{V}[X_v]}{n^2 \alpha^2} \le \frac{|\mathcal{X}| \cdot 2 \cdot 2^2}{n^2\alpha^2 \cdot \epsilon_{Histogram}^2}$. This gives us the desired $(\alpha,\beta)$ accuracy bound when $\epsilon_{Histogram} = \frac{2\sqrt{2 |\mathcal{X}| / \beta}}{n \alpha}$.
\end{itemize}

In \F \ref{fig:pmw-laplace-demo} we took $|\mathcal{X}| = 128$ and $\beta = 10^{-3}$, which gives $\frac{2\sqrt{2 |\mathcal{X}| / \beta}}{\ln(1/\beta)} \simeq 146$. That means that after 146 queries it is more advantageous to use the Laplace Histogram rather than Direct Laplace.
However, for a larger domain such as \citibike ($|\mathcal{X}| = 604, 800$), the same calculation shows that we need more than 10,069 queries for Laplace Histogram to outperform Direct Laplace. For this number of queries, \sysname is already close to convergence using much less budget than Exact-Cache (\F \ref{fig:monoblock_citibike_system_eval_zipf0}), itself an improvement over Direct Laplace.
Partioning the dataset (e.g. across 50 timestamps) has the same effect as increasing $|\mathcal{X}|$.

Finally, these sketches used basic composition, which is suboptimal for Direct Laplace: using advanced composition would make Direct Laplace more competitive, as the privacy budget grows only in the square root of the number of queries (instead of linearly).


\end{document}